\documentclass[twoside,leqno]{article}

\usepackage[letterpaper]{geometry}

\usepackage{ltexpprt}
\usepackage{hyperref}
\usepackage{ifthen}
\usepackage[dvips]{color}
\usepackage{xcolor}
\definecolor{DarkRed}{RGB}{150,0,0}
\definecolor{DarkGreen}{RGB}{0,150,0}
\definecolor{DarkBlue}{RGB}{0,0,150}
\definecolor{purple}{RGB}{200,0,200}

\newtheorem{definition}{Definition}[section]

\usepackage{tikz}
\usetikzlibrary{shapes,positioning, arrows.meta}

\usepackage{beton}
\usepackage[T1]{fontenc}

\usepackage{mathptmx}

\usepackage{bm}

\usepackage{amsmath,amssymb}
\usepackage{mathtools}
\usepackage{thm-restate}

\def\paritykOV{\mathsf{parity}\text{-}k\text{-}OV}
\def\paritykSUM{\mathsf{parity}\text{-}k\text{-}SUM}
\def\paritykXOR{\mathsf{parity}\text{-}k\text{-}XOR}

{\makeatletter
 \gdef\xxxmark{%
   \expandafter\ifx\csname @mpargs\endcsname\relax % in minipage?
     \expandafter\ifx\csname @captype\endcsname\relax % in figure/caption?
     \marginpar{\textcolor{red}{\textbf{xxx}}}% not in a caption or minipage, can use marginpar
          \else
       \textbf{\textcolor{red}{xxx}} % notice trailing space
    \fi
   \else
     \textbf{\textcolor{red}{xxx}} % notice trailing space
   \fi}
 \gdef\xxx{\@ifnextchar[\xxx@lab\xxx@nolab}
 \long\gdef\xxx@lab[#1]#2{\textbf{[\xxxmark \textcolor{{blue}#2} ---{\sc {\color{blue}#1}}]}}
 \long\gdef\xxx@nolab#1{\textbf{[\xxxmark \textcolor{blue}{#1}]}}
}

\newcommand{\tO}{\tilde{O}}

\newcommand{\eps}{\epsilon}

\newtheorem{assumption}{Assumption}

\def \Z {\mathbb Z}
\def \R {\mathbb R}

\newcommand{\kOV}{$k$\mbox{-}OV}

\newcommand{\klcs}[1][]{\ifthenelse{\equal{#1}{}}{$k$-LCS}{${#1}$-LCS}}

\newcommand{\kwlcs}[1][]{\ifthenelse{\equal{#1}{}}{$k$-WLCS}{${#1}$-WLCS}}

\newcommand{\kNLstC}[1][]{\ifthenelse{\equal{#1}{}}{$k$-NLstC}{${#1}$-NLstC}}

\newcommand{\kELstC}[1][]{\ifthenelse{\equal{#1}{}}{$k$-ELstC}{${#1}$-ELstC}}

\newcommand{\czkc}[1][]{\ifthenelse{\equal{#1}{}}{FZ$k$C}{FZ${#1}$C}}

\newcommand{\czkch}[1][]{\ifthenelse{\equal{#1}{}}{FZ$k$CH}{FZ${#1}$CH}}

\newcommand{\cfkc}[1][]{\ifthenelse{\equal{#1}{}}{F$\mathfrak{f}k$C}{F$\mathfrak{f}{#1}$C}}

\newcommand{\cfkch}[1][]{\ifthenelse{\equal{#1}{}}{F$\mathfrak{f}k$CH}{F$\mathfrak{f}{#1}$CH}}

\newcommand{\ckov}[1][]{\ifthenelse{\equal{#1}{}}{F$k$-OV}{F${#1}$-OV}}

\newcommand{\ckovh}[1][]{\ifthenelse{\equal{#1}{}}{F$k$-OVH}{F${#1}$-OVH}}

\newcommand{\cksum}[1][]{\ifthenelse{\equal{#1}{}}{F$k$-SUM}{F${#1}$-SUM}}

\newcommand{\cksumh}[1][]{\ifthenelse{\equal{#1}{}}{F$k$-SUMH}{F${#1}$-SUMH}}

\newcommand{\fzkc}[1][]{\ifthenelse{\equal{#1}{}}{FZ$k$C}{FZ${#1}$C}}

\newcommand{\ckxor}[1][]{\ifthenelse{\equal{#1}{}}{F$k$-XOR}{F${#1}$-XOR}}

\newcommand{\ckxorh}[1][]{\ifthenelse{\equal{#1}{}}{F$k$-XORH}{F${#1}$-XORH}}

\newcommand{\ckfunc}[1][]{\ifthenelse{\equal{#1}{}}{F$k$-$\mathfrak{f}$}{F${#1}$-$\mathfrak{f}$}}

\newcommand{\ckfunch}[1][]{\ifthenelse{\equal{#1}{}}{F$k$-$\mathfrak{f}$H}{F${#1}$-$\mathfrak{f}$H}}

\DeclarePairedDelimiter{\ceil}{\lceil}{\rceil}
\DeclarePairedDelimiter{\floor}{\lfloor}{\rfloor}

\newcommand{\goodPoly}{good low-degree polynomial}

\newcommand{\kxor}[1][]{\ifthenelse{\equal{#1}{}}{$k$-XOR}{${#1}$-XOR}}

\newcommand{\ksum}[1][]{\ifthenelse{\equal{#1}{}}{$k$-SUM}{${#1}$-SUM}}

\newcommand{\vksum}[1][]{\ifthenelse{\equal{#1}{}}{V$k$-SUM}{V${#1}$-SUM}}

\newcommand{\ukov}[1][]{\ifthenelse{\equal{#1}{}}{U$k$-OV}{U${#1}$-OV}}
\newcommand{\nukov}[1][]{\ifthenelse{\equal{#1}{}}{NU$k$-OV}{NU${#1}$-OV}}
\newcommand{\ukxor}[1][]{\ifthenelse{\equal{#1}{}}{U$k$-XOR}{U${#1}$-XOR}}

\newcommand{\GoodDPolys}[1][]{\ifthenelse{\equal{#1}{}}{Good $d$-Degree Polynomials}{Good ${#1}$-Degree Polynomials}}
\newcommand{\goodDPoly}[1][]{\ifthenelse{\equal{#1}{}}{good $d$-degree polynomial}{good ${#1}$-degree polynomial}}

\newcommand{\OKPolys}[1][]{\ifthenelse{\equal{#1}{}}{Fine $d$-Degree Polynomials}{Fine ${#1}$-Degree Polynomials}}
\newcommand{\okPoly}[1][]{\ifthenelse{\equal{#1}{}}{fine $d$-degree polynomial}{fine ${#1}$-degree polynomial}}

\newcommand{\bikmath}{\binom{k}{2}}

\newcommand\relatedversion{}

\begin{document}

\title{\Large Average-Case Hardness of Parity Problems: \\ Orthogonal Vectors, k-SUM and More\relatedversion}

\author{Mina Dalirrooyfard\thanks{Morgan Stanley, \texttt{minad@mit.edu}}, Andrea Lincoln\thanks{Boston University, \texttt{andrea2@bu.edu}}, Barna Saha\thanks{University of California San Diego, \texttt{bsaha@ucsd.edu}. Supported by NSF grants 1652303, 1909046, 2112533, and HDR TRIPODS Phase II grant 2217058. }, Virginia Vassilevska Williams \thanks{Massachusetts Institute of Technology, \texttt{virgi@mit.edu}. Supported by NSF Grant CCF-2330048, BSF Grant 2020356 and a Simons Investigator Award.}}
%\date{June 2019}
\date{}
%\begin{document}

%\maketitle
%\date{}

\maketitle

% Copyright Statement
% When submitting your final paper to a SIAM proceedings, it is requested that you include
% the appropriate copyright in the footer of the paper.  The copyright added should be
% consistent with the copyright selected on the copyright form submitted with the paper.
% Please note that "20XX" should be changed to the year of the meeting.

% Default Copyright Statement
% \fancyfoot[R]{\scriptsize{Copyright \textcopyright\ 2025 by SIAM\\
% Unauthorized reproduction of this article is prohibited}}

%\thispagestyle{empty}
\begin{abstract}
\small\baselineskip=9pt

This work establishes conditional lower bounds for average-case {\em parity}-counting versions of the problems $k$-XOR, $k$-SUM, and $k$-OV. 
The main contribution is a set of self-reductions for the problems, providing the first specific distributions, for which:

\begin{itemize}
    \item $\paritykOV$ is $n^{\Omega(\sqrt{k})}$ average-case hard, under the $k$-OV hypothesis (and hence under SETH),
    \item $\paritykSUM$  is $n^{\Omega(\sqrt{k})}$  average-case hard, under the $k$-SUM hypothesis, and
    \item $\paritykXOR$  is $n^{\Omega(\sqrt{k})}$  average-case hard, under the $k$-XOR hypothesis.
\end{itemize}

Under the very believable hypothesis that at least one of the $k$-OV, $k$-SUM, $k$-XOR or $k$-Clique hypotheses is true, we show that parity-$k$-XOR, parity-$k$-SUM, and parity-$k$-OV all require at least $n^{\Omega(k^{1/3})}$ (and sometimes even more) time on average (for specific distributions).

To achieve these results, we present a novel and improved framework for worst-case to average-case fine-grained reductions, building on the work of Dalirooyfard, Lincoln, and Vassilevska Williams, FOCS 2020.

\end{abstract}

%\thispagestyle{empty}
%\newpage

%\ifsubmission
%\else
%\pagenumbering{roman}
%\tableofcontents
%\newpage
%\pagenumbering{arabic}
%\fi 

\section{Introduction}
\label{sec:intro}

During the last few years there has been significant progress on the theoretical foundations of average-case fine-grained complexity. 
This area utilizes fine-grained worst-case to average-case reductions to provide lower bounds conditioned on
popular hardness 
hypotheses from fine-grained complexity,  for key computational problems over natural input distributions
\cite{BallRSV18,Goldreich20,UniformCliqueABB,factoredProblems,Goldreich20,arxivHS20}. The major progress has so far been mostly restricted to subgraph counting problems in graphs \cite{BallRSV18,Goldreich20}, satisfiability \cite{CHV23} and the so-called {\em factored} problems \cite{factoredProblems}. However, no average-case hardness has been proven so far for the core problems of fine-grained complexity like $k$-SUM, $k$-OV, and $k$-XOR, even for their counting versions. 

A first step was made by \cite{factoredProblems} who showed average-case hardness for certain ``factored'' counting versions of $k$-SUM, $k$-OV and $k$-XOR. However, these versions
 are much more expressive than their non-factored counterparts, and less natural. Ideally, we would like to get average-case hardness for the traditional detection versions of $k$-SUM, $k$-OV, and $k$-XOR, as these problems are central in fine-grained complexity. Yet, even for the counting versions of these problems it is completely unclear how to obtain hardness from the hardness of their ``factored'' cousins.

In this paper, we take another step forward. First, we present new average-case lower bounds for the {\em counting} versions of $k$-SUM, $k$-OV, and $k$-XOR under any of the following traditional fine-grained complexity hypotheses: the worst-case random ETH, the $k$-XOR hypothesis, the $k$-SUM hypothesis, or the $k$-Clique hypothesis. Second, we show that our lower bounds also hold for the {\em parity} versions of the problems where we only need to return the parity of the count (i.e., returning a single bit). The hardness of the parity versions can naturally be seen as an intermediate step towards the ultimate goal of resolving the average case complexity of the detection problems. 

Previously, such average-case hardness for the parity version has been shown to hold for the $k$-clique problem \cite{UniformCliqueABB}. However, it is not clear how to apply the techniques in \cite{UniformCliqueABB}  to the parity versions of $k$-SUM, $k$-OV, and $k$-XOR.

In this paper, we build upon the framework of Dalirrooyfard, Lincoln, and Vassilevska Williams \cite{factoredProblems} who introduced factored versions of problems and proved average-case hardness results for them. We show how to reduce factored problems to their non-factored counterparts. By doing so, we can not only handle the counting versions of the more natural problems, but also their parity versions. In particular, we present easy-to-sample distributions on which the average-case parity versions of $k$-OV ($\paritykOV$), $k$-SUM ($\paritykSUM$),  and $k$-XOR ($\paritykXOR$) problems are hard. 

The definitions of the problems are as follows. The $\paritykSUM$ problem takes $n$ numbers in the range $[-n^k, n^k]$ as input and asks for the parity of the number of $k$-tuples of input numbers that sum to zero. The $\paritykOV$ problem takes $n$ Boolean vectors as input and asks for the parity of the number of $k$-tuples of vectors whose bitwise product is the all-zeros vector\footnote{Equivalently, this is the number of $k$-tuples with generalized inner product $0$.}. The $\paritykXOR$ problem takes $n$ vectors as input and asks for the parity of the number of $k$-tuples of vectors whose bitwise XOR is the all-zeros vector. The best known algorithms for $\paritykSUM$ on $d$ bit numbers, $\paritykOV$ on $d$ bit vectors, and $\paritykXOR$ on $d$ bit vectors in the worst-case run in $dn^{\Theta(k)}$ time. 
%(a.k.a. $\oplus k$-SUM)
%(a.k.a. $\oplus k$-XOR)
%(a.k.a. $\oplus k$-OV)
%In this paper, we present average-case distributions on which we demonstrate hardness for  $\paritykSUM$, $\paritykOV$, and $\paritykXOR$. 

The hardness for all the above problems are shown from the very believable hypothesis that at least one of the $k$-XOR, $k$-SUM, $k$-Clique, $k$-OV, or random ETH hypotheses are true. Moreover, we provide random {\em self-reductions}, i.e., we derive average-case lower bounds for each problem from the corresponding worst-case hypothesis. 
%{\color{blue} 
Self reductions themselves are very interesting as they are often the first step towards showing the tight average-case hardness of the underlying problems over the distributions on which they are conjectured to be hard. 

%Our average-case hardness results in fact apply to all the results of \cite{factoredProblems}. Moreover, we also show parity hardness for any arbitrary $k$-node pattern $H$ in a graph which is a natural extension to the parity $k$-clique problem considered by Boix, Brennman, and Bresler \cite{UniformCliqueABB}.

Further, extending the average-case hardness of the parity counting version of $k$-clique considered by Boix{-}Adser{\`{a}}, Brennan and Bresler \cite{UniformCliqueABB}, we obtain fine-grained average-case hardness results for counting the number of subgraphs $H$ in a graph $G$, modulo $2$, for any $k$-node pattern $H$.
 
 In the process, we also simplify the framework of \cite{factoredProblems} which we hope will lead to a wider adaptation of the framework. 
 %As an example, we show how the simplified framework applies to showing average case hardness of computing parity of the number of copies of any $k$-node subgraph $H$, and not just $k$-clique \cite{UniformCliqueABB}. 
 %By doing so, we not only get hardness for the parity versions of $k$-OV, $k$-SUM, and $k$-XOR, we also generalize the work of \cite{UniformCliqueABB} to consider any $k$-node subgraph $H$ and the hardness of computing parity of the number of copies of $H$. Our improved framework is also more efficient (The overhead is reduced by a factor of approximately $\lg\lg(n)$ in the exponent).
%To discuss these results, we begin by defining the problems of interest. We provide an informal definition here and refer the reader to Section \ref{sec:prelims} for formal definitions. 
%In this paper, In our examination of $\paritykOV$, $\paritykSUM$, and $\paritykXOR$, %}
%{\color{red} State why self reductions are important.}
%While our lower bounds are not tight, 
%

\begin{theorem}[Informal]
  Assuming that \textbf{\emph{at least one}} of the worst-case $k$-OV, $k$-SUM, or $k$-XOR hypotheses holds, \textbf{\emph{all}} of $\paritykOV$, $\paritykSUM$ and $\paritykXOR$ are $n^{\Omega(\sqrt{k})}$-hard on average for a natural distribution.
\end{theorem}

\begin{theorem}[Informal]
  Assuming that \textbf{\emph{at least one}} of the randomized exponential time hypotheses (rETH), $k$-OV, $k$-SUM, $k$-clique or $k$-XOR hypotheses hold, \textbf{\emph{all}} of the counting problems $\paritykOV$, $\paritykSUM$, and $\paritykXOR$ are $n^{\Omega(k^{1/3})}$-hard on average for a natural distribution \footnote{In fact, a stronger statement is true: If \textbf{\textit{any}} of the problems $k$-OV, $k$-SUM, $k$-clique or $k$-XOR require $n^{\Omega(k)}$ time in the worst case then \textit{\textbf{all}} of the counting problems $\paritykOV$, $\paritykSUM$, and $\paritykXOR$ are $n^{\Omega(k^{1/3})}$-hard on average.}.
\end{theorem}
To the best of our knowledge, these are the first average-case super-linear lower bounds for any of these problems.

\subsection{Our Results}
Our results on average-case hardness of $\paritykOV$, $\paritykSUM$, and $\paritykXOR$ are summarized in Table \ref{table:average_case_hardness}. We give  hardness results that are implied by the random exponential time hypothesis (rETH), the $k$-XOR hypothesis, and the $k$-SUM hypothesis, formally defined in Section \ref{sec:prelims}.
%The rows are problems solved in the average-case. The columns are worst-case hardness hypotheses.
%We start with by explaining the implications of our improvements to the previous framework for average-case hardness of the problems we consider. We next explain our simplified framework for worst-case to average-case reductions from a problem to its ``factored" version, improving upon \cite{factoredProblems}. We finally present reductions from factored versions of fine-grained problems back to their original versions (this incidentally also allows for reductions \emph{between} all the fine-grained problems). We combine the improved framework and reductions from factored to un-factored problems to get our main results.

\begin{table}[h!]
\centering
\begin{tabular}{|c|c|c|c|c|}
\hline
\textbf{} & \textbf{rETH} & \textbf{$k'$-XOR Hypothesis} & \textbf{$k'$-SUM Hypothesis} & \textbf{$k'$-Clique Hypothesis} \\
\hline
$\mathsf{parity}\text{-}K\text{-}OV$ & $N^{\Omega{(\sqrt{K})}}$ [\ref{thm:ovxorhard}]& $N^{K^{1/3}/4-o(1)}$ [\ref{thm:ovxorhard}]& $N^{K^{1/3}/4-o(1)}$ [\ref{thm:ovxorhard}]& $N^{(\sqrt{2K}+1)\omega/6-o(1)}$ [\ref{thm:cliqueToOthers}] \\
\hline
$\mathsf{parity}\text{-}K\text{-}SUM$ & $N^{\Omega{(\sqrt{\frac{K}{\lg{K}}})}}$ [\ref{thm:sumhard}] & $N^{\sqrt{\frac{K}{8\lg{K}}}-o(1)}$ [\ref{thm:sumhard}] & $N^{\ceil{\sqrt{K}/2}/2-o(1)}$  [\ref{thm:sumhard}]& $N^{(\sqrt{2K}+1)\omega/6-o(1)}$ [\ref{thm:cliqueToOthers}] \\
\hline
$\mathsf{parity}\text{-}K\text{-}XOR$ & $N^{\Omega{(K^{1/3})}}$[\ref{thm:ovxorhard}] & $N^{\ceil{\sqrt{K}/2}/2-o(1)}$[\ref{thm:ovxorhard}] & $N^{\ceil{K^{1/3}/2}/2-o(1)}$[\ref{thm:ovxorhard}] & $N^{(\sqrt{2K}+1)\omega/6-o(1)}$ [\ref{thm:cliqueToOthers}]  \\
\hline
\end{tabular}
\caption{Average-case hardness of $\paritykOV$, $\paritykSUM$, and $\paritykXOR$ from multiple well-known hypotheses. The rows correspond to problems that are solved in the average-case. The columns are worst-case hardness hypotheses. Note that we use $k'$ to indicate that the hardness comes from a $k' \ne K$ for these problems (a smaller $k'$ is used to prove the hardness for a larger $K$). %For more details you can read the linked theorems in brackets.
}
\label{table:average_case_hardness}
\end{table}

%\paragraph{Hardness results for average-case parity $k$-OV, $k$-SUM and $k$-XOR}

Formal statements of the theorems are given below. 
%In the following theorems we present the formal statements of reductions to $\paritykOV$, $\paritykSUM$, and $\paritykXOR$ implying their respective hardness. 
%We give  hardness results that are implied by the random exponential time hypothesis (rETH), the $k$-XOR hypothesis, and the $k$-SUM hypothesis, formally defined in section \ref{sec:prelims}. %In each statement we give the name of the distribution we prove hardness for. We also state how hard we show each problem to be from each hypothesis. 
The key take-away is that if the size of the problem is $N$, we can show $N^{\Omega(\sqrt{K})}$ average-case hardness for $\paritykOV$, $\paritykSUM$, and $\paritykXOR$ assuming the worst-case hardness of $k$-$OV$, $k$-$SUM$, and $k$-$XOR$ respectively.  Furthermore, we show $N^{\Omega(K^{1/3})}$ average-case hardness for  $\paritykOV$ and $\paritykXOR$ from any of the hypotheses, and we show $N^{\Omega(
\sqrt{K/\lg(K)})}$ average-case hardness for  $\paritykSUM$ from any of the hypotheses. 

We generate hardness over distributions that are not the uniform distribution, but they are easy to sample \footnote{This is quite common in average-case fine-grained complexity, for example see \cite{GoldreichR18}}. For simplicity, we do not provide the details of the distributions in the theorem statements below and refer the readers to the respective sections where the theorem proofs are provided for further details. 
%In these theorem statements we avoid discussing the distribution at hand for simplicity, these distributions depend on the problem and are gadgety constructions that are filled with iid uniformly random bits. One can think of the distributions as structures where details are filled in by truly uniformly random inputs, hence they are easy to sample distributions and close to the uniform distribution\footnote{This is not uncommon in average-case fine-grained complexity, for example see \cite{GoldreichR18}}. %\mina{add a word about using weird distributions in other works, like the Goldreich Rothblum work of counting t-cliques}%\mina{add a word about the weird looking distributions, or replace them in the theorem statements with words as an informal version of the theorem.}

%\mina{put them in a table} \andrea{Done}
%{\color{red}Barna: expected the table to be here. We should not repeat the results many times.} \andrea{moved the table.}

%\paragraph{\mina{suggestion for theorems 1.3 and 1.4}} \mina{Since the distributions are meaningless at this point, we may as well present them in the most simple way.}
\begin{restatable}{theorem}{ovxorhard}\label{thm:ovxorhard}
Let $K$ be a constant. Let $P\in \{\mathsf{parity}\text{-}K\text{-}OV, \mathsf{parity}\text{-}K\text{-}XOR\}$. There are easy to sample distributions $D_1^P(N,K),D_2^P(N,K)$ and $D_3^P(N,K)$ such that any algorithm that solves $P$ of size $N$ with vectors of dimension $\Theta(K\lg{N})$ with probability $1-\frac{1}{\Theta(2^K)}$ requires at least:
\begin{itemize}
    \item $N^{\Omega{(\sqrt{K})}}$ time assuming rETH, if the input is drawn from $D_1^P(N,K)$. 
    \item $N^{K^{1/3}/4-o(1)}$ time assuming the $\sqrt{K}$-XOR hypothesis, if the input is drawn from $D_2^P(N,K)$. 
     \item $N^{K^{1/3}/4-o(1)}$ time assuming the $K^{1/3}$-SUM hypothesis, if the input is drawn from $D_3^P(N,K)$. 
\end{itemize}
\end{restatable}

\begin{restatable}{theorem}{sumhard}\label{thm:sumhard}
Let $K$ be a constant. There are easy to sample distributions $D_1(N,K),D_2(N,K)$ and $D_3(N,K)$ such that any algorithm that solves $\mathsf{parity}\text{-}K\text{-}SUM$ of size $N$ with vectors of dimension $\Theta(K\lg{N})$ with probability $1-\frac{1}{\Theta(2^K)}$ requires at least:
\begin{itemize}
    \item $N^{\Omega{(\sqrt{\frac{K}{\lg{K}}})}}$ time assuming rETH, if the input is drawn from $D_1(N,K)$. %$D_{SUM}(N,K,O(\lg{N}),\Theta(K\lg{K}))$. 
    \item $N^{\sqrt{\frac{K}{8\lg{K}}}-o(1)}$ time assuming $\sqrt{\frac{K}{\lg{K}}}$-XOR hypothesis, if the input is drawn from $D_2(N,K)$. %$D_{SUM}(N,K,0.5\lg\frac{4N}{\sqrt{K}},\Theta(K\lg{K}))$. 
     \item $N^{\ceil{\sqrt{K}/2}/2-o(1)}$ time assuming $\sqrt{K}$-SUM hypothesis, if the input is drawn from $D_3(N,K)$.  %$D_{SUM}(N,K,0.5\lg\frac{4N}{\sqrt{K}},\sqrt{K})$. 
\end{itemize}
\end{restatable}

We also get hardness for average-case $\mathsf{parity}\text{-}K$-OV and average-case $\mathsf{parity}\text{-}K$-XOR from the $k$-clique hypothesis.

\begin{restatable}{theorem}{cliqueToXORSUMOV}\label{thm:cliqueToOthers}
Let $P\in \{\mathsf{parity}\text{-}K\text{-XOR}, \mathsf{parity}\text{-}K\text{-OV}, \mathsf{parity}\text{-}K\text{-SUM}\}$. Let the input size of $P$ be $N$. If the $k$-clique  hypothesis, where $K=\bikmath$, is true then there is an explicit distribution $D_P(N,K)$ on the input of $P$ where $P$ is $N^{(\sqrt{2K}+1)\omega/6-o(1)}$ average-case hard,
%
% If the $k$-clique hypothesis is true then, then there are explicit distributions over inputs of $K$-XOR, $K$-OV $K$-SUM
% \begin{itemize}
%     \item There is an explicit distribution over inputs on which $K$-XOR over $N$ vectors is $N^{(\sqrt{2K}+1)\omega/6-o(1)}$ average-case hard.
%     \item There is an explicit distribution over inputs on which $K$-OV over $N$ vectors is $N^{(\sqrt{2K}+1)\omega/6-o(1)}$ average-case hard.
%     \item There is an explicit distribution over inputs on which $K$-SUM over $N$ numbers is $N^{(\sqrt{2K}+1)\omega/6-o(1)}$ average-case hard.
% \end{itemize}
where $\omega$ is the exponent of matrix multiplication.
\end{restatable}
%We give hardness reductions from worst-case $k$-clique to average-case hardness for $k$-XOR and $k$-OV (over non-uniform distributions).  
%Our results answer the open question posed in
%\cite{afargholiV16} by Jafargholi and Viola in their appendix B, where they show how to reduce $4$-clique to $6$-SUM over the group $\mathbb{Z}^t_3$. 
%Our theorem below is the generalization of this result and reduces  average-case parity $k$-clique to average-case $\mathsf{parity}\text{-} \bikmath$-OV, $\mathsf{parity}\text{-} \bikmath$-XOR, and  $\mathsf{parity}\text{-} \bikmath$-SUM in general (over non-uniform distributions). Thus, we answer 

Along the way we answer a question raised by Jafargholi and Viola (see Appendix B \cite{afargholiV16}) where they show how to reduce $4$-clique to $6$-SUM over the group $\mathbb{Z}^t_3$, but left it open how to carry the reduction over $\mathbb{Z}^t_2$ or $\mathbb{Z}$. In Theorem \ref{thm:cliqueToOthers} using Theorem \ref{thm:framework-general} (where the $\mathbb{Z}^t_2$ appears) we resolve the question by having a framework which works over this field. 

\subsection{Technical Overview} 
In this section, we give a high level technical overview of our results. We start with the simplification of the framework of Dalirrooyfard, Lincoln and Vassilevska Williams \cite{factoredProblems}.
\paragraph*{Simplification of the worst-case to average-case reduction framework.} 
Dalirrooyfard, Lincoln and Vassilevska Williams \cite{factoredProblems} developed a framework for deriving worst-case to average-case reductions for a given problem $P$ to its factored version provided that $P$ can be represented as a `good low degree polynomial' $f$ over $\mathbb{F}_p$ for some prime $p$, i.e. for every $x$, $f(x)=P(x)\mod p$.

% and finite field $\mathbb{F_p}$.
%has a uniform average-case hard distribution. 
The `good low degree polynomials' are those polynomials that have a low degree $d$ and their input is partitioned into $d$ sections such that each monomial has \emph{exactly} one variable from each section. The later property is referred to as `strongly $d$-partititeness' in \cite{factoredProblems}.
%be `strongly $d$-partite'. The term $d$-partitite refers to which means that the input of the polynomial is partitioned into $d$ sections, and each monomial has \emph{exactly} one variable from each section. 
Via \cite{factoredProblems} we have that for problems that have a good low degree polynomial representation, there is a worst-case to average-case fine-grained reduction for a natural easy-to-sample-from distribution.

% Our first result depicted in Theorem \ref{thm:framework-general} and Theorem \ref{thm:framework-parity} is the improvement of the worst-case to average-case reduction framework of \cite{factoredProblems}. In their framework a problem which has a `good low degree polynomial' over $\mathbb{F}_p$ for some prime $p$ has a uniform average-case hard distribution. The `good low degree polynomials' of \cite{factoredProblems} must have a low degree $d$ and be `strongly $d$-partite', which limits the kinds of monomials the polynomial is allowed to have. 
Our first contribution is a simplification of the framework of \cite{factoredProblems}, which strengthens the framework and makes it more efficient. %\mina{any concrete examples for this claim?}. 
We relax the framework of \cite{factoredProblems} as follows. 

First, instead of representing a problem via a polynomial {\em modulo a prime}, we consider polynomials over the integers. We then relax the notion of strongly $d$-partiteness by allowing our polynomials to be just $d$-partite.
This means that the input variables to the polynomial are still partitioned in $d$ sections, however each monomial has \emph{at most} one (as opposed to exactly one) variable from each section. We denote such polynomials as {\em fine $d$-degree polynomials}. This relaxation of the $d$-partiteness property allows more freedom in the framework.

%For the problems where the output is binary, %In our framework the output of the problem is taken mod $p$ similar to \cite{factoredProblems}. If $p=2$ 
%we only need the underlying problem to be expressible as a $d$-degree polynomial for some relatively small $d$ modulo $2$. For non-binary problems, we need this polynomial to be $d$-partite as opposed to strongly $d$-partite, which means that the input of the polynomial is partitioned into $d$ sections, and each monomial has \emph{at most} one (as opposed to exactly one) variable from each section. We denote such polynomials as {\em fine $d$-degree polynomials}. Note that we require the exact value of the polynomial to be equal to the output of the problem. This relaxation from good low degree polynomials along with ideas we mention below help simplify the proofs for the framework substantially. %Moreover, removing these restrictions from the polynomials allows our framework to be applicable to a much wider variety of problems.
%{\color{red}WRITE SOMETHING ON HOW THIS FLEXIBILITY HELPS. For example, if not new problem, does it help in simplifying the proofs}\mina{done}%In both reductions, we first reduce a worst-case instance to $2^{O(d)}$ average case instances, in different ways.

Note that removing restrictions on the framework expands the set of problems that it can be applied to (handling parity, functions with a larger degree, and functions that aren't ``strongly $d$-partite''). We hope the proof simplification, and reducing restrictions pave the road for future broader adaptation of the framework to understand fine-grained average case complexity of new problems. %See Section \ref{sec:frameworkBetter} for more details.

We now describe the new framework.
The degree $d$ of the polynomial describing the problem $P$ appears in the \emph{success probability} that we need for an average-case algorithm for $P$ required by the reduction. The success probability of an average-case $T(n)$ time algorithm is the probability that the algorithm gives the right answer in at most $T(n)$ time steps.  In the following statement of our framework, we need success probability of around $1-1/\lg^d{n}$.

\begin{restatable}{theorem}{frameworkgeneral}\label{thm:framework-general}
Let $P$ be a problem that takes an input $I \in \{0,1\}^n$ and has an output over the integers in $[-M,M]$ for some integer $M$.
Additionally, assume that a fine $d$-degree polynomial $f$ exists such that $P(I)=f(I)$ for all $I \in \{0,1\}^n$.  
Let $A$ be an average-case algorithm that runs in time $T(n)$ such that when $\vec{v}$ is sampled uniformly from $\mathbb{Z}_2^n$, then:
$$Pr[A(\vec{v}) = P(\vec{v})] \geq 1-\frac{1}{2^{d+2}(d+\log_2M)^d}.$$
Then there is a randomized algorithm $B$ that runs in time $O \left((2d+2\log_2M)^d(n+T(n)) \right)$ such that
for \emph{any} vector $\vec{v} \in \{0,1\}^n$:
$$Pr[B(\vec{v}) = P(\vec{v})] \geq 3/4.$$
\label{thm:BetterFrameworkLargeFeild}
\end{restatable}

%The above framework applies to counting problems modulo any number $p$. For $p=2$, we show that there is a worst-case to average-case reduction for a problem $P$ if the output of $P$ can be written as a polynomial of degree $d$, \textit{without any other constraints}.  In the general framework we need an average-case algorithm with success probability of around $1-1/\lg^d{n}$, whereas for $p=2$, we only need a success probability of $1-1/2^d$. The following Theorem captures the framework for $p=2$.

We further strengthen Theorem \ref{thm:framework-general} for problems with binary output, so that the success probability needed is $1-1/2^{d+3}$ (see Theorem \ref{thm:framework-parity}). 
%{\color{red} Barna: I vote the remove the below theorem and just say that when the output of $P(I) \in \mathbb{Z}_2$, we can further strengthen the above theorem.}

% \begin{restatable}{theorem}{frameworkparity}
% \label{thm:framework-parity}
% Let $P$ be a problem that takes an input $I\in\{0,1\}^n$ and has a binary output, i.e. $P(I)\in \mathbb{Z}_2$.
% Additionally, assume that a $d$-degree polynomial $f$ exists such that $P(I)=f(I) \pmod{2}$.  
% Let $A$ be an average-case algorithm that runs in time $T(n)$ such that when $\vec{v}$ is sampled uniformly from $\mathbb{Z}_2^n$, then:

% %Let $P$ be a problem such that a polynomial  $f$ of degree at most $d$ exists where $P(\vec{v}) =  f(\vec{v}) \pmod{2}$ 
% %for all $\vec{v}\in \mathbb{Z}_2^n$.
% %Let $A$ be an algorithm  that runs in time $T(n)$ such that when $\vec{v}$ is sampled uniformly from $\mathbb{Z}_2^n$:
% $$Pr[A(\vec{v}) = P(\vec{v})] \geq 1-1/2^{d+3}.$$
% Then there is a randomized algorithm $B$ that runs in time $O \left( 2^{d+1}(n  + T(n)) \right)$ such that
% for \emph{any} $\vec{v} \in \{0,1\}^n$:
% $$Pr[B(\vec{v}) = P(\vec{v})] \geq 3/4.$$ 
% \end{restatable}

We now give the main ideas behind our framework. We reduce a worst-case instance to $2^{O(d)}$ average case instances. For binary-output problems (Theorem \ref{thm:framework-parity}), we generate $d+1$ random vectors of length $n$ (where $n$ is the size of the input, i.e. the number of input bits) and then for each of the $2^{d+1}-1$ non-trivial\footnote{where there is a non-zero coefficient} linear combinations of these vectors (mod 2), 
we make an average-case instance by adding (mod $2$) the linear combination to the original input to create a random input. Then it is easy to see that the output of the worst-case instance is simply the sum of the outputs of these $2^{d+1}-1$ average-case instances, as all the random vectors we added cancel out their contributions once we sum them up.

However, this approach does not work for general integer problems (Theorem \ref{thm:framework-general}), as every bit in the original input is a number. Therefore, we cannot do the bit operations like we could do with mod 2. %, and this is because valid inputs to the worst-case or average-case instances of the problems (not polynomials) are $n$ bit inputs \mina{as opposed to?}. However, when $p>2$ every bit in the original input will instead be a random number mod $p$ because we go through a polynomial over the field $F_p$. To avoid this problem, we take a different approach. 
Instead, we produce $2^d$ average-case instances as follows. Given a $d$-partite input, let us refer the partition of the input bits by $P_1,..,P_d$. We produce $n$ random $t$-bit numbers $r_1,\ldots,r_n$, 
where we specify $t$ later in the algorithm. We then consider all possible $2^d$ assignments of $0,1$ to the partitions $P_1,\ldots,P_d$ as labels. Let vectors $\vec{w} \in \{0,1\}^d$ represent these labels. %, which we interpret as assigning a bit to each  of the $d$ parts $P_1,P_2,..,P_d$ of the input (there are $2^d$ such vectors). 
  Each label vector  $\vec{w}$ produces a new instance of the problem as follows: %For each partition label vector $\vec{w}$ %we alter each bit of the original input as follows. 
Let $x_i$ be the $i^{th}$ bit of the input, and suppose that it is in partition $P_j$. %For each bit $x_i \in P_j$ we apply the following procedure, 
If the label of $P_j$ is $0$, i.e. $\vec{w}[j]=0$, then we add $r_i$ to $x_i$. If the label of $P_j$ is one, then we add $-r_i$ to $x_i$. In other words,
we add $(-1)^{\vec{w}[j]} r_i$ to $x_i$. %(so, we either subtract or add $r_i$ depending on the bit assigned to the partition, $P_j$, that $x_i$ is in). 
So this produces $2^d$ instances, one for each possible vector $\vec{w}$. 

Note that we are extending each input bit to $t$ bits since we are adding a $t$-bit number to it. So we need to convert these instances with $nt$-bits of input back to instances with $n$-bits of input to be able to use the fine $d$-degree polynomial representing the problem. %{\color{red} WHAT IS FINE $d$ DEGREE POLYNOMIALS?}\mina{added in a few paragraphs above} 
The idea for this conversion is based on the following observation: if we are multiplying $d$ numbers each with $t$ bits, we can instead break this down into $t^d$ bit-wise multiplications that we weigh appropriately, by expanding each number $\vec{x} = (\vec{x}[t-1],\ldots,\vec{x}[0])$ into $\sum_{i=0}^{t-1}\vec{x}[i]2^i$. For example, to compute $\vec{x}\cdot\vec{y}$, we could instead compute the following $t^2$ bit-wise multiplications: for each $i,j\in \{1,\ldots,t\}$, compute $\vec{x}[i]\cdot \vec{x}[j]$ with weight $2^{i+j}$. Summing up the outcome of these weighted bit multiplications is equal to $\vec{x}\cdot\vec{y}$. %if you are multiplying the $i_1^{th}$ bit, the $i_2^{th}$ bit, ..., and the $i_d^{th}$ bit that outcome should be weighted by $2^{i_1}\cdot 2^{i_2} \cdots 2^{i_d}$. 
%We can produce $t^d$ inputs, where we have $n$ input bits (instead of $n$ $t$ bit numbers), to achieve this outcome. 

By expanding each input bit 
 and taking $t > \log{M}$,
 %and The reason we need to expand each input bit is to be able to compute the output mod $p$, and so we set $t$ to be bigger than $\lg{p}$. 
 we can ensure that we don't lose any information while adding these random numbers to our input. Moreover, $d$-partiteness is needed in converting the long inputs back to $n$-bit inputs. Without $d$-partiteness it is not clear how to break each instance with long inputs into shorter input instances.

% {\color{red} Barna: Is there a way to mention for p>2, what are the main changes from the constructions of DLW20?}
% \andrea{I am not totally sure what you mean by this. Most of the results in this paper we do mod 2 (because it is more efficient and slightly more impressive). However, in the first paper we never get ANY implication for average-case counting $\{$ $K$-SUM, $K$-XOR, $K$-OV $\}$ and here we do. So, the difference is that we get any results for these problems. The construction difference is just that we can take the average-case hardness of factored problems and use them to get average-case hardness for the un-factored versions (over weird distributions).}

\paragraph*{From Factored Problems to Unfactored Problems.} Our main results are reductions from factored problems to un-factored problems for both counting and parity versions. More particularly, from factored $k$-$A$ to $K$-$B$ where $A,B\in \{$XOR, SUM, OV$\}$, and $K$ is a function of $k$.

To understand our techniques it is crucial to understand factored vectors and factored problems. As mentioned earlier, a factored vector $\vec{v}$ of dimension $d$ with parameters $b$ and $g=d/b$ consists of $g$ sets $\vec{v}[1],\ldots,\vec{v}[g]$, each containing $b$-bit numbers. For instance for $b=g=2$, a factored vector $\vec{w}$ could have  sets $\vec{w}[0]=\{11,01\}$ and $\vec{w}[1]=\{00,11\}$. Each factored vector can represent many un-factored vectors. For instance, $\vec{w}$ represents the following $4$-bit vectors: $1100,1111,0100,0111$.

We first clarify the notation for a vector and a factored vector. For a \emph{vector} $\vec{v}$, $\vec{v}[i]$ is the $i^{th}$ bit of $\vec{v}$, so $\vec{v}[i]\in \{0,1\}$. For a \emph{factored vector} $\vec{v}$, $\vec{v}[i]$ is a set of vectors of length $b$, i.e. $\vec{v}[i]\subseteq \{0,1\}^{b}$.

Suppose $S(\vec{w})$ is the set of un-factored vectors that a factored vector $\vec{w}$ represents. Then (the counting version of) factored $k$-OV with parameters $b$ and $g$ gets as an input $k$ lists $I_1,\ldots,I_k$, each consisting of $n$ factored vectors and wants the number of un-factored vectors $u_1,\ldots,u_k$ such that $u_1\cdot \ldots \cdot u_k=0$ (they are orthogonal) and $u_i\in S(\vec{v_i})$ for some factored vector $\vec{v_i} \in I_i$, for all $i=1,\ldots,k$. See Figure \ref{fig:factored}.

\begin{figure}
    \centering
    \includegraphics[width=\linewidth]{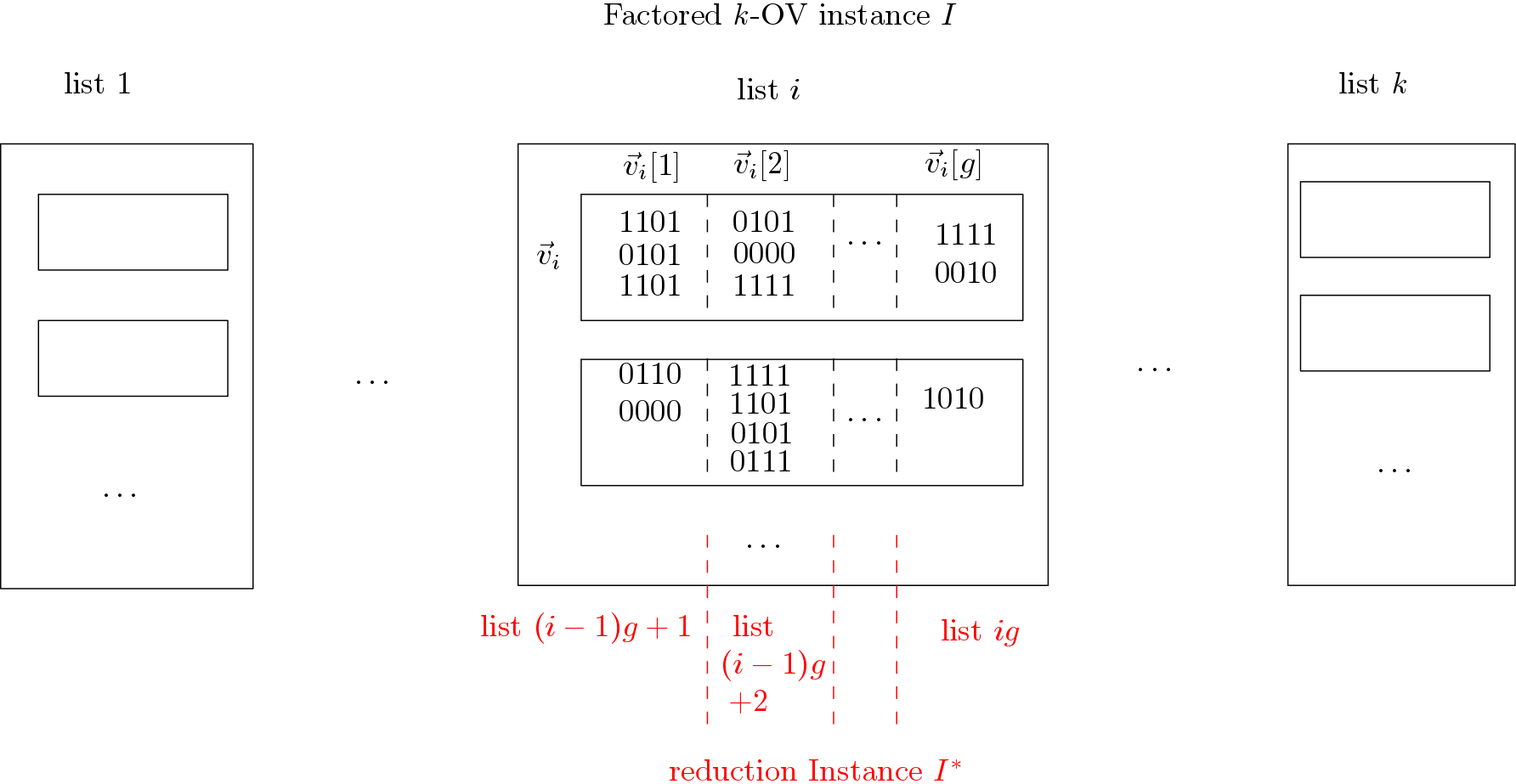}
    \caption{An example of a factored $k$-OV instance, where each factored vector has $g$ sets of $4$-bit numbers ($b=4$), see $\vec{v}$ in list $i$ as an example for a factored vector. In reducing a factored $k$-OV instance $I$ to a $kg$-OV instance $I^*$, we take each set in factored vectors as a new list, as specified by the red markers.  
    }
    \label{fig:factored}
\end{figure}

 % While Dalirrooyfard et al \cite{factoredProblems} could reduce factored (eg factored $\# k$-OV, factored $\# k$-SUM) %{\color{red}(WHICH FACTORED PROBLEM?)} 
 % problems to problems such as Longest Common Subsequence (LCS), they could not reduce from factored problems (e.g. factored $\# k$-OV) back to their less expressive unfactored versions (e.g. $\# k$-OV).  In this paper we give a surprisingly efficient reduction from factored problems back to unfactored problems, which is our main contributions. 

 The counting versions of the factored problems satisfy the constraints in the framework proposed in \cite{factoredProblems}. 
Similarly, they satisfy our less restrictive constraints in frameworks specified in Theorem \ref{thm:BetterFrameworkLargeFeild} and Theorem \ref{thm:framework-parity}. However, while \cite{factoredProblems} stopped at proving average-case hardness for counting versions of the factored problems, our main contribution is to reduce the factored problems to their un-factored versions via an  efficient reduction. \textit{This completes the loop so that we can get hardness results for natural core problems in fine-grained complexity, and not just for artifically defined factored problems.}

An efficient reduction between factored problems and un-factored problems is surprising due to the expressiveness of the factored versions. Intuitively, a factored problem represents a very large compressed instance of the original problem. 
Consider the case of $k$-OV, a factored $k$-OV instance represents a \emph{super polynomial} number of (highly correlated) vectors. Reducing back to $k$-OV naively would require a super-polynomial input size. As $k$-OV is only $n^k$ hard this would be too inefficient. We present a reduction method that manages to capture the compression in the way we produce our instance. %We do this by starting from the $k$-OV hypothesis and ending with an average-case parity $(k^2)$-OV instance. 
 
%\mina{should add the k-clique reductions} \mina{make sure we have this: Our results answer the open question posed in 3SUM, 3XOR, Triangles by Jafargholi and Viola \cite{afargholiV16}.}

%\mina{should do a different example.} We want to give a depiction of the general strucutre of our reductions. \andrea{TODO: example not of sat.}

Now we explain our reduction from a factored $k$-OV instance $I$ with parameters $b$ and $g$ to a $K$-OV instance $I^*$, where $K=kg$. We illustrate the reduction using the counting version, but the same reduction also works for the parity version. The $k$-OV instance $I$ consists of $k$ lists $I_1,\ldots,I_k$ of factored vectors, and the $K$-OV instance $I^*$ has $K$ lists $I^*_1,\ldots,I^*_K$ of vectors. We are going to turn each ``subset" of the factored vectors into a list: To make this concrete, consider all the factored vectors in the $i$th list $I_i$ of the $k$-OV instance $I$, for some $i\in \{1,\ldots,k\}$. Consider the $j$th set $\vec{v}[j]$ of $b$-bit vectors for each of these factored vectors $\vec{v}\in I_i$. The union of these sets is going to construct the $(i-1)g+j$th list $I^*_{(i-1)g+j}$ of the $K$-OV instance $I^*$ %Vectors in list $(i-1)g+j$ of $I^*$ come from the $b$-bit numbers in the $j$th set of the factored vectors in list $i$ of $I$ 
(see Figure \ref{fig:factored}). 

If we let list $I^*_{(i-1)g+j}$ be exactly the union of the sets the reduction won't work. In fact we need to make sure that the $b$-bit vectors selected from lists $(i-1)g+1,\ldots,(i-1)g+g$ of $I^*$ are all from the same factored vector in $I$. To do this we expand each $b$-bit vector to add some ``check" bits.

%In fact, we are going to expand each of these $b$-bit numbers. The new longer vector has two parts. The first part is going to check if all the vectors picked from lists $(i-1)g+1,\ldots,(i-1)g+g$ are from the same factored vector in list $i$ of $I$. 

Moreover, this bit expansion has another effect: Suppose the vectors picked from lists $(i-1)g+1,\ldots,(i-1)g+g$ in $I^*$ are all from factored vector $\vec{v}$ in list $i$ of $I$. We build our reduction in such a way that the dot product of these vectors results in an un-factored vector from $S(\vec{v})$, i.e. the set of un-factored vectors that the factored vector $\vec{v}$ represents. This way, the dot product of all $K=kg$ vectors in $I^*$ is going to simulate one solution to the $k$-factored instance $I$.
%The second part of a long vector in $I^*$ has the associated $b$-bit vector in it, such that the dot product of these $g$ vectors in lists $(i-1)g+1,\ldots,(i-1)g+g$ results in an un-factored vector in $S(\vec{v})$, i.e. an un-factored vector that $\vec{v}$ represents. This way, the dot product of all $kg$ vectors is going to simulate one solution to the $k$-factored instance.

The reductions from factored $k$-SUM and factored $k$-XOR to $K$-SUM and $K$-XOR respectively are similar, we only have to replace dot product to the appropriate function and make some minor adjustments. Finally, note that these reductions work when we consider the parity versions of the problems. 

\paragraph*{From Worst case to Average case.} 
In order to get hardness for average-case  $\paritykOV$ from rETH, we first (1) reduce (worst-case) $K$-$OV$ (for some $K$ that is a function of $k$) to (worst-case) factored $\paritykOV$, (2) reduce worst-case factored $\paritykOV$ to average case factored $\paritykOV$, and then (3) use the reduction explained above to reduce (average-case) factored $\paritykOV$ to (average-case) $\paritykOV$. To get hardness for average-case $\paritykSUM$ and $\paritykXOR$ from $\paritykSUM$ and $\paritykXOR$ respectively, we perform similar reductions to the steps above. Note that as mentioned, \cite{factoredProblems} would get stuck in the last step, and hence could not get hardness for un-factored version of the problems above.

To get hardness from a different problem, for example from the $K$-$SUM$ hypothesis to average-case $\paritykOV$, we add a step to the reduction above. We first (1) reduce $K$-$SUM$ to factored $\mathsf{parity}$-$K$-$SUM$ and then (2) worst-case factored $\mathsf{parity}$-$K$-$SUM$ to average-case factored $\mathsf{parity}$-$K$-$SUM$. Then we apply the additional step: (3) we reduce factored $\mathsf{parity}$-$K$-$SUM$ to factored $\mathsf{parity}$-$K$-$OV$. Finally 
 (4) we reduce factored $\mathsf{parity}$-$K$-$OV$ to $\mathsf{parity}$-$k$-$OV$. 
 %A weaker version of this additional step was mentioned in \cite{factoredProblems}, 
%{\color{red} What does "mentioned" mean? Was this done in DLW20 with inferior parameter?, or it can be done but was not done.) }
%where the reductions there were not strong enough to give us good lower bounds. 
In other words, when we reduce factored parity-$k$-SUM %{\color{red}(SHOULD IT BE $k$-SUM??} 
with parameters $b$ and $g$ to factored parity-$k'$-OV with parameters $b'$ and $g'$, we want to keep $b'$, $g'$ and $k'$ as close to $b$,$g$ and $k$ as we can, respectively. 
%{\color{red} Mention what would be the equivalent of Theorem 1.3 if we use their reduction for parity-K-OV starting from parity-K'-SUM. Theorem 1.3 condition is not on parity-K'-SUM, but just K'-SUM. Why?}%However, \cite{factoredProblems} reductions have a blow up in $b'$. The reason is that \cite{factoredProblems} didn't need efficient reductions, hence they kept them simple.
%%% I think the comment in red was resolved
In Section \ref{sec:betweenProblem-bg} we give more efficient reductions between factored problems (between any two of factored $k$-OV, $k$-SUM and $k$-XOR), using some of the encoding ideas explained above. Note that in this additional step, both problems considered are factored, so we only need to encode one operation (say sum) into another (say dot product).

\subsection{Comparison to Prior Work}

\def\sharpP{\#P}
\def\parityP{\oplus P}
%\mina{maybe again emphasise that the other framework versions used a lot of number theory stuff?}
Perhaps the most related result to ours is the classical worst-case to average-case reduction for the problem of computing the permanent of an $n\times n$ matrix, which is complete for the counting complexity class $\sharpP$~\cite{Lipton89,CPS99,Guruswami06}. 

There has been a number of other works in the recent past that showed fine-grained average-case hardness of various computational problems. Ball, Sabin, Rosen and Vasudevan~\cite{BallRSV18} kickstarted this series of works by using the local correctability of low-degree polynomials (equivalently Reed-Muller codes~\cite{Lipton89,FeigenbaumF91,GLRSW91,GS92}) to show the hardness of counting problems {\em modulo large enough prime numbers} assuming one of several popular worst-case conjectures in fine-grained complexity. Subsequent works of Goldreich and Rothblum~\cite{GoldreichR18} and Boix{-}Adser{\`{a}}, Brennan and Bresler~\cite{UniformCliqueABB} used the same local correctability properties to show the hardness of counting the number of $k$-cliques for some samplable distribution~\cite{GoldreichR18} and for Erd\"{o}s-Renyi graphs~\cite{UniformCliqueABB}. As for the other problems, average-case hardness for the uniform distribution of $k$-SUM was shown when the range the numbers are drawn from is large \cite{BrakerskiSV21}. Specifically, if $n$ numbers are drawn uniformly at random from $[-R,R]$ then an algorithm for $k$-SUM running in $R^{o(1/lg(k))}$ would give surprising improvements for lattice problems.

Chen, Hirahara, and Vafa \cite{CHV23} attempted to find the minimal worst-case complexity assumption which implies average-case hardness for NP and PH. Most relevantly for this paper they show that if $\sum_k$-SAT can't be solved in time $2^{\Tilde{O}(\sqrt{n})}$ then $\sum_2 Time[n]$ can't be solved in quasi-linear time. Our results give stronger lower-bounds, but we start from stronger assumptions (e.g. rETH).

Some recent work has given average-case to average-case reductions. These papers show the equivalence or hardness of new distributions with previous well-studied distributions. In \cite{DBLP:conf/focs/DinurKK21} the authors reduce between sparse and dense settings of k-SUM and k-XOR. In \cite{agrawal2023k} the authors also present average-case to average-case hardness results for $k$-SUM, paying particular attention to the sparse regime where $r$ integers are chosen uniformly at random from $\{0,\ldots,M-1\}$ for $M\gg r^k$. In our paper, by contrast, we focus on \emph{worst-case} to average-case reductions. However, these recent papers highlight the interest in the hardness of $k$-SUM and $k$-XOR in the average-case. 

%The work of Goldreich~\cite{Goldreich20} is the closest to ours in technical terms.

While the results of \cite{UniformCliqueABB} and \cite{Goldreich20} give average-case hardness for counting the parity of the number of $k$-cliques these results have not yet been applied to give lower bounds for other problems. The overhead of Boix{-}Adser{\`{a}}, Brennan and Bresler grows as $\lg(n)^{\binom{k}{2}}$ and Goldreich's paper grows as $2^{\binom{k}{2}}$, however, the hardness of $k$-clique grows as $n^k$ \cite{UniformCliqueABB}\cite{Goldreich20}. This causes the reductions to break down when $k=\omega(\lg(n)/\sqrt{\lg\lg(n)})$ and $k= \omega(\lg(n))$ respectively.

Dalirooyfard, Lincoln and Vassilevska Williams~\cite{factoredProblems} do show an average-case hardness result for factored $\paritykXOR$, however they did not show hardness for (un-factored) $\paritykXOR$. The key contribution of this work is a (worst-case) reduction from factored problems to their un-factored (i.e. regular, good old) versions. %An efficient reduction between factored problems and un-factored problems is surprising due to the expressiveness of the factored versions. Intuitively, a factored problem represents a very large compressed instance of the original problem. Consider the case of $k$-OV, a factored $k$-OV instance represents a \emph{super polynomial} number of (highly correlated) vectors. Reducing back to $k$-OV naively would require a super-polynomial input size. As $k$-OV is only $n^k$ hard this would be too inefficient. In this paper we present a reduction method that manages to capture the compression in the way we produce our instance. We do this by starting from the $k$-OV hypothesis and ending with an average-case parity $(k^2)$-OV instance. \mina{moved this paragraph up}

There has also been several recent works that study the fine-grained complexity of parity problems. For example, \cite{DellHMTW14} demonstrates a lower bound for the permanent from $\#$ETH, the counting version of ETH. To do this they show a sparsification lemma for $\#$ETH. Dell, Lapinskas and Meeks \cite{DellL21,DellLM22} study fine-grained reductions from approximate counting to decision. These results can't be used to get average-case decision hardness from the average-case hardness of counting problems (e.g. \cite{UniformCliqueABB},  \cite{factoredProblems}, and this paper) because the worst-case to average-case counting reductions rely crucially on the exact answers, not approximate answers. Another work on the fine-grained complexity of parity is that of \cite{AbboudParity} which shows that distance problems such as graph diameter fine-grained reduce to computing the parity of the corresponding distance values. These works study the worst-case versions of parity problems whereas our focus is on average-case hardness.%\mina{is the last paragraph the recent papers we discovered?}
%\vnote{This work is far more important to cite than Abboud: https://arxiv.org/abs/1206.1775  See publications of https://dblp.org/pid/27/3557.html}
%\andrrea{Thank you Vinod!}

%\vinod{https://arxiv.org/abs/2010.08821} 

%Additionally, note that the parity problems they study ask for the parity of some property, not the parity of the count of the number of witnesses. They note that their reduction techniques don't work for problems where you count witness (see their discussion of negative weight triangle).  We study problems where we are counting the number of witnesses, as opposed to computing the parity of a graph property. So their results do not let us show average-case decision hardness from our proofs of parity hardness. However, we hope that extensions of their line of work may give parity witness counting to decision problems reductions in the future. Such work could potentially offer barriers to hardness or new lower bounds for decision problems in the average-case. 
%\andrea{not sure if I strike the right tone here. I was trying for `their work is cool, but, reviewer, decision is not the same as parity in our context'. Did I manage it?}

%\vnote{Note to Self: sticking with local correctability over large prime fields will increase the number of variables by a multiplicative factor of log n and so that has no hope of giving us the tight hardness we get in this paper?}

\section{Preliminaries}
\label{sec:prelims}

For a distribution $D$ and a variable $r$, $r\sim D$ means that we sample $r$ from the distribution $D$. For a vector $\vec{v}$, $\vec{v}\sim D^n$ means that we sample each entry of $\vec{v}$ from $D$.

For any counting problem $P$, let $\oplus P$ be the parity version of $P$, where the output of $\oplus P$ is the parity of the output of $P$. For the rest of this paper we will use $\oplus$SAT, $\oplus k$-OV, $\oplus k$-XOR, and $\oplus k$-SUM to refer the problem of returning the parity of the number of solutions to these problems. 
%\barna{Good place to define the parity version. Lets use $\paritySAT$, $\paritykOV$ etc. until this point.}\anote{I think that $\oplus k$-OV etc are clear, but I am happy for the $\paritykOV$ version to be used.}

\subsection{Factored Problems}
We use factored problems as a bridge to reduce worst-case to average-case of many problems.
%\xxx{We will use the $k$-partite versions of the factored problems through out the whole paper. They are \textit{so} much easier to work with. }
%\mina{When do we use backslash vec? maybe let's use it only for factored vectors to increase readability?}
\paragraph{Factored Vector~\cite{factoredProblems}} Given parameters $b$ and $g$, a $(b,g)$-factored vector $\vec{v}$ consists of $g$ sets of $b$-bit zero-one vectors. In particular, for $i\in \{1,\ldots,g\}$, $\vec{v}[i]$ is the $i^{th}$ set of $b$-bit vectors. Let the set of factored vectors with parameters $g$ and $b$ be $Fac(b,g)$.

Below are three $(3,2)$-factored vectors $\vec{u},\vec{v}$ and $\vec{w}$.

\begin{align*}
   &\vec{u}[0] = \{101,011\} \text{~~} & \vec{u}[1] = \{001,100\}\\
   &\vec{v}[0] = \{\} \text{~~} & \vec{v}[1] = \{000,011,100,111\}\\
   &\vec{w}[0] = \{100\} \text{~~} & \vec{w}[1] = \{010,111\}
\end{align*}

We first define the factored $k$-OV problem using the example above, and then give a definition for factored problems in general. A factored $k$-OV problem of size $n$ consists of $k$ sets $V_1,\ldots, V_k$ each with $n$ factored vectors. For our example, we consider a $2$-OV instance of size $3$, where $V_1=V_2=\{\vec{u},\vec{v},\vec{w}\}$. Informally, the problem asks to compute the sum of the number of ``ways" each two factored vectors from $V_1$ and $V_2$ can be orthogonal to each other. For example, the number of ways $\vec{u}$ and $\vec{w}$ are orthogonal is the number of vectors $u_0\in \vec{u}[0]$, $u_1\in \vec{u}[1]$, $w_0\in \vec{w}[0]$, $w_1\in \vec{w}[1]$, such that $u_0\cdot w_0=0$ and $u_1\cdot w_1=0$\footnote{Two zero-one vectors $a=(a_1,\ldots,a_t)$ and $b=(b_1,\ldots,b_t)$ with length $t$ are orthogonal (have dot product zero) if $a\cdot b =\sum_{i=1}^t a_ib_i=0$.}. There are in fact only two $4$-tuples of $(u_0,u_1,w_0,w_1)$ with this property: $(001, 001,100,010)$ and $(001, 100,100,010)$. %\vnote{Perhaps I am missing something, but isn't $u_1\cdot w_1 = 1$ in all the three tuples?} 
It is easy to see that the number of ways $\vec{v}$ is orthogonal to $\vec{u}$ and $\vec{w}$ is zero, so the answer to this factored $2$-OV problem is $2$.

%\barna{We can just keep $(g,b)$-factored vector here, instead of defining another $Fac(b,g)$. If so, it is better to just say in the below definition $k$ $(g,b)$-factored vectors $V_1,..V_k$-}

\paragraph{Factored Problems~\cite{factoredProblems}} Let $\frak f$ be a function that gets $k$ $b$-bit numbers as input and outputs zero or one. For any function $\mathfrak{f}$ that gets $k$ $b$-bit vectors as input and outputs zero or one, we define the factored $\mathfrak{f}$ problem \ckfunc[]$(n,b,g)$ as follows.
The input to the problem is $k$ sets $V_1,\ldots,V_k$ each having $n$ factored vectors from $Fac(b,g)$,
%\barna{what is size $n$ of factored vectors mean?}\mina{the sets are of size $n$} 
the set of $(b,g)$-factored vectors. We refer to $n$ as the size of the problem. Informally, the output is the sum of the number of ``ways" any 
$k$ factored vectors $\vec{v}_1\in V_1, \ldots,\vec{v}_k\in V_k$ ``zero" $\mathfrak{f}$: the number of $gk$-tuples $(w_1^1,\ldots, w_1^g,w_2^1,\ldots,w_k^g)$ where $w_i^j\in \vec{v}_i[j]$, and $\mathfrak{f}(w_1^j,w_2^j,\ldots, w_k^j)=0$ for all $j=1,\ldots,g$. So more formally, the output of a \ckfunc[]$(n,b,g)$ instance is defined as follows:
\begin{equation*}
    Fk\text{-}\mathfrak{f}(V_1,\ldots,V_k):=\sum_{\vec{v}_1,\ldots,\vec{v}_k\in V_1,\ldots,V_k} |\{(w_1^1,\ldots,w_k^g)| \forall i\forall j: w_i^j\in \vec{v}_i[j] \text{ and }\forall j:\mathfrak{f}(w_1^j\ldots,w_k^j) = 0\}|
\end{equation*}

In the above definition $j\in \{1,\ldots,g\}$ and $i\in \{1,\ldots,k\}$. 

Now the decision version of the problem returns True if \ckfunc[]$(V_1,\ldots,V_k)>0$, and the parity version $\oplus$\ckfunc[] outputs the parity of \ckfunc[]$(V_1,\ldots,V_k)$. In this paper we mostly focus on the parity version of the problems.

\subsection{Un-factored Problems}

We state our hardness hypotheses. We state these hypotheses in the word-RAM with $O(\log(n))$ bit words. 
%\vvw{when stating hardness hypotheses, need to state the model of computation. word-RAM with $O(\log n)$ bit words, in our case.}
%%% Vinod Defs

\def\vecv{\mathbf{v}}
\def\vectt{\mathbf{t}}
\def\F{\mathbb{F}}
\def\matV{\mathbf{V}}
\def\matI{\mathbf{I}}
\def\matA{\mathbf{A}}
\def\matB{\mathbf{B}}
\def\vecs{\mathbf{s}}
\def\vecy{\mathbf{y}}
\def\vece{\mathbf{e}}
\def\Wt{\mathsf{Wt}}

\begin{definition}[The $k$-clique Hypothesis]
    Given an unweighted graph $G$ with $n$ nodes and $m=O(n^2)$ edges counting the number of $k$-cliques in the graph requires $n^{\omega k/3 -o(1)}$ time, even for randomized algorithms. 
    \label{def:kcliqueHypothesis}
\end{definition}

\begin{definition}[The $k$-XOR Hypothesis]
In the \emph{$k$-XOR problem}, we are given $k$ unsorted lists $L_1,\ldots,L_k$ each containing $n$ $d$-bit vectors %$v_1,\ldots,v_n \in \F_2^\ell$ 
for some dimension $d=O(\log{n})$, and want to determine if there are $v_1\in L_1,\ldots,v_k\in L_k$ such that the XOR of  $v_1, v_2,\ldots ,v_k$ equals zero. The counting version of $k$-XOR asks how many tuples of $k$ numbers $a_1\in L_1, \ldots, a_k \in L_k$ XOR to zero.  The $k$-XOR hypothesis states that the $k$-XOR problem requires $n^{\ceil{k/2}-o(1)}$ time, even for randomized algorithms. %\xxx{cite} %return true if there are $k$ indices $i_1,\ldots,i_k \in [n]$ such that the XOR of the $k$ associated vectors $v_{i_1},\ldots,v_{i_k}$ equals zero.%\mina{Do we really need the search version?}\andrea{nope! switched} %$$\bigoplus_{j\in [k]} \vecv_ s{i_j} = 0$$
\end{definition}

%\vnote{This is the homogeneous version of the problem. In the inhomogeneous version, you are also given a target $\vectt \in \F^{\ell}$ and are asked to add up to $\vectt$.} \ailnote{I think the inhomogenous version is easier for the rest of the paper.}\mina{i think we never use the inhomogenous version, we alway consider the target to be zero.}

\begin{definition} [The $k$-SUM Hypothesis \cite{C3sum}]
	In the \emph{\ksum[]~problem}, we are given $k$ lists $L_1,\ldots,L_k$ each consisting of $n$ numbers (over $\Z$ or $\R$) and want to determine if there are $a_1\in L_1, \ldots, a_k \in L_k$ such that $\sum_{i=1}^k a_i = 0$. The counting version of \ksum[]~asks how many tuples of $k$ numbers $a_1\in L_1, \ldots, a_k \in L_k$ sum to zero. 
The \ksum[]~hypothesis states that that the \ksum[]~problem requires $n^{\lceil k/2 \rceil -o(1)}$ time for randomized algorithms \cite{C3sum}.
%This is equivalent to saying no $n^{\lceil k/2 \rceil-\eps}$ time algorithm exists for \ksum[]~for constant $\eps>0$. 

	\label{def:ksum}
\end{definition}

\begin{definition}[(Strong) Exponential Time Hypothesis\cite{cseth}]
Let $c_k$ be the smallest constant such that there is an algorithm for $k$-CNF SAT that runs in $2^{c_k n +o(n)}$ time.
Let $r_k$ be the smallest constant such that there is a randomized algorithm for $k$-CNF SAT that runs in $2^{r_k n +o(n)}$ time.

The \emph{Exponential Time Hypothesis (ETH)} states that $c_k>0$ for all $k\geq 3$. 

The \emph{Random ETH (rETH)} states that $r_k>0$ for all $k\geq 3$.

The \emph{Strong Exponential Time Hypothesis (SETH)} states that there is no constant $\eps>0$ such that $c_k \leq 1-\eps$ for all constant $k$.
\end{definition}
 
Intuitively, ETH states that $k$-CNF SAT requires $2^{\Omega{(n)}}$ time and SETH states that there is no constant $\epsilon>0$ such that there is a $O(2^{n(1-\epsilon)})$ time algorithm for $k$-CNF SAT for all constant values of $k$.

\begin{definition} [The $k$-OV Hypothesis \cite{virgiSurvey}]
In the \emph{\kOV~problem}, we are given $k$ lists $L_1,\ldots, L_k$ of $n$ zero-one vectors of length $d$ as input. If there are $k$ vectors $v_1 \in L_1,\ldots,v_k \in L_k$ such that for $\forall i\in[1,d]~\exists j\in[1,k]$ such that $v_i[j]=0$ we call these $k$ vectors an orthogonal $k$-tuple. The output of \kOV is  true if there is an orthogonal $k$-tuple in the input and false otherwise. The counting version of $k$-OV asks for the number of orthogonal $k$-tuples. The \kOV~hypothesis states that that the \kOV~problem requires $n^{k-o(1)}$ time, even for randomized algorithms \cite{virgiSurvey}.  
	\label{def:kOV}
\end{definition}

First we recall a result that shows hardness on $k$-OV assuming rETH.

\begin{lemma}
%\mina{maybe remove the proof and just cite somewhere that proves this? } \andrea{I am sure this is proved somewhere, but I am not sure where XD}  
Assuming rETH there exists a  fixed constant $c_\epsilon$ such that $\forall c\geq c_{\epsilon}$ the $k$-OV problem with $n$ vectors each of length $c_\epsilon k\lg(n)$ requires $n^{
\Theta(k)}$ time.
\label{lem:dimensionReductionETH}
\end{lemma}
\begin{proof}
By rETH $3$-SAT requires $2^{\Theta(n)}$ time. We use the Sparsification Lemma from Calabro, Impagliazzo and Paturi \cite{CalabroIP06} to state that if $3$-SAT requires $2^{\Theta(n)}$ time then $3$-SAT with at most $m=(6/\epsilon)^{9}n$ clauses requires $2^{\Theta(n) - \epsilon n}$ time\footnote{The previous sparsification lemma is also sufficient for this result \cite{ImpagliazzoPZ98}. However, we actually use the more efficient version for convenience.  }. So we can say that there is some constant $c$ such that $3$-SAT  requires $2^{\Theta(n)}$ time on formulas of size $m = c n$. 

Now we can use Williams' reduction from sparsified $K$-SAT
%$\ell_{\epsilon/2}$-SAT 
to $k$-OV \cite{Williams05}, where we just apply it for $K=3$. 
 In this reduction we create $N = 2^{n/k}$ vectors of length $d=m=cn$. 
When written in terms of $N$ we get that $d = ck\lg(N)$, and $k$-OV for $N$ vectors of dimension $ck\lg (N)$ requires $2^{\Theta(n)}$ time under rETH.
%So we can get a lower bound on this instance of $k$-OV in terms of  its input size and thus the $k$-OV instance requires $N^{k(1-\epsilon-o(1))}$ time.
\end{proof}

Now we show that $k$-XOR and $k$-SUM are hard on small range numbers, under the $k$-XOR and $k$-SUM hypothesis respectively. 
\begin{lemma}
If the $k$-SUM hypothesis is true then the $k$-SUM problem on numbers in the range $[-2n^k,2n^k]$ ($k\lg(n)+2$ bit numbers) requires $n^{\lceil k/2 \rceil -o(1)}$ time.

If the $k$-XOR hypothesis is true then $k$-XOR problem on $k\lg(n)+2$ bit numbers requires $n^{\lceil k/2 \rceil -o(1)}$.

\label{lem:lengthForkSumkXOR}
\end{lemma}
\begin{proof}

In both cases we use hash functions to hash big values to the range $[-2n^{k}, 2n^{k}]$ which preserve linear relationships and introduce few false positives. 

We use the nearly linear hash function of Dietzfelbinger \cite{Dietzfelbinger18} \footnote{The readers can see its use on $k$-SUM in  \cite{Patrascu10} and \cite{Wang14}.}. %The $k$-SUM problem with numbers in the range $[-2n^k, 2n^k]$ for a total range of $4n^k$ with a nearly linear hash function. 
This hash function  has no false negatives (any solution remains a solution) and the expected number of false positives is at most $n^k/R = 1/4$. So, the chance of no false positives is at least $3/4$. 

We show that the same idea works for $k$-XOR. Consider an instance of $k$-XOR with $kn$ vectors of length $d$ ($n$ vectors in each partition). We create this instance as an $kn\times d$ binary matrix $I$ where each row represents one number. 
%Take the input, $I$, of $n$ vectors each of length $d$ of $k$-XOR and interpret this as a $\{0,1\}^{n\times d}$ matrix. 
Let $R$ be a uniformly random matrix from $\{0,1\}^{d \times (k \lg(n)+2)}$, and compute $I'=I\times R \pmod{2}$.
%Multiply this by a uniformly random matrix $R$ from $\{0,1\}^{d \times (k \lg(n)+2)}$. 

The output of $I'$ can be interpreted as $kn$ vectors each of length $k\lg(n)+2$, where we have ``shrunk" each vector. Note that $I'$ can be seen as a $k$-XOR instance with the new vectors, where each vector is assigned to the partition that the original vector was assigned to. 

Now we observe that all linear relationships between the rows of $I$ are preserved in $I'$, that is if $I[j_1] \oplus I[j_2] \oplus ... I[j_\ell] = \vec{0}$ then $I'[j_1] \oplus I'[j_2] \oplus ... I'[j_\ell] = \vec{0}$, for $j_1,\ldots,j_{\ell}\in \{1,\ldots,kn\}$. So, all witnesses to the original $k$-XOR problem remain witnesses here. The probability that $\ell$ vectors, for any constant $\ell$, in $I'$ sum to the zero vector when the corresponding vectors in $I$ did not is $2^{-k\lg(n)-2}$ (note $k\lg(n)+2$ is the new length of our vectors). If $\ell$ vectors sum to a non-zero value then  note
$$I'[j_1] \oplus ... I'[j_\ell] = \vec{s} = I[j_1]R \oplus ... I[j_\ell]R =  (I[j_1] \oplus ... I[j_\ell])R.$$ 
Now consider any given bit of $\vec{s}$. If the original vectors were non-zero then $\vec{s}[i]$ is a $1$ or $0$ mod $2$ with equal probability. The length of the new vector is $k\lg(n)+2$, so the probability the new sum is the all zeros vector is $2^{-k\lg(n)-2}$.
If we union bound across all $k$-tuples of vectors the probability that we have a false witness is at most $1/4$s. So, if we can solve the $k$-XOR problem on $I'$ with vectors of length $k\lg{n}+2$, then we can solve the $k$-XOR problem on vectors of any length with a reduction that takes time $ndk\lg(n)$, which is the input size multiplied by a sub-polynomial factor. 
\end{proof}

\begin{lemma}
\label{lem:decisionToParity}
If the $k$-OV hypothesis is true then $\paritykOV$ requires $n^{k-o(1)}$. 

If the $k$-XOR hypothesis is true then $\paritykXOR$ requires $n^{\lceil k/2 \rceil -o(1)}$.

If the $k$-SUM hypothesis is true then $\paritykSUM$ requires $n^{\lceil k/2 \rceil -o(1)}$.
\end{lemma}
\begin{proof}
These are standard folklore reductions, however we present them here for completeness. First, we check if there are a huge number of solutions. Second, we sub-sample the input and ask the parity solver for the parity of the input. If there are no solutions then the parity solver will always return even parity. We will argue that with good probability we will return odd parity if there are any solutions. 

For each of $k$-OV, $k$-SUM and $k$-XOR we are given $k$ lists $L_1, \ldots, L_k$ of either numbers or  vectors, call this an instance $I$. Let $S_i(I)$ be the set of numbers or vectors in list $L_i$ that appear in at least one solution the problem. Let $C_i(I) = |S_i(I)|$.

Now, our goal is to sub-sample the lists in such a way that we are left with a single solution. The sampling procedure is as follows. Consider a tuple $(j_1,\ldots, j_k) \in [1,\lg(n)+1]^k$. For each list $L_i$ we will sub-sample a list deleting each entry with probability $1-2^{-j_i}$. We do this $\lg^2(n)$ times for each $k$-tuples $(j_1,\ldots, j_k) \in [1,\lg(n)+1]^k$ (note this produces only $O(\lg^{k+2}(n) =n^{o(1)})$ instances). We then run a parity counter on each instance we produce, if any parity counter returns odd we return that there exists a solution (note the parity can never be odd if there were zero solutions to start with). If all return even then we return that there is no solution. Because we make $n^{o(1)}$ calls to instances, if this reduction is correct, it implies that up to sub-polynomial factors the running time of parity counting versions of the problems are at least the running time  of the decision versions.

Let us argue correctness. We will imagine the sampling procedure as happening to each list in order to argue about the probability of a single solution being left. Let $c_1 =C_1(I)$. 
Now consider all the calls we make where $\lfloor \lg(c_1) \rfloor \geq j_1 \geq \lceil \lg(c_1) \rceil$. 
%In each such call the probability that our sub-sampling procedure leaves exactly one element of $S_1$ is one minus the probability that all were deleted and minus the probability that two or more remain\mina{why do we even need the previous sentence?}. Agreed. Removed
The probability that our sub-sampling procedure leaves exactly one element of $S_i(I)$ is $c_1 2^{-j_1} (1-2^{-j_1})^{c_1-1} > \frac{1}{2} \cdot \frac{1}{4}^{2}$. So with probability at least $2^{-5}$ we retain exactly one element of $S_i(I)$. Call this new instance with just $L_1$ sub-sampled $I_1$. Now we proceed as follows, imagine we are given an instance $I_{i-1}$ where the first $i-1$ lists have only one element that participates in a solution, i.e. $C_\ell(I_{i-1})=1$ for all $\ell \in [1,i-1]$. Now let $c_i = C_i(I_{i-1})$ (Note this count is on our new instance $I_{i-1}$ where the first $i-1$ lists have been sub-sampled). Once again consider $j_i$ where $\lfloor \lg(c_i)+1 \rfloor \geq j_i \geq \lceil \lg(c_i)+1 \rceil$. The probability that our sampling leaves exactly one element of $S_i(I_{i-1})$ is once again $c_1 2^{-j_i} (1-2^{-j_i})^{c_i-1} >2^{-5}$. Note that the values of interest for the above argument for $j_2,\ldots, j_k$ depend on which elements were sampled so far. Now consider $I_k$, an instance where we have sampled such that there is exactly one solution per list, this leaves exactly one solution. Such an instance is sampled with probability at least $2^{-5k}$, for the correct set of $j_1,\ldots, j_k$, if there were any solutions in the originally instance $I$. We sample for each set of $j_i$ $\lg^2(n)$ times, so the chance we find a solution is $2^{-\Omega(\lg^2(n))}$. So, we have correctness, finishing our reduction. 
\end{proof}

%\mina{commented noisy xor and learning parity definitions.}

%\paragraph{$(k,k')$-Noisy XOR.} Given $n$ vectors $\vecv_1,\ldots,\vecv_n \in \mathbb{F}_2^\ell$, output indices $i_1,\ldots,i_k \in [n]$ such that $$\Wt(\bigoplus_{j\in [k]} \vecv_{i_j}) \leq k'$$
%where $\Wt$ refers to Hamming weight. 

%\begin{proposition} 
% For every $k,k'$, if you can solve $k$-XOR, you can solve $(k,k')$-Noisy-XOR. If you can solve $(k,k')$-Noisy-XOR, you can solve $(k+k')$-XOR.
%\end{proposition}

%This is folklore and has been made several times in the crypto literature. The first statement is trivial. As for the second, given a solution $\matV \vece = \vece'$ where $\matV$ is the matrix whose columns are the vectors $\vecv_i$, $\vece$ has Hamming weight $k$ and $\vece'$ has Hamming weight $k'$,  notice that $$ [\matI\  || \ \matV] \left[\begin{array}{c} \vece' \\ \vece \end{array} \right] = \mathbf{0}$$
%This gives a $(k+k')$-XOR solution to the instance $[\matI || \matV]$. With just a little more work, this can be turned into an average-to-average reduction.  

%\paragraph{Learning Parity with Noise.}\mina{not the full definition} 
%Given an LPN instance $\vecy = \matA \vecs+ \vece$, one can write it as the problem of finding a sparse $\vece$ such that $\matB \vece = \vectt := \matB \vecy$. (Here, $\matB = \matA^{\perp}$ is simply a basis of the left-kernel of $\matA$.) The only problem now is that in LPN, $\vece$ is not required to be {\em super}-sparse, like in $k$-XOR. 

\subsection{Framework Definitions}
\begin{definition}[Polynomial extension]
%Let $\mathbb{F}$ be any finite field, and 
Let $f : \{0,1\}^t\rightarrow \mathbb{Z}$ be any function. %mapping the $t$-dimensional Boolean hypercube to $\mathbb{F}$. 
A $t$-variate polynomial $g$ over $\mathbb{Z}$ is said to be an \textit{extension} of $f$ if $g$ agrees with $f$ at all Boolean-valued inputs, i.e., $g(x) = f(x)$ for all $x \in\{0,1\}^t$. For a problem $P$ with $t$ zero-one inputs and an output in $\mathbb{Z}$, the extension of $P$ is defined similarly. 
\end{definition}

The following definitions are inspired by \cite{factoredProblems}, where all functions and problems are considered modulo some prime $p$. Here, we first re-define the notions in \cite{factoredProblems} without use of any field $\mathbb{F}_p$, and then relax these definitions.%However, we will weaken their restrictions. First, we will define a strongly $d$-partite polynomial. 

\begin{definition}[(Modified) Strongly $d$-partite polynomial \cite{factoredProblems}]
%Let $\mathbb{F}$ be any finite field and 
Let  $f$ be a polynomial in $\mathbb{F}[x_1,\ldots,x_n]$.
 We say $f$ is \textit{strongly $d$-partite} if one can partition the inputs $\{x_1,\ldots,x_n\}$ into $d$ sets $S_1,\ldots, S_d$ such that $f$ can be written as a sum of monomials $x_{i_1}\cdots x_{i_d}$ %\vvw{I assume that you don't want the sum $\sum_i$ in front of the monomial, or otherwise you want to add coefficients in front of the monomials} 
 of degree \emph{exactly} $d$ with coefficients in $\mathbb{F}$ where every variable $x_{i_j}$ is from the partition $S_j$. 
That is, if there is a monomial $x_{i_1}^{c_1} \cdots x_{i_k}^{c_k}$ in $f$ then it must be that for all $j$, $c_j =1 $ and for all $j \ne \ell$ if $x_{i_j} \in S_m$ then $x_{i_\ell} \notin S_m$. %\vnote{Does the degree of each monomial have to be $d$? or could it be smaller? I am getting mixed messages from the above definition.}
\end{definition}

This definition is equivalent to the definition of a set-multilinear polynomial in circuit complexity \cite{github}. 

Next, we will define a $d$-partite polynomial. This is a weaker definition than the strongly $d$-partite polynomials defined by \cite{factoredProblems}, in the sense that the monomials are no longer required to have degree exactly $d$. We also focus on polynomials over the integers.% \vvw{in the sense that the monomials are no longer required to have degree exactly $d$}.

\begin{definition}[$d$-partite polynomial]
%Let $\mathbb{F}$ be any finite field and 
Let the polynomial $f$ be polynomial in $\mathbb{Z}[x_1,\ldots,x_n]$.% \vvw{is $f$ a polynomial in $\mathbb{F}[x_1,\ldots,x_n]$?}.
We say $f$ is \textit{$d$-partite} if one can partition the inputs into $d$ sets $S_1,\ldots, S_d$ such that $f$ can be written as a linear combination of monomials of the form $x_{i_1}\cdots x_{i_{\ell}}$ where $\ell \leq d$,%\barna{can we omit the subscript from $j$ like the above definition?} 
where for all all $p \neq q$ and $p,q \in [1,\ell]$, the variables $x_{i_p}$ and $x_{i_q}$ are from {\em different} partitions $S_{j_p}, S_{j_q}$, $j_p\neq j_q$, $j_p,j_q\in [d]$.
%
%and every variable $x_{i_p}$ is from the partition $S_{j_p}$. Additionally, for all $p \neq q$ and $p,q \in [1,\ell]$ we have that $j_p \neq j_q$. %\vnote{I am very confused... the variables and partitions are indexed by the same number $j_p$ but that doesn't make sense. There are $n$ variables and $d$ partitions.}
\end{definition}

We will define $P(\vec{v})$ to be the correct output for problem $P$ given input $\vec{v}$. Note this is only well defined when $\vec{v} \in \{0,1\}^n$. Now we will define a version of the \goodPoly~from \cite{factoredProblems} where $d$, the degree of the polynomial, is an explicit parameter. This is because unlike \cite{factoredProblems}, we don't want to restrict $d$.
%\barna{What is parameterized version of the good-low degree polynomial?-Explain this a bit more.}

%\vvnote{The definition of a good (and ok) polynomial is just that of a multilinear extension of $P$ (standard in the PCP literature) together with the partition requirement. See \url{http://people.cs.georgetown.edu/jthaler/IPsandextensions.pdf}section 1.3.}\mina{added.i}

\begin{definition}[(Modified) Good Low Degree Polynomial]
Let $P$ be a problem taking in $n$-bit inputs and outputting an integer in the range $[-M,M]$ where $M$ is a integer and $M< n^c$ for some constant $c$. 
A \emph{\goodDPoly} for $P$  %(\gDPol{d}{P})\mina{we probably should never use the abbreviation, it is a bit too much} \vnote{agreed!!} 
is a polynomial $Q$ over $\mathbb{Z}$ where:%\barna{Let's not use $g$ to denote a polynomial, and allocate that for $(g,b)$-factored vectors.}
\begin{itemize}
    %\item There are at most polynomial monomials in $f$.
    \item %If $\vec{I} = b_1,\ldots, b_n$, then 
    %$f(b_1,\ldots, b_n) = f(\vec{I})=P(\vec{I})$ where $b_i$ maps to either a zero or a one in the prime finite field.
    The polynomial $Q$ is an extension of $P$.
    %\item The polynomial $g$ has degree $d$. \virgi{Isn't this implied by the bullet below?}
    \item The polynomial $Q$ is strongly $d$-partite. 
\end{itemize}
\label{def:good-d-Poly}
\end{definition}

The definition of good low degree polynomials in \cite{factoredProblems} is over a prime field $\mathbb{F}_p$ instead of $\mathbb{Z}$. \cite{factoredProblems} showed that all factored problems have good $d$ degree polynomials for an appropriate $d$ and prime numbers $p$. However they actually show that there are polynomials that are strongly $d$-partite and their values matches that of the problem \emph{exactly}, and not modulo prime $p$. We restate this result below, and in this paper by good low degree polynomial we refer to Definition \ref{def:good-d-Poly}.

\begin{lemma}
\label{lem:factored-gdlp}
\cite{factoredProblems} For any integers $n,g,b>0$ with $b=o(\log{n})$ and function $f$, \ckfunc$(n,g,b)$ has good $gk$-degree polynomial.
\end{lemma}
We are going to use Lemma \ref{lem:factored-gdlp} as a middle step in a lot of our reductions. 

%\virgi{OK Poly is okay, but maybe Fine Poly is better?}
Next we will define {\em \okPoly}, which is a less restricted version of \goodDPoly. % allowing higher degree functions. In the sense that in \cite{factoredProblems} a good low-degree polynomial defined the degree $d$ as $o(lg(n)/lglg(n))$, but here, we simply state the degree and note the overhead of the reduction in theorem and lemma statements. \mina{where in the definition we are allowing higher degree functions?} 
We will use this later in our reduction to CNF-SAT. % satisfiability \vvw{? the reduction to CNF-SAT?} reduction. 
We have reduced the constraints on the function by switching from strongly $d$-partite to simply $d$-partite.

%\virgi{Currently OK poly is defined to have range in $F_p$, however later on you seem to use it as if its range is the integers..}
\begin{definition}
Let $P$ be a problem taking in $n$-bit inputs and outputting an integer in the range $[-M,M]$ where $M$ is a integer and $M< n^c$ for some constant $c$.
A \emph{\okPoly}~for problem $P$ is a polynomial $Q$ over $\mathbb{Z}$ where:
\begin{itemize}
    %\item There are at most polynomial monomials in $f$.
    \item %If $\vec{I} = b_1,\ldots, b_n$, then 
    %$f(b_1,\ldots, b_n) = f(\vec{I})=P(\vec{I})$ where $b_i$ maps to either a zero or a one in the prime finite field.
    The polynomial $Q$ is an extension of $P$.
    %\item The polynomial $g$ has degree at most $d$.\virgi{Isn't this implied by the bullet below?}
    \item The polynomial $Q$ is $d$-partite. 
\end{itemize}
%\vvw{Use another letter rather than $g$ since we use $g$ already. Why not $Q$ again?}
\label{def:okPoly}
\end{definition}

\section{Framework}
\label{sec:frameworkBetter}

In this section our inputs are drawn from the uniform distribution over all inputs. % We give the next definition for ease of discussion.
%
% \begin{definition}
% We use $x \sim \mathbb{F}^n_{p}$ to mean that $x$
% is drawn uniformly at random from all $p^n$ values in the support of $\mathbb{F}^n_{p}$.
% \end{definition}
%
First we prove a theorem that holds when a problem $P$ has all outputs in $\{0,1\}$ (as happens with parity). 
%This uses similar techniques that show up in Alon et al \cite{BhattacharyyaKSSZ09}.

%\mina{removed G-d-DP(P) from below. Just need it to be degree at most $d$.}
%\frameworkparity*

\begin{restatable}{theorem}{frameworkparity}
\label{thm:framework-parity}
Let $P$ be a problem that takes an input $I\in\{0,1\}^n$ and has a binary output, i.e. $P(I)\in \mathbb{Z}_2$.
Additionally, assume that a $d$-degree polynomial $f$ exists such that $P(I)=f(I) \pmod{2}$.  
Let $A$ be an average-case algorithm that runs in time $T(n)$ such that when $\vec{v}$ is sampled uniformly from $\mathbb{Z}_2^n$, then:

$$Pr[A(\vec{v}) = P(\vec{v})] \geq 1-1/2^{d+3}.$$
Then there is a randomized algorithm $B$ that runs in time $O \left( 2^{d+1}(n  + T(n)) \right)$ such that
for \emph{any} $\vec{v} \in \{0,1\}^n$:
$$Pr[B(\vec{v}) = P(\vec{v})] \geq 3/4.$$ 
\end{restatable}
%\mina{modifying the parity proof that Andrea had written before.}

%\vvnote{If the degree $d$ is $n$, then this is a trivial union bound. The cleverness seems to come from a form of ``derandomization'' using the $d$-partiteness of the polynomial $f$ (which in turn implies degree $\leq d$). Does this make sense?}

%\vvnote{There is an interesting contrast here between decoding low-degree extensions and the decoding described here. In the former, the function $f$ is arbitrary, but you are given oracle access to a low-degree extension of $f$ over a large finite field $\mathbb{F}$, $|\mathbb{F}| > d$. Here, you only have access to $f$, but $f$ is over a small finite field, in fact $\mathbb{F}_2$ (so the prior techniques don't apply). And $f$ is still guaranteed to have low degree and some other properties such as $d$-partiteness.}

%\vvnote{Hm, do you really need (strong) $d$-partiteness for this theorem to hold?}

%\vvnote{Please see the result on self-correcting polynomials over $\mathbb{F}_2$. \url{https://www.cs.princeton.edu/courses/archive/spr04/cos598B/bib/RonAKKL.pdf}} 
%\mina{seems like the idea of this paper works, and we don't need $k$-partiteness. I updated the theorem and the proof.} \anote{Very cool! :D}
\begin{proof}
Let us say that we are given an input $\vec{v}=(v_1,\ldots,v_n)$ on which we want to compute $P(\vec{v})$. All summations in this proof are taken mod $2$.  %We will start by defining some useful notation. 
%Let $part(I_i)$ denote the partition that $I_i$ belongs to. Note that $part(I_i)\in \{1,\ldots,d\}$.
We will describe the algorithm $B$ given our algorithm $A$.

%Let $f((\vec{I_1},\ldots,\vec{I_d})) = f(\vec{I})$ and $P((\vec{I_1},\ldots,\vec{I_d})) =P(\vec{I})$, where $\vec{I_j}$ is the set of variables from partition $j$.%, an abuse of notation that will help us describe inputs. 

Consider $d+1$ random vectors $\vec{y_1},\ldots,\vec{y}_{d+1}\in\{0,1\}^n$, where each bit of every $\vec{y_i}$ is taken uniformly at random. We will describe $2^{d+1}-1$ inputs $\vec{u_S}$, where $\emptyset \neq S\subseteq [d+1]$. Let 
$$
\vec{{u}_S}  = 
\vec{v} + \sum_{i\in S}\vec{y_i} 
$$
Now we define the output of $B(\vec{v})$ as
$$
B(\vec{v}) \equiv \sum_{\emptyset\neq S\subseteq{[d+1]}} A(\vec{u_S}) \pmod{2}
$$

%More specifically, for each $j=1,\ldots,d$, consider a random vector $\vec{R_j}$ of length $|I_j|$. Let $\vec{R} = (\vec{R_1},\ldots,\vec{R_d})$.
%We will describe $2^d$ inputs $\vec{J}_s$ each input is indexed by a unique string $s \in \{0,1\}^d$.  
%Let 
%$${J}_i  =  \begin{cases}
%{I}_i \oplus {r}_i, & \text{ if } s[part(I_i)]=1  \\
%{r}_i. & \text{ if } s[part(I_i)]=0
%\end{cases}$$
%and let
%$$\vec{J}_s = ({J_1},\ldots,{J_n}).$$

%Now, the output of $B(\vec{I})$ will be outputting:
%$$B(\vec{I}) \equiv \sum_{s\in \{0,1\}^n} A(\vec{J}_s) \mod 2.$$

%We will now prove two claims to get our desired result. 

\paragraph{Claim: if all $2^{d+1}$ calls to $A$ return the correct values then $B(\vec{v}) = P(\vec{v})$.}

For any polynomial $g$, integer $k$ and vectors $\vec{u},\vec{x_1},\ldots,\vec{x_{k}}\in \{0,1\}^n $, we define $T_g(\vec{u},\vec{x_1},\ldots,\vec{x_{k}})=\sum_{\emptyset\neq S\subseteq [k]} g(\vec{u}+\sum_{i\in S}\vec{x_i})$. So we have $B(\vec{v})\equiv \sum_{\emptyset\neq S\subseteq{[d+1]}} A(\vec{u_S})= \sum_{\emptyset\neq S\subseteq{[d+1]}} f(\vec{u_S})=T_f(\vec{v},\vec{y_1},\ldots,\vec{y}_{d+1})$. We want to show that $T_f(\vec{v},\vec{y_1},\ldots,\vec{y_{d+1}})=f(\vec{v})$.

Let $m$ be any polynomial with one monomial and $n$ inputs and degree $1\le k\le d$. We show that $T_m(\vec{0},\vec{y}_1,\ldots,\vec{y}_{d+1})=0$. To see this, we want the number of linear combinations $\sum_{i=1}^{d+1}b_i\vec{y}_i$ for $b_i\in \{0,1\}$, where $m(\sum_{i=1}^{d+1}b_i\vec{y}_i)\equiv 1$ to be even (each selection of $b_i$s maps to a set $S$). This number is the number of solutions of a linear system of $k$ equations in $d+1$ variables, and so is divisible by $2^{d+1-k}$, so it is even. Note that $b_i=0$ for all $i$ is not a valid answer to this linear system, and so we don't count it. %\mina{This previous paragraph is what the low degree polynomial paper proves}. 

Now let $m:=m(r_1,\ldots,r_n)=\prod_{i\in L}r_i$ be one of the monomials of $f$, for variables $r_i\in\{0,1\}$ and $L\subset [n]$ where $|L|\le d$. If we prove that $T_{m}(\vec{v},\vec{y_1},\ldots,\vec{y_{d+1}})=m(\vec{v})$, then summing these equalities for all monomials $m$, we get that $T_f(\vec{v},\vec{y_1},\ldots,\vec{y_{d+1}})=f(\vec{v})$.

To prove that $T_{m}(\vec{v},\vec{y_1},\ldots,\vec{y_{d+1}})=m(\vec{v})$, in the polynomial  $m(\vec{u}_S)=m(\vec{v}+\sum_{i\in S}\vec{y_i})$ we consider $\vec{v}$ as fixed and $\vec{y_i}$s as variables, so we can write $m(\vec{u_S})$ as follows:
$$
m(\vec{u_S}) = \sum_{L'\subseteq L} a_{L\setminus L'}m_{L'}(\sum_{i\in S}\vec{y_i}),
$$

where $m_{L'}(r_1,\ldots,r_n):=\prod_{i\in L'}r_i$, and $a_{L\setminus L'}=\prod_{i\in L\setminus L'} \vec{v}[i]$. For $L'=\emptyset$, we define $m_{L'}$ as the fixed value monomial $1$.

Now note that $\sum_{\emptyset\neq S\in [d+1]}m_{L'}(\sum_{i\in S}\vec{y_i})=T_{m_{L'}}(\vec{0},\vec{y}_1,\ldots,\vec{y}_{d+1})$ which is zero for $L'\neq \emptyset$, and is $|S|$ for $L'=\emptyset$. So
\begin{align*}
   T_{m}(\vec{v},\vec{y_1},\ldots,\vec{y_{d+1}})&\equiv \sum_{\emptyset\neq S\subseteq [d+1]} m(\vec{u}_S) \\
   &\equiv \sum_{\emptyset\neq S\subseteq [d+1]} \sum_{L'\subseteq L}a_{L\setminus L'}m_{L'}(\sum_{i\in S}\vec{y_i}) \\
   &\equiv \sum_{L'\subseteq L} a_{L\setminus L'} \sum_{\emptyset\neq S\subseteq [d+1]} m_{L'}(\sum_{i\in S}\vec{y}_i)\\
   &\equiv \sum_{L'\subseteq L} a_{L\setminus L'} T_{m_{L'}}(\vec{0},\vec{y}_1,\ldots,\vec{y}_{d+1})\\
   &\equiv |S|a_{L}\\
   &\equiv a_{L} = m(\vec{v}).
   %\sum_{S}a_{L}= (2^{d+1}=-1)m(\vec{I}) %\equiv m(\vec{I}) 
\end{align*}

Thus, $B(\vec{v}) = f(\vec{v})=P(\vec{v})$ when all of our calls to $A$ return correctly.

%%%%%%%%%%%%%%%%%%%%previous proof%%%%%%%

%Let $f$ have $m$ monomials. Let $I_{i_1}\ldots I_{i_d}$ be the $t$th monomial, where $j=part(I_{i_j})$.
%Now consider what is happening in the above sum to the monomial $I_{i_1}\ldots I_{i_d}$. 

%Consider the sum of the $t^{th}$ monomial over all possible inputs $\vec{J}_S$. If we write the $t^{th}$ monomial of $\vec{J}_s$ in terms of $I_i$ and $r_i$s, it has several monomials. Define the $deg*$ of a monomial to be the number of $I_i$ terms (Note that the number of $I_i$ terms and $r_j$ terms is $d$). A monomial of $deg*$ of $c$ will be produced by all $s$ that have $1$s in the corresponding $c$ locations. There are $2^{d-c}$ such strings. So, every monomial is summed an even number of times, except the monomial of $deg*$ of $d$ which is $I_{i_1}\ldots I_{i_d}$, which is produced once, when $s=(1,\ldots,1)$. 
%Thus, 
%$$f(I) \equiv \sum_{s\in \{0,1\}^n} f(\vec{J}_s) \mod 2.$$
%Thus, $B(\vec{I}) = P(\vec{I})$ when all of our calls to $A$ return correctly. \vvnote{Nice!}

\paragraph{Claim: the probability that all $2^{d+1}$ calls to $A$ return correct values is at least $3/4$.} First, note that for all $\emptyset \neq S$, $\vec{u}_S$ looks like each bit was chosen uniformly at random iid from $Ber[1/2]$. Of course, $\vec{u}_S$ and $\vec{u}_{S'}$ can be very correlated, but each looks iid from $Ber[1/2]$. This is true because sum of iid bits is uniformly random and so each bit of $v$ is XORed with a random value.

Second, we can use the union bound. The probability $A$ errs is at most $2^{-d-3}$. We make $2^{d+1}$ calls, so the probability that $A$ is wrong at least once is at most $2^{-2} = 1/4$. So, all of our $2^{d+1}$ calls will be correct with probability at least $3/4$.

\paragraph{Analyzing $B$.}
So, if $A$ runs in $T(n)$ time then $B$ takes $O(2^{d+1}n +2^{d+1}T(n))$ time, we need to form each of the $2^{d+1}$ inputs and we need to make $2^{d+1}$ calls to $A$. 

Second, $B$ is correct if all of its calls to $A$ give correct answers, and that happens at least $3/4$ of the time. So, as desired,
for \emph{any}  $\vec{v} \in \{0,1\}^n$:
$$Pr[B(\vec{I}) = P(\vec{v})] \geq 3/4.$$

\end{proof}

%\xxx{Larger field can use ok polys}

%\mina{The low degree polynomials idea seems not to directly work here like the $F_2$ case if we remove the $k$-partiteness. The biggest problem seems to be in the bit-manipulation part: To evaluate $f(r_1,\ldots,r_n)$ for random $z$-bit numbers $r_i$ for some $z>1$, we need $2^n$ evaluations of the $0-1$ input form.}

%\vvw{Ok, so I had some comments about this before. If $P$ has output in a finite field $F$, as was stated earlier, there is no notion of maximum output. You should rephrase this. Maybe $\log M$ is the number of bits needed to represent an element in the field?} \andrea{i hope this is better :)}
In the following, we extend Theorem~\ref{thm:framework-parity} to hold for problems that have their outputs over the integers in $[-M,M]$.
\frameworkgeneral*

\begin{proof}
Our approach is similar to Theorem \ref{thm:framework-parity}, with a few modifications. Let $M$ be the range of the fine $d$-degree polynomial $f$ of $P$, so that $f(I)\le M$ for any input instance $I$.  Let $z=\ceil{\log_2 M}$ so that $2^z\ge P(\vec{v})$ for all $\vec{v}\in \mathbb{Z}_2^n$. Note that $M\le p\cdot n^d$, so $\log_2 M\le d\log_2 n$
%First suppose that we have an average-case algorithm $C$ with input from $\mathbb{Z}_t^n$, for $t=d+z$, with prob ???. 

Suppose that we are given an input $\vec{v}=(v_1,\ldots,v_n)$ on which we want to compute $P(\vec{v})$. Let $part(v_i)$ be the partition that $v_i$ belongs to. Note that it suffices to compute $P(\vec{v})$ mod $2^z$. Consider $n$ random numbers $r_1,\ldots,r_n$, each having $t=z+d$ bits. We will describe $2^d$ inputs $\vec{u_s}$ where each input is indexed by a unique string $s=\{0,1\}^d$. 
Let 
$${u}_i  =  \begin{cases}
-{v}_i + {r}_i, & \text{ if } s[part(v_i)]=1  \\
-{v}_i - {r}_i. & \text{ if } s[part(v_i)]=0
\end{cases}$$
where $u_i$ is a $t$ bit number, so all operations are mod $2^t$. Note that since $v_i$ is a one-bit number and $r_i$ is a $t$-bit number, we consider $v_i$ as a $t$-bit number by $t-1$ zeros to the $t-1$ significant bit, so that the definition of $u_i$ makes sense. 
Let
$$\vec{u}_s = ({u_1},\ldots,{u_n}).$$

Now for $C$ that we define later, $B(\vec{v})$ will be outputting:
\begin{equation}\label{eq:defineB}
B(\vec{v}) \equiv \frac{1}{2^d}\sum_{s\in \{0,1\}^d} C(\vec{u}_s) \mod 2^z
\end{equation}

To define $C(\vec{u})$ for an input $\vec{u}=(u_1,\ldots,u_n)$ where $u_i$ is a $t$ bit  number, we do the following bit manipulation. Let $u_i[b]\in \{0,1\}$ be the $b^{th}$ bit of $u_i$, for $b\in \{0,\ldots,t-1\}$. Let 
\begin{equation}\label{eq:defineC}
C(\vec{u})=\sum_{(b_1,\ldots,b_d)\in[0,t-1]^d} 2^{b_1+\ldots+b_d}A(u_1[b_{part(u_1)}],\ldots,u_n[b_{part(u_n)}]).
\end{equation}
This completes the definition of $B$.

Note that computing $C$ needs $t^d$ calls to $A$, and we need to compute $C$ for $2^d$ inputs. So overall we have $(2t)^d$ calls to A.
%\mina{todo: prove correctness, time complexity and prob.}
\paragraph{Claim: if all $(2t)^d$ calls to $A$ return the correct value then $B(\vec{v})=P(\vec{v})$}

%Let $f=f_{even}+f_{odd}$, where $f_{even}$ has the even degree monomials of $f$ and $f_{odd}$ has the odd degree monomials of $f$.
We first prove that if $f=\sum f_\ell$ where $f_{\ell}$ is the sum of all the monomials of $f$ of degree $\ell$, then  $C(\vec{u})\equiv \sum (-1)^{d-\ell} f_\ell(\vec{u}) \pmod{2^t}$. Consider a monomial $u_{i_1}\ldots u_{i_\ell}$ in $f$, suppose that it is the $j^{th}$ monomial in $f$. For now we suppose that all our calls to $A$ output the correct value, so all the calls to $A$ on the righthand side of Equation \ref{eq:defineC} can be replaced by $f$. So consider the $j^{th}$ monomial of all the terms on the righthand side of Equation \ref{eq:defineC}. WLOG suppose that $part(u_{i_w})=w$. They are of the form $2^{b_1+\ldots+b_\ell}u_{i_1}[b_1]\ldots u_{i_\ell}[b_\ell]2^{b_{\ell+1}+\ldots+b_d}.$ So their sum is $(2^{t}-1)^{d-\ell}u_{i_1}\ldots u_{i_\ell}\equiv (-1)^{d-\ell}u_{i_1}\ldots u_{i_\ell} \pmod{2^t}$. So $C(\vec{u})\equiv \sum (-1)^{d-\ell} f_\ell(\vec{u}) \pmod{2^t}$.

Now we show that there is no monomial in $B$ with an $r_i$ variable in Equation \ref{eq:defineB}. To see this, fix some $i$. For any $s\in \{0,1\}^d$, let $s'[\ell]=s[\ell]$ for all $\ell\neq part(v_i)$, and let $s'[part(v_i)]\neq s[part(v_i)]$. Note that $s''=s$. Now if a monomial contains $r_i$ in $\vec{u}_s$ for some $s$, then $\vec{u}_{s'}$ has the same monomial but negated. 

Now consider a monomial $v_{i_1}\ldots v_{i_\ell}$ in $f(\vec{v})$. We want to find the coefficient of this monomial in $B$. This monomial is in each $C({\vec{u_s}})$ for all $s$. Since its coefficient in $f_{\ell}(\vec{u_s})$ is $(-1)^\ell$, its coefficient in $C({\vec{u_s}})$ is $(-1)^{\ell}\cdot (-1)^{d-\ell}$. So its coefficient is $1$ in $B$ mod $2^z$.

\paragraph{Claim: the probability that all $(2t)^d$ calls to $A$ return correct values is at least $3/4$.}

First note that for each $s$, the $t$ bit numbers $u_i$ look as if each of their bits is chosen uniformly at random iid from $Ber[1/2]$. Note that $\vec{u_s}$ and $\vec{u}_{\tilde{s}}$ can be very correlated, but each looks iid from $Ber[1/2]$. 

Second, we can use the union bound. The probability that $A$ errs is $1-\frac{1}{2^{d+2}(d+z)^d}$. We call $A$ $2^d(z+d)^d$ times, so the probability that $A$ is wrong at least once is at most $2^{-2} = 1/4$. So all calls will be correct with probability $3/4$.

\paragraph{Analyzing $B$.} If $A$ runs in $T(n)$ time, then $B$ takes $O(2^{d}(z+d)^d(n+T(n)))=O((2d+2\log_2M)^d(n+T(n)))$ time. Moreover, $B$ is correct if all of its calls to $A$ give correct answers, and that happens at least $3/4$ of the time. So, as desired, for \textit{any} $\vec{v}\in \{0,1\}^n$: $Pr[B(\vec{v})=P(\vec{v})]\ge 3/4$.
\end{proof}
%\xxx{Then we should discuss the general methods for increasing probabilities.}

%\mina{can add a corollary that if $f$ is a polynomial of fixed degree then $M=poly(n)$.}

%\section{Average-Case Parity SAT}
%\label{sec:avgCaseSat}
%\input{SAT}

\section{Reductions Between Factored Problems}
\label{sec:betweenProblem-bg}
In this section we discuss reductions between factored problems. These reductions help us to get hardness results for average case problems from SETH, $k$-XOR and $k$-SUM hypothesis. Dalirrooyfard et al \cite{factoredProblems} have the following generic reduction.

\begin{theorem}
\cite{factoredProblems} Let $k\ge 2$. Then we have
\begin{itemize}
\item A $\oplus$\ckfunc$(n,b,g) \rightarrow \oplus$\ckxor$(n,k^3b,g)$ reduction exists that takes time $O(ng2^{k^3b})$. 
\item A $\oplus$\ckxor$(n,b,g) \rightarrow \oplus$\ckov$(n,2k^3b,g)$ reduction exists that takes time $O(ng2^{2k^3b})$. 
\item A $\oplus$\ckxor$(n,b,g) \rightarrow \oplus$\cksum$(n,(\ceil{\lg{k}}+1)b,g)$ reduction exists that takes time $O(ng2^{b\lg{k}})$. 
\end{itemize}
\end{theorem}

In their application area, they did not care about constant factor blow ups to $b$ and $g$ in their reductions. However, we need to have the smallest blowup possible, and therefore more efficient reductions.

\subsection{More efficient reductions between factored problems}

We improve the reductions from \ckfunc[] to \ckov[], \ckxor[], and \cksum. 

\begin{theorem}
Let $k\geq 2$. We reduce \ckfunc[] with $n$ factored vectors and $g$ sets of $b$-length strings to \ckxor[], \ckov[] and \cksum[] with $n$ factored vectors and $g$ sets of strings of length $kb,2kb$ and $kb$ respectively. In fact,
%using our notation in Definition \ref{def:bgReduction}, 
we prove the following.
\begin{itemize}
\item A $\oplus$\ckfunc$(n,b,g) \rightarrow \oplus$\ckxor$(n,kb,g)$ reduction exists that takes time $O(ng2^{kb})$. 
\item A $\oplus$\ckfunc$(n,b,g) \rightarrow \oplus$\ckov$(n,2kb,g)$ reduction exists that takes time $O(ng2^{2kb})$. 
\item A $\oplus$\ckfunc$(n,b,g) \rightarrow \oplus$\cksum$(n,kb,g)$ reduction exists that takes time $O(ng2^{kb})$. 
\end{itemize}
%\virgi{Write out here what the previous reductions were so that one can easily see the improvement.}

\label{thm:completenessOfFactoredProblems}
\label{thm:ovXORsumFromAny}
\end{theorem}
\begin{proof}
Define $S_{\mathfrak{f}}^{u_1}$ to be the set of $k$-tuples of $b$-length strings, with the first string equal to $u_1$, such that they satisfy $\mathfrak{f}$. Formally, let $S_{\mathfrak{f}}^{u_1} = \{(u_1,s_2,\ldots,s_k) | \mathfrak{f}(u_1,s_2,\ldots,s_k)=1 \text{ such that } s_i \in \{0,1\}^b\forall i\in[2,k] \}$. 
Let the $k$ partitions of the \ckfunc[]~instance be $P_1,\ldots, P_k$. For each of these reductions our goal will be to have the strings of the first partition guess the full set of $k$ strings. Then the strings in the other partitions will verify these values. This will become clear once we define the reductions. Note that as long as we have a way to check the equality of $k$ separate strings pairwise simultaneously we can use this reduction technique. 

\paragraph{\ckxor[]:} We will define $k$ functions $\gamma_1, \gamma_2, \ldots, \gamma_k$ from strings of length $b$ to strings of length $bk$. These functions will help us define the \ckxor[] instance. The function $\gamma_i$ will be applied to the strings of partition $P_i$. For convenience $0^x$ is the string of $x$ zeros. Let $\bullet$ be the concatenation operator. 
\begin{align}
    \gamma_1(u_1) &= \{ u_1\bullet s_2 \bullet\ldots \bullet s_k|(u_1,s_2,\ldots, s_k) \in S_{\mathfrak{f}}^{u_1} \} \\
    \gamma_2(u_2) &= \{s_1 \bullet u_2 \bullet 0^{(k-2)b} | s_1 \in \{0,1\}^b\} \\
    \gamma_i(u_i) &= \{0^{(i-1)b} \bullet u_i  \bullet 0^{(k-i)b}\} \text{  }\forall i\in[3,k]
\end{align}
Now we define the \ckxor[] instance with factored vectors that have $g$ sets of $kb$-length strings. To do so, we define $\Gamma_i(\vec{v})$ as a function that takes as input a factored vector $\vec{v}$ with $b$ and $g$ sets and returns a factored vector with $bk$ bits and $g$ sets. Let $\Gamma_i(\vec{v})[j]$ be the $j^{th}$ set of the factored vector produced by $\Gamma_i$. We define $\Gamma_i(\vec{v})$ using the function $\gamma_i$. We define $\Gamma_i(\vec{v})[j]$ to be the set of strings $s\in \gamma_i(u)$ for all strings $u$ in $\vec{v}[j]$, the $j^{th}$ set of the factored vector $\vec{v}$. 
$$\Gamma_i(\vec{v})[j] = \{s | s\in \gamma_i(u) \text{ for all } u\in \vec{v}[j]\} .$$
We define the $i^{th}$ partition of the \ckxor[] instance by $P'_i = \{\vec{u} | \vec{u} = \Gamma_i(\vec{v}) \text{ for all } \vec{v} \in P_i\}$. Our new instance is the instance of \ckxor[]~over $P'_1,\ldots,P'_k$. 

To prove that the reduction works, suppose that $\vec{v}_1,\ldots,\vec{v}_k$ is a solution to \ckfunc[], where each $\vec{v}_i$ is a factored vector, and suppose that in this solution, the string $u_i^j\in \vec{v}_i[j]$ is chosen. Then we show that $\Gamma_1(\vec{v}_1),\ldots,\Gamma_k(\vec{v}_k)$ produces an analogous solution in the \ckxor[] instance. For each $j=1,\ldots,g$, this solution picks $u_1^j\bullet u_2^j \bullet \ldots \bullet u_k^j$ from $\Gamma_1(\vec{v}_1)$, $u_1^j\bullet u_2^j \bullet 0^{(k-2)b}$ from $\Gamma_2(\vec{v}_2)$ and $0^{(i-1)b}\bullet u_i^j \bullet 0^{(k-i)b}$ from $\Gamma_i(\vec{v}_i)$ for $i=3,\ldots,k$. From the definitions this solution exists and the strings xor to zero.

Now consider a solution in the \ckxor[] instance. So there exist factored vectors $\vec{v}_1,\vec{v}_2,\ldots, \vec{v}_k$ in the \ckfunc[] instance where $\Gamma_1(\vec{v}_1),\Gamma_2(\vec{v}_2),\ldots,\Gamma_k(\vec{v}_k)$ create a solution for the \ckxor[] instance. Suppose that in this solution, for each $i=1,\ldots,k$ and $j=1,\ldots,g$, a string in $\gamma_i(u_i^j)$ is chosen, where $u_i^j\in \vec{v}_i[j]$. Fix some $j$. We want to show that $\mathfrak{f}(u_1^j,\ldots,u_k^j)=1$, or equivalently $(u_1^j,\ldots,u_k^j)\in S^{u_1}_{\mathfrak{f}}$. The first $b$ bits of any string in $\gamma_1(u_1^j)$ is $u_1^j$, the first $b$ bits of any string in $\gamma_i(u_i^j)$ for $i>2$ is $0^b$ and so the first $b$ bits of $\gamma_2(u_2^j)$ must be $u_1^j$ so that the xor of these strings becomes zero. So the string chosen from $\gamma_2(u_2^j)$ is $u_1^j\bullet u_2^j\bullet 0^{(k-2)b}$. Now looking at the second $b$ bits of each string, the strings chosen from $\gamma_1(u_1^j)$ and $\gamma_2(u_2^j)$ have non-zero bits in those positions and so the second $b$ bits of the string from $\gamma_1(u_1^j)$ must be $u_2^j$. Similarly, looking at the $i^{th}$ $b$ bits of all the strings, the only strings that have non-zero bits are in $\gamma_i(u_i^j)$ and $\gamma_1(u_1^j)$ and so the $i^{th}$ $b$ bits of the string from $\gamma_1(u_1^j)$ must be $u_i^j$. So the string chosen from $\gamma_1(u_1^j)$ is $u_1^j\bullet u_2^j\bullet \ldots \bullet u_k^j$ and so  $(u_1^j,\ldots,u_k^j)\in S^{u_1}_{\mathfrak{f}}$. Note that previously we proved that this solution in \ckfunc[] is analogous to the solution in \ckxor[] that we started from. 

So we proved that the number of solutions in both instances is the same. 
%Note that the number of new solutions in the new version is exactly the same as in the old version. The only way for a choice of strings to xor to zero is if the first string correctly guesses the other string values. Furthermore, the strings only xor to zero if the string from  $P_i$ was in fact the $s_i$ guessed by the string in $P'_1$. We are counting all valid $k$-tuples that appeared in the original instance. 

\newcommand{\Zeta}{Z}
\paragraph{\ckov[]:} As before we will define $k$ functions $\gamma_1, \gamma_2, \ldots, \gamma_k$ from strings of length $b$ to strings of length $2bk$. The function $\gamma_i$ will be applied to the strings of partition $P_i$. We will define $\bar{s}$ to be an operator on zero-one strings that flips all the bits. So for example if $s=01101$ then $\bar{s} = 10010$. Now note that if $|s|=|s'|=b$ and the bitwise AND of $s \bullet \bar{s}$ and $\bar{s'} \bullet s'$ is the all zeros string then $s =s'$. We will use the all ones string to replace the all zeros string in the \ckxor[] reduction because $x \wedge 1 = x$, allowing the ones to not interfere with the two strings we want to compare in that location.
\begin{align}
    \gamma_1(u_1) &= \{ u_1 \bullet \bar{u}_1\bullet s_2 \bullet  \bar{s}_2 \bullet\ldots \bullet s_k \bullet \bar{s}_k|(u_1,s_2,\ldots, s_k) \in S_{\mathfrak{f}}^{u_1} \} \\
    \gamma_2(u_2) &= \{\bar{s}_1 \bullet s_1 \bullet \bar{u}_2 \bullet u_2 \bullet 1^{(k-2)2b} | s_1 \in \{0,1\}^b\} \\
    \gamma_i(u_i) &= \{1^{(i-1)2b} \bullet \bar{u}_i \bullet u_i \bullet 1^{(k-i)3b}\} \text{  }\forall i\in[3,k]
\end{align}
Similarly to the case of XOR we will define $\Gamma_i(\vec{v})$ as a function over factored vectors with $b$ and $g$ sets that returns a factored vector with $2bk$ bits and $g$ sets. Let $\Gamma_i(\vec{v})[j]$ be the $j^{th}$ set of the factored vector produced by $\Gamma_i$. Then we define
$$\Gamma_i(\vec{v})[j] = \{s | s\in \gamma_i(u) \text{ for all } u\in \vec{v}[i]\}.$$
Now we will define the $i^{th}$ partition in the \ckov[] instance as $P'_i = \{\vec{u} | \vec{u} = \Gamma_i(\vec{v}) \forall \vec{v} \in P_i\}$. Our new instance is the instance of \ckov~over $P'_1,\ldots,P'_k$. Note that the number of new solutions in the new version is exactly the same as in the old version, and it can be proven similar to the \ckxor[] case.

\paragraph{\cksum[]:} We are going to use the same basic idea here as in the two previous problems. The strings in the first partition are going to guess the whole solution to the \ckfunc[], and the other partitions confirm this guess. 
%We want to simultaneously check equality of $k$ values.  
Additionally given a string $|s|=b$ let $\nu(s)$ return the number in $[-2^{b-1},2^{b-1}-1]$ \footnote{Note that in non-factored $k$-SUM you want $k$ numbers which sum to zero. In \cksum~, and factored k-SUM more generally, the number has been split into $g$ sections of $b$ bits. We look for $k$ vectors where each of the $b$ bit parts sum to zero. As long as $\lg(k)^g = n^{o(1)}$ the factored split-up-mod version can solve the version where you have natural numbers over a larger range. (For intuition: you basically need to guess and enforce carries.) } represented by the zero one string $s$. Once again we will define functions that produce sets from a single string. Here they will be sets of numbers (which of course can be written and interpreted as strings).
\begin{align}
    \gamma_1(u_1) &= \{ \nu(u_1) + \sum_{i=2}^k 2^{b(i-1)}\nu(s_i)|(u_1,s_2,\ldots, s_k) \in S_{\mathfrak{f}}^{u_1} \} \\
    \gamma_2(u_2) &= \{\nu(s_1) + 2^b u_2 | s_1 \in \{0,1\}^b\} \\
    \gamma_i(u_i) &= \{ \nu(u_i)2^{b(i-1)}\} \text{  }\forall i\in[3,k]
\end{align}
Similarly to the case of XOR we will define $\Gamma_i(\vec{v})$ as a function over factored vectors with $b$ bits and $g$ sets that returns a factored vector with $bk$ bits and $g$ sets. Let $\Gamma_i(\vec{v})[j]$ be the $j^{th}$ set of the factored vector produced by $\Gamma_i$. Then we define
$$\Gamma_i(\vec{v})[j] = \{s | s\in \gamma_i(u) \text{ for all } u\in \vec{v}[i]\}.$$
Now we define the $i^{th}$ partition of the \cksum[] instance as $P'_i = \{\vec{u} | \vec{u} = \Gamma_i(\vec{v}) \forall \vec{v} \in P_i\}$. Our new instance is the instance of \cksum[]~over $P'_1,\ldots,P'_k$. 

Now as before, we can use these to check if we have a valid solution. Any valid solution in \ckfunc[] corresponds to exactly one solution of this new \cksum[] instance (you must line up correct values for each entry in $\gamma_1$ which correspond to a full guess of a solution to $\mathfrak{f}$. To get a total sum of zero each portion must sum to zero. 

\paragraph{Runtime.} The runtime in each case is the size of the instance produced. We prove the runtime of \ckxor[] reduction, and the rest is similar. For any $b$-bit string $u$, the set $\gamma_1(u)$ has size at most $2^{(k-1)b}$, the set $\gamma_2(u)$ has size $2^b$ and the set $\gamma_i(u)$ for $i>2$ has size $1$. In partition $P_1$ in the \ckfunc[] instance, we have $n$ factored vectors each having $g$ sets of at most $2^b$ strings of length $b$. So in total in partition $P_1$ there are at most $n2^b$ strings, and so in $P'_1$ there are at most $n2^b\cdot 2^{(k-1)b}$ strings. Similarly, in $P'_2$ there are at most $n2^{2b}$ strings and in $P'_i$ for $i>2$ there are at most $n2^b$ strings. So in total the size of the $\ckxor$ instance is $O(ng2^{bk})$.
\end{proof}

\subsection{Even More Efficient Reductions to k-SUM}
The generic reductions are much improved over the previous incarnations. However, we can do even better when specifically going from \ckov[] or \ckxor[] to \cksum[]. Here we will use the fact that the sum of $k$ numbers in $\{0,1\}$ is in $[0,k]$. The reduction from \ckxor~to \cksum~is stated in \cite{factoredProblems} but we prove it here again for completeness.

\begin{lemma}
A $\oplus$\ckov$(n,b,g) \rightarrow \oplus$\cksum[(k+1)]$(n+1,b \lceil \lg(k) \rceil,g)$ reduction exists that takes time $O(ngk2^b)$. 

A $\oplus$\ckxor$(n,b,g) \rightarrow \oplus$\cksum[(k+1)]$( n+1,b\lceil \lg(k) \rceil,g)$ reduction exists that takes time $O(ngk2^b)$. 
\label{lem:SUMreductions}
\end{lemma}
\begin{proof}
For our convenience let $c = \lceil \lg(k) \rceil$. The first part of both reductions is the same and we explain it here. Then we explain the second part of each separately. 

In both reductions we start with a $k$ partite problem with $k$ partitions containing factored vectors, $P_1,\ldots,P_k$. In both reductions we will produce new lists of factored vectors $P'_i$ by padding the original factored vectors. Here we explain what $P'_1,\ldots,P'_k$ are, and then we define $P'_{0}$ separately for when our reduction is from \ckov[] and \ckxor[].

Specifically, given a string $s$ with $b$ bits $s[0], \ldots, s[b-1]$ let $pad(s)$ be a function which returns a string of length $cb$ where $pad(s)[ci] = s[i]$ and $pad(s)[j] = 0$ if $j \not \equiv 0 \mod c$. So, we have padded the string with zeros around each number. Now let padding a factored vector be defined as running the function $F_{pad}(\vec{v})$  over a factored vector $\vec{v}$. We define this function as $\vec{u} = F_{pad}(\vec{v})$ such that for all $i\in [1,g]$
$$\vec{u}[i] = \{ pad(s) | s\in \vec{v}[i] \}.$$
That is, we pad every string that appears in every set of the factored vector. Now finally let us define $PAD(P_i)$ as
$$
PAD(P_i) \in \{F_{pad}(\vec{v}) | \vec{v} \in P_i \}.
$$
We define $P'_i=PAD(P_i)$. 
Consider $k$ zero one strings $s_1, \ldots,s_k$ and then consider the sum of the padded strings as if they were integers $x= pad(s_1) + \ldots + pad(s_k)$. Then the sum $s_1[i]+\ldots+s_k[i]$ is equal to the number represented by the bits $x[ci, ci+c-1]$. By padding our vectors we ensure that our sums are separated. Now note that with \ckov[]~we are looking for vectors where for all $i\in[0,b]$ the number represented by $x[ci, ci+c-1] <k$. Further note that for \ckxor[]~we are looking for vectors where for all $i\in[0,b]$ the number represented by $x[ci, ci+c-1] \equiv 0 \mod 2$. We are going to control this sum by creating appropriate factored vectors in partition $P'_0$. We in fact create a single factored vector with ``guesses" of all possible ``valid" sums (strings where each group of bits  $x[ci, ci+c-1]$ is one of the possible valid values). 
Below we formally define $P'_0$.

\paragraph{\ckov[]:} Let $S_{<k}$ be the set of all string representations of numbers in $[0,k-1]$.  Let $S_{concat}$ be the set of all possible concatenations of strings $s_1,\ldots,s_b$ where $s_i \in S_{<k}$. Let $S_{concat}^{(-)}$ be the set of all numbers represented in $S_{concat}$, now negated. Finally define a factored vector $\vec{u}_{OV}$ where for all $i\in[1,g]$ we have $\vec{u}_{OV}[i] = S_{concat}^{(-)}$. Now define $P'_0$ to be a partition containing only $\vec{u}_{OV}$. %For all other partitions $P_i$ for $i\in [1,g]$ we set $P'_i = PAD(P_i)$. 
Note that the number of solutions exactly match in the \ckov[] instance and the \cksum[] instance. Considering any solution to the \ckov[] instance, for each bit in the strings chosen from set $j$ of the $k$ factored vectors, the sum of these bits is less than $k$ since there is a zero among them, and so their padded sum equals to some value represented in $S_{concat}$, and so there is (exactly) one string in the $j^{th}$ set of $\vec{u}_{OV}$ whose sum with the other strings chosen from the $j^{th}$ set of the $k$ factored vectors is zero. % as it holds all possible valid sums for each subset of the strings. Any orthogonal $k$ 

Moreover, if for some $j=1,\ldots,g$, the vectors selected from the $j^{th}$ set of the $k$ factored vectors are not orthogonal, there is a bit which is one in all the $k$ strings. So the sum of the padded strings in the corresponding \cksum[] instance is $k$, and there is no string in $\vec{u}_{OV}[j]$ that can make this sum zero. %have some bit location, $i$, on which all strings are ones. Then the sum of the corresponding padded strings on bits $[ci, ci+c-1]$ will have the number $k$ represented, which does not have its negation present in the factored vector $\vec{u}_{OV}$.

\paragraph{\ckxor[]:} Let $S_{even}$ be the set of all string representations of even numbers in $[0,k]$.  Let $S_{concat}$ be the set of all possible concatenations of strings $s_1,\ldots,s_b$ where $s_i \in S_{even}$. Let $S_{concat}^{(-)}$ be the set of all numbers represented in $S_{concat}$, now negated. Finally define a factored vector $\vec{u}_{XOR}$ where for all $i\in[1,g]$ we have $\vec{u}_{XOR}[i] = S_{concat}^{(-)}$. Now define $P'_0$ to be a partition containing only $\vec{u}_{XOR}$. %For all other partitions $P_i$ for $i\in [1,g]$ we set $P'_i = PAD(P_i)$. 
Similar as above, we can see that the number of solutions exactly match. %For any solution within the original their padded strings sum to some value represented in $\vec{u}_{XOR}$, as it holds all possible valid sums for each subset of the strings. Any invalid choices have some bit location, $i$, on which there are an odd number of bits. Then the sum of the corresponding padded strings on bits $[ci, ci+c-1]$ will have an odd number represented, which does not have its negation present in the factored vector $\vec{u}_{XOR}$.
\end{proof}

\section{Worst-Case to Average-Case Fine-Grained Problems}
\label{sec:WCACxor}
%\xxx{Mina:  The figure should be updated to include all reductions.}
%\xxx{Mina: rename the subsection with steps: define steps of the reduction first using the figure, and the once it is explained in high level, we go into details in every subsection.}

In this section we will show that we can get \emph{average-case} lower bounds for some of the most important problems in fine-grain complexity from worst-case hypotheses. The distribution on which these problems are hard is not the uniform distribution, but an easy to sample distribution. 

Our high level approach is as follows. We want to derive hardness for average-case $K$-OV, $K$-XOR and $K$-SUM of size $N$. 
We want to base the hardness on $k$-OV, $k$-XOR or $k$-SUM hypothesis, for some $k$ as a function of $K$, and we use Lemma \ref{lem:decisionToParity} which states that under $k$-OV, $k$-XOR and $k$-SUM hypothesis, $\oplus k$-OV, $\oplus k$-XOR and $\oplus k$-SUM are hard respectively. Suppose that the starting problem is $\oplus k$-$P$ and the problem that we seek average case hardness for is $\oplus K$-$Q$ (for example $P$ can be OV and $Q$ can be SUM. They can be the same problem as well). We do the following steps to reduce worst-case $\oplus k$-$P$ to average-case $\oplus K$-$Q$. See Figure \ref{fig:roadmapsection6}.
\begin{enumerate}
    \item We reduce $\oplus k$-$P$ to factored $\oplus k$-$P$ (Theorem \ref{lem:selfReductionsWithbg}).
    \item We reduce factored $\oplus k$-$P$ to factored $\oplus k$-$Q$ (Section  \ref{sec:betweenProblem-bg}).
    \item We reduce worst case factored $\oplus k$-$Q$ to average case factored $\oplus k$-$Q$ on uniform distribution (Theorem \ref{thm:frameworkImprovement}).
    \item Finally we reduce factored $\oplus k$-$Q$ to $\oplus K$-$Q$. (Theorem \ref{thm:factoredToUnfactored}).
\end{enumerate}

Note that if the problems $P$ and $Q$ are the same, we don't need step $2$. In the following subsections we introduce the tools used for each step, and then we put them together to prove our results. 

%As step $0$, we recall a reduction from rETH to $k$-OV. This reduction helps us base some of our hardness results on rETH.  
%If the starting or ending problems are the same problem $f$, say both are the orthogonal vectors problem, we do the following: We show first give a reduction from $f$ to \ckfunc. Then we  

\begin{figure}
    \centering
    \includegraphics[width=\linewidth]{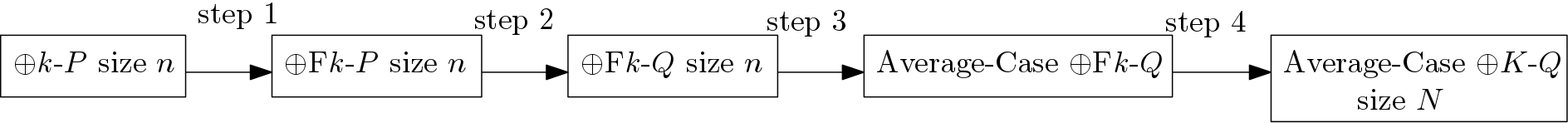}
    \caption{%\mina{todo: add parity sign to all }\andrea{done}
    Roadmap of the approach to prove average-case complexity for the $K$-$Q$ problem from $k$-$P$ problem. Each problem is $k$ (or $K$) partite, and the size of the problems refer to the number of vectors or factored vectors in each partition.}
    \label{fig:roadmapsection6}
\end{figure}

%We give the first example of an average-case distribution on which we can prove that $k$-XOR, $k$-OV, and $k$-SUM are super-linearly hard from worst-case hypotheses. 

%\begin{figure}[h]
%\centering
%\includegraphics[width=\textwidth]{Figures/AvgCaseHardOV-SUM-XOR.png}
%\caption{Reductions from worst-case FGC hypotheses to average-case parity $K$-OV, $K$-SUM, and $K$-XOR. Dashed lines are reductions from prior work. Thick lines are reductions from this paper. We shorten average-case to AC in the figure. We use $\oplus$ to signify that it is the variant which counts the parity of the number of solutions.}
%\end{figure}

%\subsection{How Hard is OV with small dimension?}

%Our reduction will have an overhead that grows exponentially in the size of the vectors we are using. So, we would like to use low dimension vectors for our problems. For worst case $k$-XOR and $k$-SUM one can use hashing to show that vectors or numbers made up of $k\lg(n)$ bits are hard instances. For $k$-OV we have to give up something in the exponent, but, we can show non-trivial hardness for vectors of small dimension. 

%We can get a weaker guarantee from rETH. We mention this because this weaker guarantee is often sufficient for us.

\subsection{Step 1: Un-factored to factored reduction}
Dalirrooyfard et al \cite{factoredProblems} show that one can reduce any problem to its factored version with only constant blowups in the size of the problem. 

\begin{theorem}[\cite{factoredProblems}]
In $O(n)$ time, one can reduce an instance of size $n$ of $k$-OV, $k$-XOR, $k$-SUM and Z$k$C to a single call to an instance of size $\tO(n)$  of \ckov[], \ckxor[], \cksum[]~and \czkc~respectively.
\label{thm:unfactoredToFactoredOld}
\end{theorem}

A brief explanation of the proof of Theorem \ref{thm:unfactoredToFactoredOld} is as follows. The reduction works by splitting up each vector/number into $g$ pieces of length $b$ bits. Then, from each original vector, we make one factored vector that has $g$ sets each containing exactly one vector (the $g$ vector pieces from the original). 
So in fact, if the dimension of the original vector is $d$ then $b\cdot g=d$. For $k$-SUM if the numbers were $d$ bits long then $b \cdot g =d$. This means that $\oplus k$-OV, $\oplus k$-XOR, and $\oplus k$-SUM reduce to a single call of size $n$ of $\oplus$\ckov[], $\oplus$\ckxor[], and $\oplus$\cksum[] respectively.

Using rETH, $k$-XOR and $k$-SUM hypothesis, the following lemma which is inferred from \cite{factoredProblems} specifies the parameters for which factored versions of $k$-OV, $k$-XOR and $k$-SUM are hard under the corresponding hypothesis.

%The next lemma \mina{which is implicitly mentioned in \cite{factoredProblems}?} allows us to connect the worst-case hardness of SETH, $k$-XOR and $k$-SUM to the average-case hardness of $\oplus$\ckov, $\oplus$\ckxor, and $\oplus$\cksum, respectively. This will help us in our later reductions. We will give reductions from $\oplus$\ckov to $k'$-OV, from $\oplus$\ckxor to $k'$-XOR, and $\oplus$\cksum~to $k'$-SUM, where $k'$ is going to be a function of $k$ and other parameters of the problem. To generalize and strengthen our results we will also use the completeness of  $\oplus$\ckov, $\oplus$\ckxor, and $\oplus$\cksum~to give (less tight) average-case hardness for each of these problems from all of the above hypotheses. 

\begin{lemma}
We have the following hardness results from rETH, $k$-XOR and $k$-SUM.
\begin{enumerate}
    %\item\label{item:1} There exists a fixed constant $c_\eps$  such that  the $\oplus$\ckov$(n,b,g)$~problem where $bg = c_\epsilon k \lg(n)$ and $kg=o(\lg(n))$ requires $n^{k(1-\epsilon)-o(1)}$ time if SETH is true. 

\item \label{item:2}There exists a fixed constant $c$ such that  the $\oplus$\ckov$(n,b,g)$~problem where $bg = c k \lg(n)$ and $kg=o(\lg(n))$ requires $n^{\Omega(k)}$ time if rETH is true. 

\item \label{item:3}The  $\oplus$\ckxor$(n,b,g)$~problem requires $n^{\lceil k/2 \rceil-o(1)}$ time when $bg \ge k\lg(n)+2$ and $kg=o(\lg(n))$ if the $k$-XOR hypothesis holds.

\item \label{item:4}The $\oplus$\cksum$(n,b,g)$~problem requires $n^{\lceil k/2 \rceil-o(1)}$ time when $bg \ge k\lg(n)+2$ and $kg=o(\lg(n))$ if the $k$-SUM hypothesis holds. 
\end{enumerate}

\label{lem:selfReductionsWithbg}
\end{lemma}
\begin{proof}
%Item \ref{item:1} directly results from Lemma \ref{lem:dimensionReductionSETH} and Theorem \ref{thm:unfactoredToFactoredOld}.% we have that the worst case $\oplus$\ckov~problem requires $n^{k(1-\epsilon)-o(1)}$ time if $bg = c_\epsilon k \lg(n)$ for some fixed constant $c_\epsilon$ if SETH holds.
Item \ref{item:2} directly results from Lemma \ref{lem:dimensionReductionETH} and Theorem \ref{thm:unfactoredToFactoredOld}.% to show that we the worst case $\oplus$\ckov~problem requires $n^{\Omega(k)}$ time if $bg = c  \lg(n)$ for some fixed constant $cn$ if rETH holds.

Using Lemma \ref{lem:lengthForkSumkXOR} and Theorem \ref{thm:unfactoredToFactoredOld} we have that the worst case $\oplus$\ckxor~and $\oplus$\cksum~problems require $n^{\lceil k/2 \rceil-o(1)}$ time when $bg \ge k\lg(n)+2$ if the $k$-XOR hypothesis or the $k$-SUM hypothesis hold respectively, and hence we get items \ref{item:3} and \ref{item:4}. %Next we use the self reduction from Theorem \ref{thm:frameworkImprovement}. The uniform average-cases of $\oplus$\ckxor~and $\oplus$\cksum~both require $n^{\lceil k/2 \rceil-o(1)}2^{-kg}$ time for algorithms that succeed with probability $1-\frac{2^{-b}}{3}$ when $bg = k\lg(n)$ if the $k$-XOR hypothesis or the $k$-SUM hypothesis hold respectively. 
 %Next we use the self reduction from Theorem \ref{thm:frameworkImprovement} to say that the uniform average-case of the $\oplus$\ckov~problem requires $n^{k-o(1)}2^{-kg}$ time if $bg = c_\epsilon  \lg(n)$ for some fixed constant $c_\epsilon$ if SETH holds.
\end{proof}

%We will use the notion of a factored problem and the fact that factored problems have a  worst-case to average-case self reduction as an intermediate step in our reductions. We will use some lemmas and theorems from previous work to do this. 

%\mina{changing the citation to the framework improvement..}

% Andrea: I added this little stub just so we don't skip a step for the reader. Even though this is very handled. 
\subsection{Step 2: Transferring Between Problems}
If the problems $P$ and $Q$ are the same there is no need for this step. If they are not the same problem then we will use the results from Section \ref{sec:betweenProblem-bg} to transfer between the problems. 

\subsection{Step 3: Worst-case to average-case reduction of factored problems}
Lemma \ref{lem:factored-gdlp} states that for any problem \ckfunc[]$(n,b,g)$, there is a good $kg$-degree polynomial. Then from Theorem \ref{thm:framework-parity}, we have the following hardness result from worst case \ckfunc[] to average case \ckfunc. 

\begin{theorem}%[Framework Improvement \cite{personalComDLVW}]
Let $P$ be any factored problem $\oplus$\ckfunc$(n,b,g)$~with $n$ factored vectors made up of $g$ subsets of $\{0,1\}^b$. 
Then an algorithm for $P$ on the uniform average-case that runs in time $T(n)$ and succeeds with probability $1-\frac{2^{-kg}}{8}$ implies a worst case randomized algorithm that succeeds with probability $3/4$ that runs in $2^{kg}T(n)$ time. 
\label{thm:frameworkImprovement}
\end{theorem}

\subsection{Step 4: Reduction from Factored to Un-factored Versions}
\label{sec:step4}
In this section we are going to present a reduction from factored problems to their un-factored versions. Each factored vector has $g$ sets of $b$-bit strings, so we are going to treat these $g$ sets as additional partitions, and hence we are going to have $gk$ partitions. We are going to make these $b$-bit strings longer, to encode which factored vector they are coming from. Hence we are going to represent a single factored vector with $g$ subsets of $\{0,1\}^b$ with $g2^b$ vectors of length $k\cdot g \cdot \lg(n)+bg$, in $g$ new partitions.  
%we are going to represent a single factored vector with $g$ subsets of $\{0,1\}^b$ with the set of all (for example) $g$-XORs of at most $g2^b$ vectors of length $k\cdot g \cdot \lg(n)+bg$. Notably, the original factored vector represents a set of vectors formed by concatenating one sub-vector from each of the $g$ sets. The set of all  (for example) $g$-XORs of a set of vectors can serve the same purpose. 
%Basically, we set aside $k\lg(n)$ bits to check that we are making a valid sum of vectors (that you are picking the same factored vector for each part).

The reductions from factored to un-factored versions of problems will be quite inefficient. However, even these inefficient reductions will give us meaningful new lower bounds on these problems. In some sense the key insight of this reduction is that factored problems present a lossy way to re-write our problems as low degree polynomials. %\mina{adding some positive points: these are giving us (probably) the first worst case to average case reductions for many of these problems?} \xxx{Yes! I think so. I think we didn't previously have a worst-case to average-case reduction for any of these problems. :D}

\begin{restatable}{theorem}{factoredtounfactored}
We have the following reductions from factored problems to their un-factored versions.
\begin{itemize}
    \item An instance of \ckxor$(n,b,g)$~%with factored vectors having $g$ sets of $b$ bits 
    can be turned into one instance of $kg$-XOR with $k2^bn$ vectors of length $bg+(k-1)g\lg(n)$. 
    \item An instance of \ckov$(n,b,g)$%~with $g$ groups of $b$ bits it 
    can be turned into one instance of $kg$-OV with $k2^bn$ vectors of length $bg+2(k-1)g\lg(n)$. 
    \item An instance of \cksum$(n,b,g)$~%with $g$ groups of $b$ bits it 
    can be turned into one instance of $kg$-SUM with $k2^bn$ vectors of length $(b+\lg(k))g+2(k-1)g(\lg(n)+\lg(k))$.
\end{itemize}
\label{thm:factoredToUnfactored}
\end{restatable}

\begin{proof}
The idea of these reductions is to generate an instance with $kg$ partitions. Each group of $g$ partitions $P'_{gj}, \ldots, P'_{gj+g-1}$ will represent a single partition $P_j$ from the original problem. Each factored vector from the original problem will be represented with at most $g2^b$ vectors, at most $2^b$ vectors in each of the $g$ partitions associated to the partition this vector is from in the original problem.
We will use $0^x$ and $1^x$ to refer to strings of length $x$ of all zeros and all ones respectively. We will use $\bullet$ to mean concatenation. Let $\nu_n(\ell)$ be a function from integers $\ell \in [0,n-1]$ to the zero one string indicating that number. For example, $\nu_{4}(3) =\text{`11'}$.

\paragraph{\ckxor~to $gk$-XOR}
First we will define a function $\gamma_{xor}(\vec{v}, \ell ,i,j)$ which takes a single factored vector $\vec{v}$, the index of $\vec{v}$ in its partition $\ell$, an index $i\in [1,g]$, and the partition index $j$. The function $\gamma_{xor}(\vec{v}, \ell ,i,j)$ is going to produce a set of strings, where each string in this set is associated to a string $s\in \vec{v}[i]$. The output strings will be having two sections. The first section is a validity check. The reason for the validity check is as follows: in the $gk$-XOR instance, we are going to select one vector from each of the $kg$ partitions, and as said above, partitions $P'_{g(j-1)+1}, \ldots, P'_{g(j-1)+g}$ will represent partition $P_j$ in the original instance. So the vectors chosen from each of these $g$ partitions must be from the same factored vector in order for the reduction to work. The validity check is going to enforce this property.  %The goal is to make sure that all of the strings are from the same factored vector. 
More formally, the validity check ensures that strings chosen from $P'_{g(j-1)+i}$ and $P'_{g(j-1)+i+1}$ are from the same factored vector. We are going to have $k(g-1)$ validity checks, one for each $i=1,\ldots,g-1$, and to ensure no overlaps we ``separate" all $k(g-1)$ validity checks by putting them in a unique position in the string.

The second section of each output string is intended to encode the string $s\in \vec{v}[i]$ it is associated to. %to help produce all the strings implied by the factored vector.
We encode $s$ in a way that when we consider all possible strings formed by xoring one string from each set $\gamma_{xor}(\vec{v}, \ell ,i,j)$ for all $i\in[1,g]$ the second parts of the strings will capture all the strings represented by the factored vector $\vec{v}$.

Let $H_j=0^{\lg(n)(g-1)(j-1)}$ and let $T_j=0^{\lg(n)(g-1) (k-j)}$, these are the zeros that separate the validity checks from each other. Let $\overline{\nu(\ell)}$ be the bitwise negation of the string $\nu(\ell)$.  
\begin{align}
\gamma_{xor}(\vec{v},\ell, 1,j) =& \{H_j \bullet  \nu_n(\ell) \bullet 0^{\lg(n)(g-2)} \bullet T_j &\bullet & s \bullet 0^{b(g-1)}&|&s\in \vec{v}[1] \}\\
\gamma_{xor}(\vec{v}, \ell, i,j) =& \{ H_j \bullet 0^{\lg(n)(i-2)} \bullet \overline{\nu_n(\ell)} \bullet  \nu_n(\ell) \bullet 0^{\lg(n)(g-i-1)} \bullet T_j  &\bullet & 0^{b(i-1)} \bullet s \bullet 0^{b(g-i)}\hspace{-0.1in}&|&s\in \vec{v}[i] \} \\
\gamma_{xor}(\vec{v},\ell, g,j) =&  \{H_j \bullet 0^{\lg(n)(g-2)} \bullet \nu_n(\ell) \bullet T_j &\bullet & 0^{b(g-1)} \bullet s&|&s\in \vec{v}[g] \}
\end{align}
Now we can define our new sets $P'_{g(j-1)+i}$ for $j=1,\ldots,k$: %\mina{This is supposed to be union right?} \xxx{Great catch!}
$$P'_{g(j-1)+i} = \bigcup_{\vec{v_\ell} \in P[j]} \gamma_{xor}(\vec{v}_\ell,\ell, i ,j).$$
So each new partition $P'_{g(j-1)+i}$ is the union of all of the sets representing the $i^{th}$ groups of strings from factored vectors in partition $P_j$. Now we show why the reduction works.

First suppose that we pick $g$ strings $s_{g(j-1)+i}$ from $P'_{g(j-1)+i}$ for $i\in[1,g]$. Then for $i=1,\ldots,g-1$ the bits from $\lg(n)(g-1)(j-1)+\lg(n)(i-1)+1$ to $\lg(n)(g-1)(j-1)+\lg(n)i+1$ are zero for all strings except the strings chosen from $P'_{gj+i}$ and $P'_{gj+i+1}$. In order for the xor of these strings to be zero, it must be that these two strings are both having the same value of $\ell$, so they must be from the same factored vector. 
So the bits from $[\lg(n)(g-1)(j-1) , \lg(n)(g-1)j]$ will XOR to zero iff all of these strings were generated with the same value of $\ell$. So, a choice of $gk$ strings will only XOR to the all zeros string on the bits $[0,\lg(n)(g-1)k]$ iff for all $j$ the strings $s_{g(j-1)+i}$ from $P'_{g(j-1)+i}$ were generated with the same $\ell$. Note that we select a unique $\ell$ for each factored vector in each partition. This shows that each choice of $kg$ vectors in the un-factored instance is corresponding to a choice of $k$ factored vectors (with a choice of $b$-bit strings) in the factored instance. From the definition of the un-factored version it can be seen that each choice of $k$ factored vectors (with a choice of $b$-bit strings) corresponds to a choice of $kg$ vectors in the un-factored instance, and essentially these two correspondences are the same.   

Now it is easy to see that the XOR of these $gk$ vectors is zero if and only if the XOR of the corresponding factored vectors with the corresponding choices of $b$-bit strings in the factored instance is zero. Note that each factored vector is choosing one vector in each of its $g$ sets, so each factored vector in the factored instance is representing a $bg$-bit string in the solution. Then $bg$ last bits of the XOR of strings chosen from $P'_{g(j-1)+i}$ for $i=1,\ldots,g$ creates this $gb$-bit vector that the factored vector from partition $j$ chooses in the solution.

%To see this more formally, consider a choice of $k$ factored vectors $\vec{v}^{(1)}, \ldots, \vec{v}^{(k)}$ where $\vec{v}^{(i)}\in P_i$. Consider a $g$ tuple, $t$ of $k$ tuples where $t[i] = (s^{(1)}_i, \ldots, s^{(k)}_i)$ for all $i\in[1,g]$ and $s^{(j)}_i \in \vec{v}^{(j)}[i]$. Now note that each string $s^{(j)}_i$ has a new string representing it in the new problem, $\gamma_{xor}(\vec{v}^{(j)}, \ell, i, j)$ (where $\ell$ is the index of $\vec{v}^{(j)}$ in partition $P_j$). Furthermore, note that the $gk$ strings represented by  $\gamma_{xor}(\vec{v}^{(j)}, \ell, i, j)$ XOR to zero iff all of the $k$ tuples of strings $(s^{(1)}_i, \ldots, s^{(k)}_i)$ XOR to zero. This gives us our goal. 

Since each solution in the factored instance is corresponding to a unique solution in the un-factored instance, the number of solutions to the new $(gk)$-XOR problem is equal to the number of solutions to the original \ckxor~problem.

\paragraph{\ckov~to $gk$-OV} We are taking the same idea as above, but, making it work for for bitwise AND instead of bitwise XOR. First we will define a function $\gamma_{OV}(\vec{v}, \ell ,i,j)$ which takes a factored vector $\vec{v}$, the index $\ell$ of that factored vector in its partition, an index $i\in [1,g]$, and the partition index $j$. Redefine  $H_j=1^{2\lg(n)(g-1)(j-1)}$ and let $T_j=1^{2\lg(n)(g-1) (k-j)}$, these are the ones that separate the validity checks from each other. Let $\lambda_n(\ell) = \overline{\nu(\ell)} 
\bullet \nu(\ell)$ and let $\overline{\lambda_n(\ell)} = \nu(\ell) \bullet \overline{\nu(\ell)}$.

\begin{align}
\gamma_{OV}(\vec{v},\ell, 1,j) =& \{H_j \bullet  \lambda_n(\ell) \bullet 1^{2\lg(n)(g-2)} \bullet T_j &\bullet & s \bullet 1^{b(g-1)}&|&s\in \vec{v}[1] \}\\
\gamma_{OV}(\vec{v}, \ell, i,j) =& \{ H_j \bullet 1^{2\lg(n)(i-2)} \bullet \overline{\lambda_n(\ell)} \bullet  \lambda_n(\ell) \bullet 1^{2\lg(n)(g-i-1)} \bullet T_j  &\bullet & 1^{b(i-1)} \bullet s \bullet 1^{b(g-i)}\hspace{-0.1in}&|&s\in \vec{v}[i] \} \\
\gamma_{OV}(\vec{v},\ell, g,j) =&  \{H_j \bullet 1^{2\lg(n)(g-2)} \bullet \overline{\lambda_n(\ell)} \bullet T_j &\bullet & 1^{b(g-1)} \bullet s&|&s\in \vec{v}[g] \}
\end{align}

Now we can define our new sets $P'_{g(j-1)+i}$:
$$P'_{g(j-1)+i} = \bigcup_{\vec{v_\ell} \in P[j]} \gamma_{OV}(\vec{v}_\ell,\ell, i ,j).$$
So each new partition $P'_{g(j-1)+i}$ is the union of all of the sets representing the $i^{th}$ groups of strings from factored vectors in partition $P_j$. As in the $k$-XOR reduction, our validity checks in the first half of the strings validate that we have selected $k$ factored vectors, so each orthogonal $gk$-tuple of vectors correspond to $k$-orthogonal factored vectors in the factored instance with a choice of $b$-bit strings in each of their $g$ sets. Moreover, each choice of $k$ factored vectors with a choice of $b$-bit strings corresponds to a choice $kg$ vectors in the un-factored instance. Now the second part of the strings defined in $\gamma_{OV}$ shows that $k$ factored vectors with a choice of $b$-bit strings are orthogonal if and only if their corresponding $kg$ vectors in the un-factored set are orthogonal. To see this, note that each factored vector is choosing one vector in each of its $g$ sets, so each factored vector in the factored insance is represenging a $bg$-bit string. In the un-factored instance, the last $bg$ bits of the bitwise AND of the strings chosen from $P'_{g(j-1)+i}$ for $i=1,\ldots,g$ create this $bg$ bit vector that the factored vector from partition $j$ represents.  
%Now consider a choice of $k$ factored vectors $\vec{v}^{(1)}, \ldots, \vec{v}^{(k)}$ where $\vec{v}^{(i)}\in P_i$. Now consider a $g$ tuple, $t$ of $k$ tuples where $t[i] = (s^{(1)}_i, \ldots, s^{(k)}_i)$ for all $i\in[1,g]$ and $s^{(j)}_i \in \vec{v}^{(j)}[i]$. Now note that each string $s^{(j)}_i$ has a new string representing it in the new problem, $\gamma_{OV}(\vec{v}^{(j)}, \ell, i, j)$ (where $\ell$ is the index of $\vec{v}^{(j)}$ in partition $P_j$). Furthermore, note that the $gk$ strings represented by  $\gamma_{OV}(\vec{v}^{(j)}, \ell, i, j)$ are orthogonal iff all of the $k$ tuples of strings $(s^{(1)}_i, \ldots, s^{(k)}_i)$ XOR to zero. This gives us our goal. 
So the number of solutions to the new $(gk)$-OV problem is equal to the number of solutions to the original \ckov~problem.

\paragraph{\cksum~to $gk$-SUM}
We are taking the same idea as above, but, making it work for addition, instead of bitwise AND or XOR. Let $X=2^{b+\lg(k)}$, $Y = 2^{(\lg(n)+\lg(k))}$, and $Z=Y^{(g-1)}$.
\begin{align}
\gamma_{SUM}(\vec{v},\ell, 1,j) =& \{X^gZ^{j-1}(\ell) &+& s&|&s\in \vec{v}[1] \}\\
\gamma_{SUM}(\vec{v}, \ell, i,j) =& \{X^gZ^{j-1}(Y^{i-1} \ell -Y^{i-2}\ell) \hspace{-1in} &+& s X^{i-1} &|&s\in \vec{v}[i] \} \\
\gamma_{SUM}(\vec{v},\ell, g,j) =&  \{-X^gZ^{j-1}Y^{g-2}\ell &+& s X^{g-1} &|&s\in \vec{v}[g] \}
\end{align}
Now we can define our new sets $P'_{g(j-1)+i}$:
$$P'_{g(j-1)x+i} = \bigcup_{\vec{v_\ell} \in P[j]} \gamma_{SUM}(\vec{v}_\ell,\ell, i ,j).$$
So each new partition $P'_{g(j-1)+i}$ is the union of all of the sets representing the $i^{th}$ groups of strings from factored vectors in partition $P_j$. As in the $k$-XOR and $k$-OV reductions, our validity checks in the first half of the strings validate that we have selected $k$ factored vectors. Note that the part of the string that encodes $\ell$ is multiplied by a large power of $2$, to make it separate from the second part of the string. Note that again the last $bg$ bits of each string associated to a $b$-bit string $s$ are meant to encode $s$. Similar as previous cases, it can be seen that the solutions in the factored instance and un-factored instance correspond to each other, and hence 
%Now consider a choice of $k$ factored vectors $\vec{v}^{(1)}, \ldots, \vec{v}^{(k)}$ where $\vec{v}^{(i)}\in P_i$. Now consider a $g$ tuple, $t$ of $k$ tuples where $t[i] = (s^{(1)}_i, \ldots, s^{(k)}_i)$ for all $i\in[1,g]$ and $s^{(j)}_i \in \vec{v}^{(j)}[i]$. Now note that each number $s^{(j)}_i$ has a new number representing it in the new problem, $\gamma_{SUM}(\vec{v}^{(j)}, \ell, i, j)$ (where $\ell$ is the index of $\vec{v}^{(j)}$ in partition $P_j$). Furthermore, note that the $gk$ numbers represented by  $\gamma_{SUM}(\vec{v}^{(j)}, \ell, i, j)$ sum to zero iff all of the $k$ tuples of numbers $(s^{(1)}_i, \ldots, s^{(k)}_i)$ sum to zero. This gives us our goal. 
the number of solutions to the new $(gk)$-SUM problem is equal to the number of solutions to the original \cksum~problem.
\end{proof}

\subsection{Average-case hardness of $\oplus K$-OV}\label{sec:ac-kov}
%\mina{add parity sign to everything.}
We are going to prove Theorem \ref{thm:ovxorhard} for $\oplus K$-OV. We want to show that for any $K$ and $N$, $\oplus K$-OV of size $N$ is hard over some distribution. We start from step 4 and work our way back to step 1, figuring out the necessary parameters based on $K$ and $N$. See Figure \ref{fig:ov_map}.

\begin{figure}
    \centering
    \includegraphics[width=\linewidth]{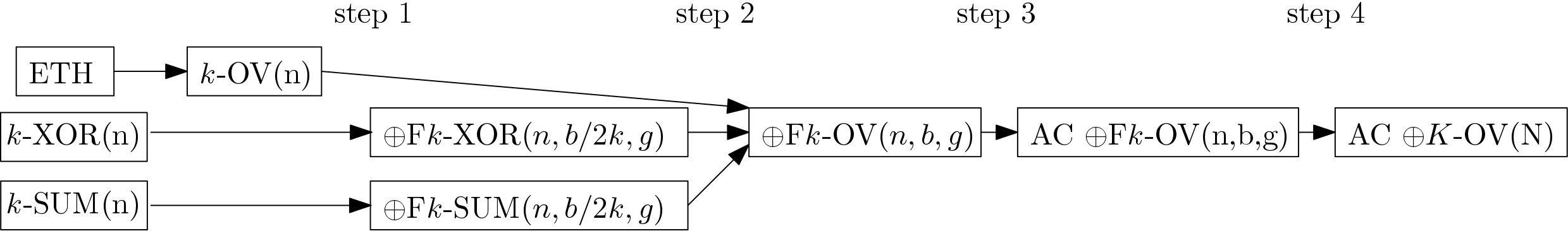}
    \caption{Reductions to average case $K$-OV of size $N$. The size of the unfactored problems is mentioned as a parameter in front of them. To see what the values of $b,g,n$ and $k$ are in terms of $K$ and $N$ see Table \ref{tab:ov-reduction-param-values}.}
    \label{fig:ov_map}
\end{figure}

\begin{table}
    \centering
    \begin{tabular}{|c|c|c|c|}
    \hline
         & $\oplus k$-OV($n$)& $\oplus k$-XOR($n$)& $\oplus k$-SUM($n$) \\ \hline
        $b$ & $\lg{N}$& $ \lg{N}$& $ \lg{N}$\\ \hline
        $g$ & $\sqrt{K}$ & $K^{2/3}$ & $K^{2/3}$\\ \hline
        $n$ & $\sqrt{N}$& $\sqrt{N}$& $\sqrt{N}$\\ \hline
        $k$ & $\sqrt{K}$& $K^{1/3}$&$K^{1/3}$ \\ \hline
    \end{tabular}
    \caption{parameter values for reductions from $k$-OV, $k$-XOR and $k$-SUM to average case $K$-OV, where the starting problem is of size $n$ and $K$-OV is of size $N$. The exact values are within constant factor away from the values mentioned in the table.}
    \label{tab:ov-reduction-param-values}
\end{table}

We define the distribution using Theorem \ref{thm:factoredToUnfactored}. Choosing appropriate parameters $b$ and $g$, this Theorem gives a reduction from $\oplus$\ckov$(\frac{Ng}{2^bK},b,g)$ to $\oplus K$-OV of size $N$ with vectors of dimension $bg+(k-1)g\lg{n}$ where $k:=K/g$ (step 4).
Starting from a uniform distribution on $\oplus$\ckov$(\frac{Ng}{2^bK},b,g)$, we get a distribution for $\oplus K$-OV through this reduction. We call this distribution $D_{OV}(N,K,b,g)$.

Let $n:=\frac{Ng}{2^bK}$. To show that $\oplus K$-OV is hard on this distribution with some success probability $q$, we have to show that $\oplus$\ckov$(n,b,g)$ is hard on uniform distribution with success probability $q$ and we will derive the appropriate value for $g$ and $b$.
We show multiple hardness for average case $\oplus$\ckov$(n,b,g)$ under different hypothesis. 

First note that by Theorem \ref{thm:frameworkImprovement} if worst-case $\oplus$\ckov$(n,b,g)$ requires $T(n)$ time, then uniform average-case $\oplus$\ckov$(n,b,g)$ with success probability $q=1-2^{-kg}/8=1-2^{-K}/8$ requires $T(n)/2^{kg}=T(n)/2^{K}$ (step 3). So we have to find a lower bound for $\oplus$\ckov$(n,b,g)$ in the worst case. 
First we show hardness under rETH.

\begin{theorem}
Let $K$ be a constant. Under rETH, any algorithm that solves average case $\oplus K$-OV of size $N$ with vectors of dimension $\Theta(K\lg{N})$ with probability $1-\frac{2^{-K}}{8}$ where the input is drawn from $D_{OV}(N,K,O(\lg{N}),\sqrt{K})$ distribution  requires $N^{\Omega(\sqrt{K})}$ time.
\label{thm:seth-to-ov}
\end{theorem}

\begin{proof}
By Theorem \ref{lem:selfReductionsWithbg} (step 1), under rETH, $\oplus$\ckov$(n,b,g)$ requires $n^{\Omega{(k)}}$ if 
\begin{enumerate}
    \item $bg>ck\lg{n}$ 
    \item $kg = o(\lg{n})$
\end{enumerate}

Note that the second constraint is equivalent to $K=o(\lg{n})$.
We are going to chose the value of $g$ and $b$ as follows and then we argue why these values give us the best bound we can get. Let $b=2c\lg{n}-0.5\lg{K}>c\lg{n}$, and let $g=\sqrt{K}$. We show that with this choice of parameters the constraints are satisfied: The first constraint is equivalent to $bg^2>c K\lg{n}$ which is clearly satisfied. We have that $n=\frac{Ng}{2^bK}=\frac{N\sqrt{K}}{Kn^{2c}/\sqrt{K}}=N/n^{2c}$. So $N=n^{1+2c}$. Since $K$ is a constant with respect to $N$ and $\lg{N}=O(\lg{n})$, the second constraint is satisfied as well. Note that in this case the dimension of vectors in the $K$-OV instance is $bg+(k-1)g\lg{n}=\Theta(K\lg{n})$. So from rETH we get that average case $\oplus K$-OV of size $N$ requires $N^{\Omega(\sqrt{K}/(1+2c))}=N^{\Omega(\sqrt{K})}$.

Now we show that our choice of parameters is optimal in the sense that in the exponent of $N$, the exponent of $K$ cannot be any constant more than $1/2$. First note that because of the second constraint, $b=\Omega{(\lg{n})}$. Let $b=O(\lg{n}{K}^t)$ for some constant $t$. Then we must have $g^2=\Omega{(K^{1-t})}$, so we let $g=K^{1-t}$. This means that we have $n=\Theta_K{(N^{1/(1+K^t)})}$, and so the exponent of the lower bound that rETH gives us is $\Theta{(K^{1/2-t/2})}$. Setting $t=0$, we get the values for $b$ and $g$ that we set first. Where we set $g=\floor{\sqrt{K}}$. %\mina{$g=\floor{\sqrt{K}}$}
\end{proof}

%Note that by Theorem \ref{lem:selfReductionsWithbg}, $\oplus$\ckov$(n,b,g)$ requires $n^{\Omega{k}}$ if $bg>ck\lg{n}$ for some constant $c$ and $kg = o(\lg{n})$. By substituting the same values for $b$ and $g$, we get the same lower bound above under rETH. \mina{state what values exactly and maybe just say the above theorem for reth instead of seth?}

\begin{theorem}
\label{thm:xor-to-ov}
Let $K$ be a constant. Any algorithm that solves $\oplus K$-OV of size $N$ with vectors of dimension $\Theta(K\lg{N})$ with probability $1-\frac{2^{-K}}{8}$ where the input is drawn from $D_{OV}(N,K,0.5\lg{\frac{16N}{K^{1/3}}},K^{2/3})$ requires at least $N^{K^{1/3}/4-o(1)}$ time assuming $K^{1/3}$-XOR Hypothesis.
\end{theorem}

\begin{proof}
Using Theorem \ref{thm:ovXORsumFromAny} (step 2) we can reduce $\oplus$\ckxor$(n,b/2k,g)$ to $\oplus$\ckov$(n,b,g)$, where $n=\frac{Ng}{2^bK}$ and $k=K/g$ and we haven't defined the values of $b$ and $g$ yet. This reduction takes $O(ng2^{b/2})$ time. 

By Theorem \ref{lem:selfReductionsWithbg} (step 1), assuming $k$-XOR hypothesis, $\oplus$\ckxor$(n,b/2k,g)$ requires $n^{k/2}$ time if
\begin{enumerate}
    \item $\frac{b}{2k}g\ge k\lg{n}+2$
    \item $kg=o(\lg{n})$
\end{enumerate}
The first condition is equivalent to $bg^3=K^2\lg{n}+4Kg$, and the second condition is equivalent to $K=o(\lg{n})$. Let $g=K^{2/3}$ and $b=\lg{n}+4$. Note that this means that $b=0.5\lg{\frac{16N}{K^{1/3}}}$. The first condition clearly holds. Moreover, we have that $n=N^{1/2}/K^{1/3}=\Theta_K(N^{1/2})$. This means that $\lg{n}=\Theta(\lg{N})$ and so the second condition holds. To see what lower bound we get, Note that $k=K/g=K^{1/3}$. So under $K^{1/3}$-XOR hypothesis, average case $\oplus K$-OV of size $N$ requires $N^{K^{1/3}/4-o(1)}$.

Now we reason how we choose the values of $b$ and $g$. Since $g=o(\lg{n})$ from the second condition, we have that $b>\lg{n}$. Suppose that we set $b=K^{c}\lg{n}$ for some constant $c$. Then $g=K^{(2-c)/3}$ and so $n=\Theta_K(N^{1/(1+K^c)})$. Thus under $k$-XOR hypothesis $K$-OV of size $N$ requires $N^{K^{1/3-4c/3-o(1)}}$. To get the maximum lower bound possible we should set $c=0$, which results in our initial values for $b$ and $g$.
\end{proof} 

\begin{theorem}
Let $K$ be a constant. Any algorithm that solves $\oplus K$-OV of size $N$ with vectors of dimension $\Theta(K\lg{N})$ with probability $1-\frac{2^{-K}}{8}$ where the input is drawn from $D_{OV}(N,K,0.5\lg{\frac{16N}{K^{1/3}}},K^{2/3})$ requires at least $N^{K^{1/3}/4-o(1)}$ time assuming $K^{1/3}$-SUM Hypothesis.
\end{theorem}

\begin{proof}
The approach is almost the same as Theorem \ref{thm:xor-to-ov}. Using Theorem  \ref{thm:ovXORsumFromAny} we can reduce $\oplus$\cksum$(n,b/2k,g)$ to $\oplus$\ckov$(n,b,g)$, where $n=\frac{Ng}{2^bK}$ and $k=K/g$ and we haven't defined the values of $b$ and $g$ yet. This reduction takes $O(ng2^{b/2})$ time. Then by Theorem \ref{lem:selfReductionsWithbg}, assuming $k$-SUM hypothesis, $\oplus$\cksum$(n,b/2k,g)$ requires $n^{k/2}$ time if $\frac{b}{2k}g\ge k\lg{n}+2$ and $kg=o(\lg{n})$. Setting $g=K^{2/3}$ and $b=\lg{n}+4$, we get the lower bound of $N^{K^{1/3}/4-o(1)}$ for average case $\oplus K$-OV under $K^{1/3}$-SUM hypothesis.
\end{proof}

\subsection{Average-case hardness of $\oplus K$-XOR}
We want to show that for any $K$ and $N$, $\oplus K$-XOR of size $N$ is hard over some distribution. 

Our approach is very similar to section \ref{sec:ac-kov} and so we remove unnecessary details. We define the distribution over which we prove $\oplus K$-XOR is hard using Theorem \ref{thm:factoredToUnfactored} (step 4). Choosing appropriate parameters $b$ and $g$, this theorem gives a reduction from $\oplus$\ckxor$(\frac{Ng}{2^bK},b,g)$ to $\oplus K$-XOR of size $N$ with vectors of dimension $bg+(k-1)g\lg{n}$ where $k:=K/g$.
Starting from a uniform distribution on $\oplus$\ckxor$(\frac{Ng}{2^bK},b,g)$, we get a distribution for $\oplus K$-XOR through this reduction. We call this distribution $D_{XOR}(N,K,b,g)$. 

To show that $\oplus K$-XOR is hard on this distribution with some success probability $q$, we need to show that $\oplus$\ckxor$(n,b,g)$ is hard on uniform distribution with success probability $q$ where $n=\frac{Ng}{2^bK}$. To do this, we use Theorem \ref{thm:frameworkImprovement} (step 3) to reduce worst case $\oplus$\ckxor$(n,b,g)$ to average case $\oplus$\ckxor$(n,b,g)$ with uniform distribution with $q=1-2^{-K}/8$. 

To get hardness from $k$-XOR hypothesis for worst case $\oplus$\ckxor$(n,b,g)$, we use Theorem \ref{lem:selfReductionsWithbg} (step 1), which says that $\oplus$\ckxor$(n,b,g)$ requires $n^{\ceil{k/2}-o(1)}$ time. To choose the parameters, by Theorem \ref{lem:selfReductionsWithbg} we need to have that $bg\ge k\lg{n}+2$ and $K=kg = o(\lg{n})$. In this case the optimal values for $b$ and $g$ are $\lg{n}+2=0.5\lg{\frac{4N}{\sqrt{K}}}$ and $\sqrt{K}$ respectively. 

To get hardness from SETH or $k$-SUM hypothesis, we first use Theorem \ref{thm:ovXORsumFromAny} (step 2) to reduce $\oplus$\ckov$(n,b/k,g)$ or $\oplus$\cksum$(n,b/k,g)$ to $\oplus$\ckxor$(n,b,g)$ in $O(ng2^{kb})$ time, and then use Theorem \ref{lem:selfReductionsWithbg} (step 1) to get hardness from SETH or $k$-SUM for $\oplus$\ckov$(n,b/k,g)$ or $\oplus$\cksum$(n,b/k,g)$. 

To get hardness from SETH, by Theorem \ref{lem:selfReductionsWithbg} we must have that $\frac{b}{k}g>ck\lg{n}$ and $K=gk=o(\lg{n})$. In this case the optimal values for $b$ and $g$ are $O(\lg{n})=O(\lg{N})$ and $K^{2/3}$ respectively. By Theorem \ref{lem:selfReductionsWithbg} $\oplus$\ckov$(n,b/k,g)$ requires $n^{\Omega{(k)}}$ time.

To get hardness from $k$-SUM, by Theorem \ref{lem:selfReductionsWithbg} we must have that $\frac{b}{k}g\ge k\lg{n}+2$ and $K=gk=o(\lg{n})$. In this case the optimal values for $b$ and $g$ are $\lg{n}+2=0.5\frac{\lg{4N}}{K^{1/3}}$ and $K^{2/3}$ respectively. By Theorem \ref{lem:selfReductionsWithbg} $\oplus$\cksum$(n,b/k,g)$ requires $n^{\ceil{k/2}}$ time.

By substituting the values of $b,g,k$ and $n$ we get the Theorem \ref{thm:ovxorhard} for $\oplus K$-XOR. See Figure \ref{fig:xor_map} and Table \ref{tab:xor-reduction-param-values}.

\begin{figure}
    \centering
    \includegraphics[width=\linewidth]{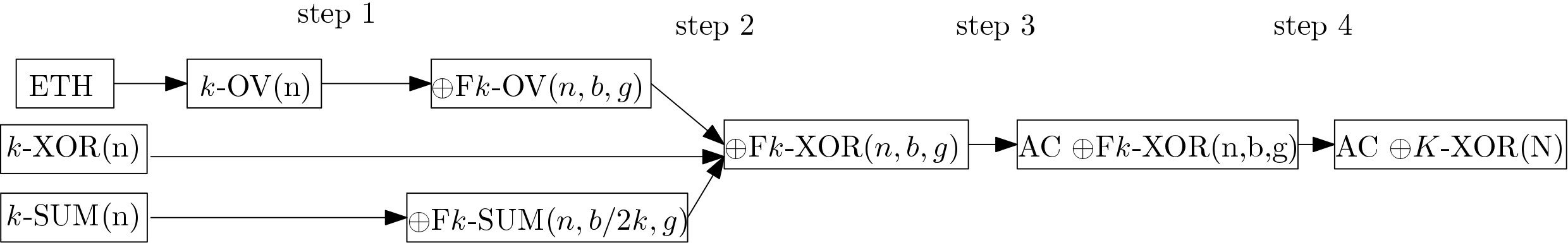}
    \caption{Reductions to average case $\oplus K$-XOR of size $N$. The size of the unfactored problems is mentioned as a parameter in front of them. To see what the values of $b,g,n$ and $k$ are in terms of $K$ and $N$ see Table \ref{tab:xor-reduction-param-values}.}
    \label{fig:xor_map}
\end{figure}

%Mina I love your figures! :D

\begin{table}
    \centering
    \begin{tabular}{|c|c|c|c|}
    \hline
         & $\oplus k$-OV($n$)& $\oplus k$-XOR($n$)& $\oplus k$-SUM($n$) \\ \hline
        $b$ & $\lg{N}$& $ \lg{N}$& $ \lg{N}$\\ \hline
        $g$ & $K^{2/3}$ & $\sqrt{K}$ & $K^{2/3}$\\ \hline
        $n$ & $\sqrt{N}$& $\sqrt{N}$& $\sqrt{N}$\\ \hline
        $k$ & $K^{1/3}$& $\sqrt{K}$&$K^{1/3}$ \\ \hline
    \end{tabular}
    \caption{parameter values for reductions from $k$-OV, $k$-XOR and $k$-SUM to average case $K$-XOR, where the starting problem is of size $n$ and $K$-XOR is of size $N$. The exact values are within constant factor away from the values mentioned in the table.}
    \label{tab:xor-reduction-param-values}
\end{table}

\ovxorhard*

\subsection{Average-case hardness of $\oplus K$-SUM}
The approach is the same as the previous two sections, except that in step 2 we get better bounds. 

We define the distribution over which we prove $\oplus K$-SUM is hard using Theorem \ref{thm:factoredToUnfactored} (step 4). Choosing appropriate parameters $b$ and $g$, this theorem gives a reduction from $\oplus$\cksum$(\frac{Ng}{2^bK},b,g)$ to $\oplus K$-SUM of size $N$ with vectors of dimension $bg+(k-1)g\lg{n}$ where $k:=K/g$.
Starting from a uniform distribution on $\oplus$\cksum$(\frac{Ng}{2^bK},b,g)$, we get a distribution for $\oplus K$-SUM through this reduction. We call this distribution $D_{SUM}(N,K,b,g)$. 

To show that $\oplus K$-SUM is hard on this distribution with some success probability $q$, we need to show that $\oplus$\cksum$(n,b,g)$ is hard on uniform distribution with success probability $q$ where $n=\frac{Ng}{2^bK}$. To do this, we use Theorem \ref{thm:frameworkImprovement} (step 3) to reduce worst case $\oplus$\cksum$(n,b,g)$ to average case $\oplus$\cksum$(n,b,g)$ with uniform distribution with $q=1-2^{-K}/8$. 

To get hardness from $k$-SUM hypothesis for worst case $\oplus$\cksum$(n,b,g)$, we use Theorem \ref{lem:selfReductionsWithbg} (step 1), which says that $\oplus$\cksum$(n,b,g)$ requires $n^{\ceil{k/2}-o(1)}$ time. To choose the parameters, by Theorem \ref{lem:selfReductionsWithbg} we need to have that $bg\ge k\lg{n}+2$ and $K=kg = o(\lg{n})$. In this case the optimal values for $b$ and $g$ are $\lg{n}+2=0.5\lg{\frac{4N}{\sqrt{K}}}$ and $\sqrt{K}$ respectively. 

%To get hardness from SETH or $k$-SUM hypothesis, we first use Theorem \ref{thm:ovXORsumFromAny} (step 2) to reduce $\oplus$\ckov$(n,b/k,g)$ or $\oplus$\cksum$(n,b/k,g)$ to $\oplus$\ckxor$(n,b,g)$ in $O(ng2^{kb})$ time, and then use Theorem \ref{lem:selfReductionsWithbg} (step 1) to get hardness from SETH or $k$-SUM for $\oplus$\ckov$(n,b/k,g)$ or $\oplus$\cksum$(n,b/k,g)$. 

To get hardness from SETH, we first use Theorem \ref{lem:SUMreductions} (step 2) to reduce $\oplus$\ckov[(k-1)]$(n-1,b/\ceil{\lg(k-1)},g)$ to $\oplus$\cksum$(n,b,g)$ in $O(ng2^{kb})$ time. Then we use Theorem \ref{lem:selfReductionsWithbg} to show that under SETH, $\oplus$\ckov[(k-1)]$(n-1,b/\ceil{\lg(k-1)},g)$ requires $(n-1)^{\Omega(k-1)}$. By Theorem \ref{lem:selfReductionsWithbg} we must have that $\frac{b}{\ceil{\lg(k-1)}}g>c(k-1)\lg{n}$ and $K-g=g(k-1)=o(\lg{n})$. In this case the optimal values for $b$ and $g$ are $\lg{n}$ and $ck\lg{k}$. Using $gk=K$, we have that $\lg{K}=\Theta(\lg{k})$, so we have $k=\Theta(\sqrt{\frac{K}{\lg{K}}})$ 

%by Theorem \ref{lem:selfReductionsWithbg} we must have that $\frac{b}{k}g>ck\lg{n}$ and $K=gk=o(\lg{n})$. In this case the optimal values for $b$ and $g$ are $O(\lg{N})$ and $K^{2/3}$ respectively. By Theorem \ref{lem:selfReductionsWithbg} $\oplus$\ckov$(n,b/k,g)$ requires $n^{\Omega{(k)}}$ time.

To get hardness from $k$-XOR, we first use Theorem \ref{lem:SUMreductions} (step 2) to reduce $\oplus$\ckxor[(k-1)]$(n-1,b/\ceil{\lg(k-1)},g)$ to $\oplus$\cksum$(n,b,g)$ in $O(ng2^{kb})$ time. Then we use Theorem \ref{lem:selfReductionsWithbg} to show that under $(k-1)$-XOR,  $\oplus$\ckov[(k-1)]$(n-1,b/\ceil{\lg(k-1)},g)$ requires $n^{\ceil{\frac{k-1}{2}}}$. By Theorem \ref{lem:selfReductionsWithbg} we must have that $\frac{b}{\ceil{\lg(k-1)}}g\ge (k-1)\lg{n}+2$ and $K-g=g(k-1)=o(\lg{n})$. In this case the optimal values for $b$ and $g$ are $\lg{n}+2$ and $k\lg{k}$. Using $gk=K$, we have that $k^2\lg{k}=K$, so we have $\sqrt{\frac{2K}{\lg{K}}}\le k$ and $N=\Theta(n^2)$.

%by Theorem \ref{lem:selfReductionsWithbg} we must have that $\frac{b}{k}g=k\lg{n}+2$ and $K=gk=o(\lg{n})$. In this case the optimal values for $b$ and $g$ are $\lg{n}=0.5\frac{\lg{N}}{K^{1/3}}$ and $K^{2/3}$ respectively. By Theorem \ref{lem:selfReductionsWithbg} $\oplus$\cksum$(n,b/k,g)$ requires $n^{\ceil{k/2}}$ time.

By substituting the values of $b,g,k$ and $n$ we get the following theorem. See Figure \ref{fig:sum_map} and Table \ref{tab:sum-reduction-param-values}.

\begin{figure}
    \centering
    \includegraphics[width=\linewidth]{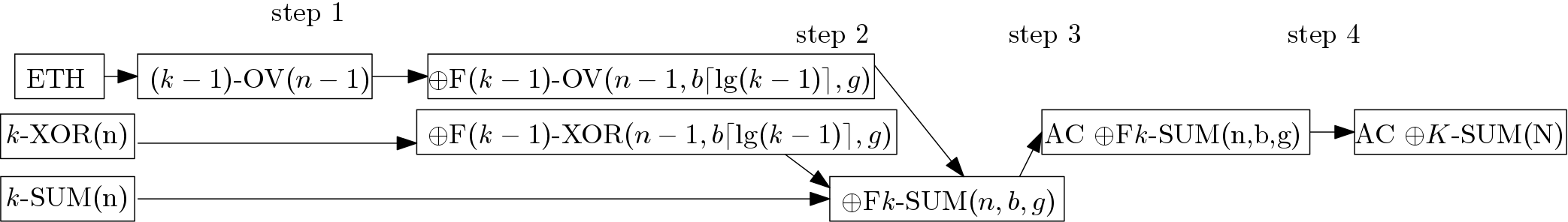}
    \caption{Reductions to average case $K$-SUM of size $N$. The size of the unfactored problems is mentioned as a parameter in front of them. To see what the values of $b,g,n$ and $k$ are in terms of $K$ and $N$ see Table \ref{tab:sum-reduction-param-values}.}
    \label{fig:sum_map}
\end{figure}

\begin{table}
    \centering
    \begin{tabular}{|c|c|c|c|}
    \hline
         & $\oplus k$-OV($n$)& $\oplus k$-XOR($n$)& $\oplus k$-SUM($n$) \\ \hline
        $b$ & $\lg{N}$& $ \lg{N}$& $ \lg{N}$\\ \hline
        $g$ & $K\lg{K}$ & $K\lg{K}$ & $\sqrt{K}$\\ \hline
        $n$ & $\sqrt{N}$& $\sqrt{N}$& $\sqrt{N}$\\ \hline
        $k$ & $\sqrt{K/\lg{K}}$& $\sqrt{K/\lg{K}}$&$\sqrt{K}$ \\ \hline
    \end{tabular}
    \caption{parameter values for reductions from $k$-OV, $k$-XOR and $k$-SUM to average case $K$-SUM, where the starting problem is of size $n$ and $K$-SUM is of size $N$. The exact values are within constant factor away from the values mentioned in the table.}
    \label{tab:sum-reduction-param-values}
\end{table}
\sumhard*

\section{From Clique to Average-Case Fine-Grained Problems}
\label{sec:cliqueToFGC}

\newcommand{\binomk}{$\binom{k}{2}$}

%\mina{we need a theorem stating the lower bound from k-clique hypothesis, to put in the intro.}

In this section we reduce from the $\oplus k$-clique problem to instances of $\oplus \bikmath$-XOR, $\oplus \bikmath$-OV and $\oplus \bikmath$-SUM. 
Our results answer the open question posed in
\cite{afargholiV16} by Jafargholi and Viola in their appendix B. In appendix B in \cite{afargholiV16} they show how to reduce $4$-clique to $6$-SUM over the group $\mathbb{Z}^t_3$. In this section we give the generalization of this result and reduce from  average-case parity $k$-clique to average-case $\oplus \bikmath$-OV, $\oplus \bikmath$-XOR, and  $\oplus \bikmath$-SUM in general. Thus, we answer the open question posed in their appendix B.
This is the reverse direction of the reduction from \cite{losingWeight}, which gives a reduction from $k$-SUM and $k$-XOR to many instances of the $k$-clique problem. 
We will handle all three problems with a similar overall structure. We will start with $\oplus k$-OV and $\oplus k$-XOR as both have $n$ vectors in their input (where as $\oplus k$-SUM has $n$ integers as input). 

% lol my bad terrible phrasing. We are separating concerns \mina{concerns?} to avoid mixing numbers and vectors. 

\begin{theorem}
We have the following reductions from average-case parity $k$-clique to average-case $\oplus \bikmath$-OV and $\oplus \bikmath$-XOR.
\begin{itemize}
    \item An instance of $\oplus k$-clique on $n$ nodes 
    can be reduced into one instance of $\oplus \bikmath$-XOR with $O(n^2)$ vectors of length $2\bikmath\lg(n)$. 
    \item An instance of $\oplus k$-clique on $n$ nodes can be reduced to one instance of $\oplus \bikmath$-OV with $O(n^2)$ vectors of length $4\bikmath\lg(n)$. 
\end{itemize}
\label{thm:cliqueToOVXOR}
\end{theorem}
\begin{proof}
We take as input a graph $G$ with $n$ nodes in $V$ and $O(n^2)$ edges in $E$.
First we will describe the general structure of these reductions. We will also treat each node in $G$ as having a unique label in $[0,n-1]$. This will allow us to index easily and avoid double counting.

\paragraph{Intuitive explanation of the reduction}
Let us start with a high-level explanation of this reduction before we get bogged down in notation. We will give an approach that lets us go from a graph to $n^2$ vectors (one for each edge). We are looking at \binomk-XOR and \binomk-OV problems and we are using the \binomk vectors to check if a given set of \binomk edges forms a clique. Each vector has $k$ parts, where each part is encoding a node. The $i^{th}$ of these parts is checking if all of the edges given agree about which node is the $i^{th}$ node in the $k$-clique. To do this each part is split into $k-2$ parts corresponding to checking pair-wise if all of the $k-1$ incoming edges have the same node for the $i^{th}$ node. So, we have an original graph with $n$ nodes and we are trying to create a reduction which lets us check over all sets of \binomk edges if those edges form  a $k$ clique. This results in some involved notation: we have the labels from the original graph but also the labels in our new possible clique. Notably, we will have some vector that corresponds to the edge between the $i^{th}$ and $j^{th}$ nodes in the clique. But, we need to check that it is valid by checking if all other edges in this maybe-clique agree about which nodes in the original graph correspond to the $i^{th}$ and $j^{th}$ nodes. We set up this structure where we can count $k$-cliques if we can use \binomk vectors and check that nodes are equal. This means we can build a reduction as long as we can set up multiple equality checks in the vector. We then use this general structure to build both the \binomk-OV and \binomk-XOR reductions. 

\paragraph{Reduction with notation and details} We assume that the nodes of the graph have a fixed ordering.
We will produce \binomk~lists with one vector for each edge in the graph. We will index into these lists $L_{i,j}$ with two numbers  $i,j \in [1,k]$ and $i< j$. Imagine we select a vector from each list: $u_{1,2}, u_{1,3}, \ldots, u_{k-1,k}$. The vector $u_{i,j}$ we selected from $L_{i,j}$ corresponds to some edge in the original graph. Let $(u_{i,j}(i), u_{i,j}(j))$ be the corresponding edge to $u_{i,j}$ (the strange notation for the nodes in the graph $u_{i,j}(i)$ is capturing the fact that we need to know which vertex is the `$i$' vertex and which is the `$j$' vertex). %\mina{the notation for nodes is very confusing, maybe let's not mention it here, and just say that $L_{i,j}$ is supposed to correspond to the edge between the $i$th vertex and $j$th vertex of the clique.}. 
The vectors in $L_{i,j}$ correspond to the edges that are between the $i^{th}$ and $j^{th}$ nodes in the clique.

We will avoid double counting by insisting that the label of node $u_{i,j}(i)$ be less than the label for node $u_{i,j}(j)$. 
%\mina{instead of $u_{i,j}[i]$ and $u_{i,j}[j]$, we can have $u_{i,j}$ and $u_{j,i}$ for $i<j$. the first index says which endpoint we are referring to.} \andrea{I am open to other opinions, but I don't think that notation would be clearer. } 
Note the original graph is not (necessarily) $k$-partite, \emph{all} edges are having corresponding values added to \emph{all} lists. We want the \binomk~tuple of values $(u_{1,2}, u_{1,3}, \ldots, u_{k-1,k})$ to form a \binomk-OV or \binomk-XOR iff the \binomk~corresponding edges form a $k$ clique. To make this correspondence work we will need to enforce the constraint that if a \binomk-tuple of values are a solution then for all $i$:  $u_{i,j}(i) = u_{i,j'}(i)$ for all $j$ and $j'$. That is, every solution does actually correspond to a single set of $k$ nodes.

\usetikzlibrary{positioning,calc}
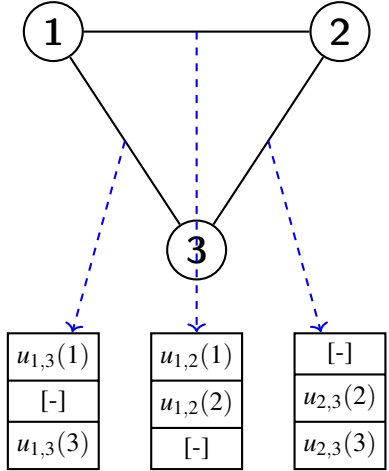
\begin{figure}
    \centering
\begin{tikzpicture}[auto, node distance=3cm, thick, main node/.style={circle,draw,font=\sffamily\Large\bfseries}]
  
  \node[main node] (1) {1};
  \node[main node] (2) [right=of 1] {2};
  \node[main node] (3) [below=2.5cm of $(1)!0.5!(2)$] {3};
  
  \draw (1) -- (2);
  \draw (2) -- (3);
  \draw (3) -- (1);

  \node[rectangle split, rectangle split parts=3, draw, anchor=center] (v12) [below=4cm of $(1)!0.5!(2)$] {$u_{1,2}(1)$\nodepart{second}$u_{1,2}(2)$\nodepart{third} [-]};
  \node[rectangle split, rectangle split parts=3, draw, anchor=center] (v23) [below=4cm of $(2)$] {[-]\nodepart{second}$u_{2,3}(2)$\nodepart{third}$u_{2,3}(3)$};
  \node[rectangle split, rectangle split parts=3, draw, anchor=center] (v31) [below=4cm of $(1)$] {$u_{1,3}(1)$\nodepart{second}[-] \nodepart{third}$u_{1,3}(3)$};

  \draw[->,dashed,blue] ($(1)!0.5!(2)$) -- (v12);
  \draw[->,dashed,blue] ($(2)!0.5!(3)$) -- (v23);
  \draw[->,dashed,blue] ($(3)!0.5!(1)$) -- (v31);
  
\end{tikzpicture}
    \caption{A depiction of the reduction from 3 clique to $3$-XOR and $3$-OV. We use the first section of the vectors to check that all edges are using the same node for 'node 1'. We use the second section to check that all edges are using the same node for 'node 2'. Same for 3. When $k$ grows there are more than two edges coming into each node. So, the section to check all edges incident to 'node 1' agree on what node 1 is grows. In general you need $k-2$ such checks.}
    \label{fig:enter-label}
\end{figure}

To do this we will split the vector into $k^2$ parts. These will correspond to checking that the nodes are consistent.  For each $i$ we will enforce $k-2$ checks. If $i\ne 1$ we will check that $u_{i,1}(i)=u_{i,j}(i)$ for all $j\ne 1$ and $j\neq i$. If $i=1$ we will check that $u_{1,2}(1) = u_{1,j}(1)$ for all $j\ne 1$ and $j\neq 2$.  We will make each check independent. 

So, for each problem we want to have a way to check equality and enforce that each check is independent of each other check. 
Recall that our problem definitions of $\oplus \bikmath$-OV and $\oplus \bikmath$-XOR enforce that one must select exactly one vector from each list. Let $\bullet$ be concatenation. %\mina{are we using the same thing for concatenation everywhere?}\andrea{I think so?} 
Let $\bar{s}$ be the bit-wise negation of $s$.

\paragraph{Equality Checks}
We will now give a function for each problem that produces equality checking.

\begin{itemize}
    \item \binomk-OV:  Given two Boolean vectors $s$ and $t$ note that $<s\bullet \bar{s}>^T \cdot <\bar{t} \bullet t> = 0$ iff $s=t$. This is what we will use for equality checking. Further note that if $v\cdot u = 0$ then the vectors $v, u, \vec{1}, \vec{1},\ldots,  \vec{1}$ are \binomk-orthogonal. 
    \item \binomk-XOR: Given two Boolean vectors $s$ and $t$ note that $s \oplus t = \vec{0}$ iff $s=t$. Further note that if $s \oplus t = \vec{0}$ then $s \oplus t \oplus \vec{0} \oplus \cdots \oplus \vec{0} = \vec{0}$.
\end{itemize}

\newcommand{\bool}{Bool}

Let $\bool(v)$ return a Boolean vector of length $\lg(n)$ that uniquely corresponds to the node $v$ in $V$ (consider the boolean representation of the label of $v$ between $[0,n-1]$). 

\paragraph{Generic Structure:}
For every edge $(v,w) \in E$ we will produce a vector for each list. We will define the function $\lambda_{ij}(v,w)$ such that it returns this vector. So this function, $\lambda_{ij}(v,w)$, lambda takes in an edge from the original graph, $(u,v)$, and a pair of indices, ($i$ and $j$), and produces one vector in the output. There are \binomk functions (for all pairs of $i$ and $j$ in $[1,k]$ where $i<j$. The function $\lambda_{ij}$ is used to create the vectors in $L_{ij}$. So we are producing $|E|$ vectors for each of \binomk lists of vectors. 

Each vector $\lambda_{ij}(v,w)$ consists of $k(k-2)$ checks. Each of these checks is checking if two values are equal. Specifically, each check is a check of if two vectors agree about the original value of the $i^{th}$ node. We will present the structure in full generality: how you can use a gadget for checking if two nodes have the same value to check if all $k$ nodes form a clique. In the vectors produced by $\lambda_{ij}(v,w)$ the first $k-2$ checks will be for node $1$, then the next $k-2$ checks will be for node $2$, etc. So, in the $i^{th}$ set of $(k-2)$ checks we are checking if all edges agree on the value of the $i^{th}$ node. To do this pairwise comparison we want to have other vectors having a neutral value in these locations (for $k$-XOR this is zero and for $k$-OV this is one).   Let $h$ be `filler' (a vector of all zeros or all ones) and let $u^x$ be $x$ copies of a vector $u$ concatenated. We will define two helper functions (corresponding to the sections checking $v$ and $w$). 

We described equality checks above. We will use $g_{+}(v)$ to produce a `positive' value and $g_{(-)}(v)$ to produce the `negative' value. That is, we want to produce sections of the vector that will multiply to zero iff the value passed in is the same. There isn't a direct notion of having a positive and negative vector in the space of vectors in OV, but, there do exist functions, $g_{+}(v)$ and $g_{(-)}(v)$, where two resulting vectors will bitwise multiply to the zero vector iff the the vector passed in to the two functions is equal. 
%\mina{what? maybe explain more what these two functions need to do, and maybe say what they are first (move the last part of the reduction to here) and then define the things below}. 
%We need two separate functions because, for OV, we need one to be `negative' and the other to be `positive'\mina{what?}. 
Finally, we will enforce that $\ell = |g_{+}(v)|=|g_{(-)}(v)|=|h|$. 

Now, let us split up the problem of comparing these sections by creating helper functions to define each of these $k$ sections. Recall that we are doing something special if $i=1$. To avoid doing all pairwise comparisons we simply check that vectors $\lambda_{ij}$ where $i\ne 1$ agree with the vector $\lambda_{1j}$ on the value of the node $j$ and with vector $\lambda_{1i}$ on the value of the node $i$. For the comparison on the $j^{th}$ node this happens in the $j^{th}$ set of $k-2$ checks, to disambiguate it is the $(i-1)^{th}$ such check. So for a vector $\lambda_{ij}(u,v)$ where $i\ne 1$ (remember $i<j$ and $u<v$ by construction) what does the $a^{th}$ set of checks look like? We will use $\bullet$ as a vector  concatenation symbol.

$$\gamma^a_{ij}(u,v) = \begin{cases}
h^{k-2} &\text{ if } a\ne i \wedge a \ne j \\
<h^{j-1} \bullet g_{(-)}(u) \bullet h^{k-j} > &\text{ if } i=a \\
<h^{i-1} \bullet g_{(-)}(v) \bullet h^{k-i} > &\text{ if } j=a \\
\end{cases}$$

So, we have an empty section of vector if we aren't in the section for checking $i$ or $j$. If we are in the section for comparing $i$ we check if node $u$ matches. If we are in the section for comparing $j$ we check if node $v$ matches. 

Now, we have special behavior around $i=1$ so lets explore that vector:
$$\gamma^a_{1j}(u,v) = \begin{cases}
h^{k-2} &\text{ if } a\ne 1 \wedge a \ne j \\
<g_{(+)}(v)^{k-2}> &\text{ if } j=a \\
<g_{(+)}(u)^{k-2}> &\text{ if } a=1 \wedge j=2\\
<h^{j-3} g_{(-)}(u) h^{k-j+1}> &\text{ if } a=1 \wedge j\ne 2\\
\end{cases}$$

So, we now have an empty section if $a\ne 1$ and we aren't checking the value of $j$. If we are checking the value of $j$ we provide $k-2$ copies of the `positive' version of $v$ (to check against other values given to $j$). If we are in section $a=1$ then we use the value of node $i=1$ given by $\lambda_{12}$ ($i=1$ and $j=2$) to all other values given to node $1$ given by $\lambda_{1j}$. Note that we use the positive value for the $\lambda_{12}$ version and the negative version for $\lambda_{1j}$ (ensuring a zero vector iff these values are equal).

We can now define $\lambda$ simply as:
$$\lambda_{i,j}(u,v) = <\gamma^1_{ij}(u,v) \bullet \gamma^2_{ij}(u,v) \bullet \cdots \bullet \gamma^{k}_{ij}(u,v)>.$$

Consider briefly a given section $a$ of the vector: note that we will have $k-2$ copies of $g_{(+)}(\cdot)$ with the value associated with the $a^{th}$ node in exactly one vector ($\lambda_{1a}$ if $a\ne 1$ and $\lambda_{12}$ if $a=1$). Then there will be $k-2$ other vectors that have a single $g_{(-)}(\cdot)$ and otherwise `filler'. All other vectors will have purely filler in this section. Each of the $(k-2)$ vectors with a single $g_{(-)}(\cdot)$ value puts it in a non-overlapping location that lines up with one of the positive $g_{(+)}(\cdot)$ values. So, iff all vectors agree on the value of the $a^{th}$ node this section will combine to the zero vector. 

Note that the vectors we put into $L_{i,j}$ has checks only in the sections related to $i$ and $j$. Within each section we are comparing all the sections representing nodes against one other section representing the same node, to ensure all original nodes correspond to one set of $k$ nodes. For each individual vector in $L_{i,j}$ each section that isn't filler is compared against \emph{exactly} one other vector. 
We state an assumption that we show is true.
\begin{assumption}\label{assumption:zero}
$g_{+}(v)$ and $g_{(-)}(w)$ and an arbitrary number of filler vectors only OV/XOR to the all zeros vector iff $v=w$. 
\end{assumption}

Then note that this structure enforces that a \binomk~tuple of vectors $(u_{1,2},u_{1,3}\ldots,u_{k-1,k})$ where $u_{i,j}\in L_{i,j}$ OVs/XORs to the all zeros vector iff $u_{i,j}(i) = u_{i,j'}(i)$ for all $i$. Let this node be $v_i$, i.e. $v_i = u_{i,j}(i) = u_{i,j'}(i)$. We now know that in the original graph there exist \binomk~edges $(v_1,v_2), (v_1, v_3),\ldots, (v_{k-1},v_k)$. That is, edges between all pairs of nodes in $v_1,\ldots, v_k$. So, it corresponds to a clique.

We avoid double counting because of the ordering we have assumed on the vertices of the graph, so essentially in the solution above the label of $v_i$ is less than the label of $v_j$ if $i<j$.
%
%Now we show that we avoid double counting. This is because the edges enforce that the label of $v_i$ is less than the label of $v_j$ if $i<j$. So the count, and by implication the parity, of the number of cliques will be equal to the number/parity of OV or XOR solutions. Specifically, for every set of $k$ nodes in the original graph there is one mapping of nodes in the original into the labels $u[1],\ldots, u[k]$ (the node with the $\ell^{th}$ smallest label in the original graph becomes $u[\ell]$). A set of $k$ nodes in the original graph will have $\binom{k}{2}$ vectors that correspond to a valid XOR/OV solution iff all $\binom{k}{2}$ edges exist between the $k$ nodes in the original graph. 
Thus the count and parity of the number of cliques will correspond to the count or parity of the number of XOR/OV solutions. 

Now, all we need to do is define $h$, $g_{+}(v)$ and $g_{(-)}(v)$ for $OV$ and $XOR$. 

\paragraph{Reduction for \binomk-OV:}
For \binomk-OV we will use $h= \vec{1}^{2\lg(n)}$. 
We will set 
$$g_{+}(v) = \bool(v) \bullet \overline{\bool(v)}$$
and
$$g_{(-)}(v) = \overline{\bool(v)} \bullet \bool(v).$$

Let $OV(\cdot)$ be the operation of entry-wise multiplying the input vectors which must all be $\{0,1\}^\ell$ for some $\ell$. So zero-one vectors of the same length.
Note that 
$$OV(g_{+}(v), g_{+}(w), h, \ldots, h) = \vec{0}$$ 
iff $v=w$, so our definition follows Assumption \ref{assumption:zero}.

\paragraph{Reduction for \binomk-XOR:}
For \binomk-XOR we will use $h= \vec{0}^{\lg(n)}$. 
We will set 
$$g_{+}(v) = \bool(v) $$
and
$$g_{(-)}(v) =\bool(v).$$

Note that 
$$g_{+}(v) \oplus g_{+}(w) \oplus h \oplus \ldots \oplus h = \vec{0}$$
iff $v=w$, so our definition follows Assumption \ref{assumption:zero}.

So, we can transform an instance of $\oplus k$-clique into either $\oplus \bikmath$-OV or $\oplus \bikmath$-XOR.
\end{proof}

We need a different approach for $k$-SUM because it uses numbers instead of vectors. The underlying approach is the same, however, it will be cleanest to present it separately.

\newcommand{\num}{Num}
\begin{theorem}
An instance of $\oplus k$-clique on $n$ nodes
    can be  reduced to one instance of \binomk-SUM with $O(n^2)$ numbers.
\label{thm:cliqueToSUM}
\end{theorem}
\begin{proof}

For the \binomk-SUM problem we will take a similar approach. We will define numeric structures that check equality of a specific node in the clique. Then, to simultaneously check all of these we will multiply each structure individually by large numbers to enforce that the sum will be zero only if every single structure sums to zero separately. We are putting this in a separate theorem because the difference in structure of a vector vs number makes it cleaner to separate.  

Once again we take as input a graph $G$ with $n$ nodes in $V$ and $O(n^2)$ edges in $E$.
First we will describe the general structure of these reductions. We will also treat each node in $G$ as having a unique label in $[0,n-1]$. This will allow us to index easily and avoid double counting.
We will produce \binomk~lists with one vector for each edge in the graph. We will index into these lists $L_{i,j}$ with two numbers $i,j \in [1,k]$ such that $i\leq j$. %Let $(u_{i,j}(i), u_{i,j}(j))$ be the corresponding edge to vector $u_{i,j}$. We will avoid double counting by insisting that the label of $u_{i,j}(i)$ be less than the label for $u_{i,j}(j)$. Note the original graph is not (necessarily) $k$-partite \emph{all} edges are having corresponding values added to \emph{all} lists. We want the \binomk~tuple of values $(u_{1,2}, u_{1,3}, \ldots, u_{k-1,k})$ to form a \binomk-SUM iff the \binomk~corresponding edges form a $k$ clique. To make this correspondence work we will need to enforce the constraint that if a \binomk-tuple of numbers are a zero sum then for all $i$ we have  $u_{i,j}(i) = u_{i,j'}(i)$ for all $j$ and $j'$. That is, every solution does actually correspond to a single set of $k$ nodes.% \mina{same comments here as before. we can also aggregate this first part for all and then have separate theorems when in comes to constructing the vectors.} \andrea{is this better?}

%We will take the same strategy as above of setting up checks and validating them. We similarly need the notion of an equality check. 
Now we describe the reduction in detail. Given two numbers $s$ and $t$ note that $s-t =0$ iff $s=t$. Next note that if $s-t =0$ then $s+(-t)+0+\cdots+0=0$ as well. 
Let $\num(v)$ return the number between $[0,n-1]$ uniquely associated with the node $v$. We will now define our helper functions. 
For \binomk-SUM we use 
\begin{align*}
g_{+}(v) &= \num(v)\\
g_{(-)}(v) &= -\num(v).
\end{align*}
% $g_{+}(v) = \num(v) $
% and
% $g_{(-)}(v) = -\num(v).$

We define the numbers in the ${k \choose 2}$-SUM instance so that they are composed of different sections, and we define $W=2k^2n$ as a very large value we can multiply to ensure that there are no carries between different sections of our numbers.  By construction we enforce that $i<j$. For every edge $(u,v) \in E$ we will produce a vector for each list $L_{i,j}$ and this vector is $\lambda_{ij}(u,v)$. We use helper functions $\gamma_{ij}^a(\cdot)$ and define:
$$\lambda_{ij}(u,v) = \sum_{a=0}^{k} \gamma_{ij}^a(u,v) W^{a(k+2)}.$$
Each $\gamma_{ij}^a(\cdot)$ which is defined below is a number in $[-W^k,W^k]$, and we want each $\gamma_{ij}^a(\cdot)$ to be non-zero if $a=i$ or $a=j$, so that we can conclude which node is the $a^{th}$ node of a possible clique. This will be more clear later.

%$$\gamma_{ij}(v) = \begin{cases}
% g_{+}(v) W^{j-1} &\text{ if } (i\ne 1 \wedge j\ne 1) \vee (i=1 \wedge j\ne 2) \\
%g_{(-)}(v) \sum_{p = 1}^{k-1} W^p   &\text{ if } i\ne 1 \wedge j=1 \\
% g_{(-)}(v) (W + \sum_{p=2}^{k-1} W^p) &\text{ if } i=1 \wedge j=2 
%\end{cases}$$
%\mina{why the second and third cases are separate?}

%$$\zeta_{ij}(v)= \begin{cases}
%g_{+}(v) W^{i-1} &\text{ if } (j\ne 1 \wedge i\ne 1) \vee (j=1 \wedge i\ne 2) \\
%g_{(-)}(v) \sum_{p = 1}^{k-1} W^p &\text{ if } j\ne 1 \wedge i=1 \\
% g_{(-)}(v) (W + \sum_{p=2}^{k-1} W^p  )&\text{ if } j=1 \wedge i=2 
%\end{cases}$$

%We are going to represent the $a^{th}$ equality check. So the equality check for the $a^{th}$ node when 
For $i \ne 1$ let:
$$\gamma^a_{ij}(u,v) = \begin{cases}
0 &\text{ if } a\ne i \wedge a \ne j \\
g_{(-)}(u) \cdot W^{j-3} &\text{ if } i=a \\
g_{(-)}(v) \cdot W ^{i-2} &\text{ if } j=a \\
\end{cases}$$

%Now we will give the representation of the equality check 
For $i=1$ let:
$$\gamma^a_{1j}(u,v) = \begin{cases}
0 &\text{ if } a\ne 1 \wedge a \ne j \\
\sum_{\ell = 0}^{k-3} g_{(+)}(v)W^{\ell} &\text{ if } j=a \\
\sum_{\ell = 0}^{k-3} g_{(+)}(u)W^{\ell} &\text{ if } a=1 \wedge j=2\\
 g_{(-)}(u) W^{j-3} &\text{ if } a=1 \wedge j\ne 2\\
\end{cases}$$

%\mina{note that in the above $i<j$ so we need to change things}
%Now we can define $\lambda_{ij}(v,w)$, we insist that the label for node $u$ be less than the label for node $v$. Further note that $i<j$. So, we are effectively assigning node $u$ to be the $i^{th}$ node in the clique and node $v$ to be the $j^{th}$ node in the clique:
%$$\lambda_{ij}(v,w) = W^{k(i-1)} \left( \gamma_{ij}(v) W^{k(j-i-1)} \zeta_{ij}(w) \right).$$
%$$\lambda_{ij}(u,v) = \sum_{a=0}^{k} \gamma_{ij}^a(u,v) W^{a(k+2)}.$$

Now we prove that our construction works. Suppose that we take numbers $\lambda_{ij}(u_{i,j},v_{i,j})$ from each list $L_{i,j}$. First, we want to make the following claim:
We have
    $ \sum_{i<j\in [1,k]} \lambda_{ij}(u_{i,j},v_{i,j}) =0$
if and only if
$ \sum_{i<j\in [1,k]} \gamma_{ij}^a(u_{i,j},v_{i,j}) =0$ for all $a\in [1,k]$.

To prove this claim first note that the values in $\gamma_{ij}^a(u_{i,j},v_{i,j})$ are bounded by $[-W^{k}, W^{k}]$. There are \binomk of these values we are summing so the range of their sum is at most $[-\bikmath W^k, \bikmath W^k]$. Now, note that we multiply all values of $\gamma_{ij}^a(u_{i,j},v_{i,j})$  by $W^{a(k+2)}$. Moreover, $W^{k+2} >> 2 \bikmath W^k$. So, if there is an $a$ where $ \sum_{i<j\in [1,k]} \gamma_{ij}^a(u_{i,j},v_{i,j}) \ne 0$ then for $a' <a$ the total sum would be less than $\sum_{i<j\in [1,k]} \gamma_{ij}^a(u_{i,j},v_{i,j})$ and for $a'>a$ the remainder of their sum mod $W^{a(k+2)}$ would be zero. On an intuitive level: no carries pass between these sums. This fifnishes the proof of the claim.

Now, we want to argue that 
$ \sum_{i<j\in [1,k]} \gamma_{ij}^a(u_{i,j},v_{i,j}) =0$
iff for all $i,i' \in [1,a-1]$ and $j, j' \in [a+1,k]$ we have that $u_{i,a} = u_{i',a} = v_{a,j} = v_{a,j'}$. That is, all of the nodes the edges picked as their $a^{th}$ node are the same. 

We will need two cases $a=1$ and $a \ne 1$. In both cases we rely on the fact that $Num(u_{i,j}),Num(v_{i,j}) \in [1,n]$ and thus our construction, once again, avoids carries. 

Let us start with the case of $a=1$. When $i\ne 1$ then as $i<j$, we have 
 $j\ne 1$ and so $\gamma_{ij}^1(u_{i,j},v_{i,j})=0$. So, the only non-zero values are $\gamma_{1j}^1(u_{1,j},v_{1,j})$ for $j>1$. 
 If $j=2$ then
$\gamma_{1j}^1(u_{1,2},v_{1,2})=\sum_{\ell =0}^{k-3} Num(u_{1,2})W^{\ell}$. Furthermore, we have $\sum_{j=3}^{k}\gamma_{1j}^1(u_{1,j},v_{1,j}) = \sum_{j=3}^k -Num(u_{1,j})W^{j-3}$.
 
 First note that if $u_{1,2} = u_{1,j}$ for all $j\in[3,k]$ then 
$\sum_{j=1}^{k}\gamma_{1j}^1(u_{1,j},v_{1,j})=0$. If there is some $\hat{j}$ where $u_{1,2} \ne u_{1,\hat{j}}$ then $(Num(u_{1,2}) - Num(u_{1,\hat{j}}))W^{\hat{j}-3} \ne 0$. Furthermore, $\sum_{j= \hat{j}+1}^k -Num(u_{1,j})W^{j-3}$ is zero mod $W^{\hat{j}-2}$. Also $\sum_{j= 1}^{\hat{j}-1} -Num(u_{1,j})W^{j-3} \in (-W^{\hat{j}-3}, W^{\hat{j}-3})$, whereas $|(Num(u_{1,2}) - Num(u_{1,\hat{j}}))W^{\hat{j}-3}|>W^{\hat{j}-3}$. So in this case $\sum_{j=1}^{k}\gamma_{1j}^1(u_{1,j},v_{1,j})\neq 0$, and so only if all $u_{1,j}$ are the same node can the sum be zero. 

 Now suppose that $a \ne 1$: if $i \ne 1$ and $j\ne a$ then $\gamma_{ij}^1(u_{i,j},v_{i,j})=0$. If $i=1$ and $j \ne a$ then $\gamma_{ij}^1(u_{i,j},v_{i,j})=0$.
 If $i=1$ and $j=a$ then $\gamma_{1a}^1(u_{1,a},v_{1,a}) =  \sum_{\ell = 0}^{k-3} g_{(+)}(v_{1,a})W^{\ell}$.
 Now consider the sum for all instances where $i\ne 1$ and $j=a$:
 $\sum_{i=2}^{a-1}\gamma_{ia}^a(u_{i,a},v_{i,a})  =\sum_{i \in [2,a-1]} g_{(-)}(v_{i,a}) \cdot W ^{i-2}$.
 Moreover, consider the sum for all instances where $i=a$: $\sum_{j=a+1}^{k}\gamma_{aj}^a(u_{a,j},v_{a,j})  =\sum_{j \in [a+1,k]} g_{(-)}(u_{a,j}) \cdot W ^{j-3}$.
 For clarity lets combine these:
 $$
 S_a:=\sum_{\ell = 0}^{k-3} g_{(+)}(v_{1,a})W^{\ell}+\sum_{\ell \in [0,a-3]} g_{(-)}(v_{\ell+2,a}) \cdot W ^{\ell} + \sum_{\ell \in [a-2,k-3]} g_{(-)}(u_{a,\ell+3}) \cdot W ^{\ell}.
 $$
 Note that if $v_{1,a} = v_{i,a} = u_{a,j}$ for all $i$ and $j$ then this does sum to zero. If there is some index, $\ell$, where $v_{1,a} \ne v_{\ell,a}$ or $v_{1,a} \ne u_{a,\ell}$ then, as before, the multiplication by $W^\ell$ will cause there to be no carries and the whole sum could not equal zero. The sum of all values multiplied by $W^{\ell'}$ where $\ell'<\ell$ will have absolute value less than $W^\ell$. All values multiplied by $W^{\ell'}$ where $\ell' > \ell$ or more will be zero mod $W^{\ell+1}$. So $S_a$
  % $$\sum_{\ell = 0}^{k-3} g_{(+)}(v_{1,a})W^{\ell} + \sum_{\ell \in [0,a-3]} g_{(-)}(v_{\ell,a}) \cdot W ^{\ell} + \sum_{\ell \in [a-2,k-3]} g_{(-)}(u_{a,\ell}) \cdot W ^{\ell}$$
 will not be zero mod $W^{\ell+1}$ if there is some index, $\ell$, where $v_{1,a} \ne v_{\ell,a}$ or $v_{1,a} \ne u_{a,\ell}$.

 So, we only get $S_a=0$ if each of the $k$ sections $\gamma^a_{ij}$ of the number sum to zero. Each section is only equal to zero if all \binomk  numbers agree on the identity of $a^{th}$ node. So, there is a zero sum iff there is a k-clique in the original graph.
 
\end{proof}

\paragraph{Putting it all together}

We will now state a theorem about the implications from the $k$-clique hypothesis (see Definition \ref{def:kcliqueHypothesis}).

% n^{\omega k /3} hardness for k-clique --> K=\binomk2, N=n^2
% N^{1/2(\omega/3)(\sqrt{2*K}+1)}

\cliqueToXORSUMOV*
\begin{proof}
    An instance of $k$-clique with $n$ nodes is $n^{\omega k/3 -o(1)}$ hard. We can make a $K$-XOR, $K$-OV, or $K$-SUM instance with $K =\bikmath$ and $N= n^2$. The $k$-clique problem (for $k=o(\lg(n))$) is equivalently hard in the worst-case and on Erd{\H o}sR\'enyi graphs \cite{UniformCliqueABB}.

    Using Theorem \ref{thm:cliqueToOVXOR} we can take an Erd{\H o}sR\'enyi graph and apply our reduction to get an explicit average-case distribution over $K$-OV, $K$-XOR, and $K$-SUM. 

    We let $n = N^{1/2}$.
    We also define $k(k-1) = 2K$, so $k \geq \sqrt{2K}-1$. 
    Now we can state the lower bound of $n^{\omega k/3 -o(1)}$  in terms of $K$ and $N$ as $N^{\omega(\sqrt{2K}+1)/6-o(1)}$. The rest is similar.
\end{proof}

%\section{From Average-Case Parity SAT to $\paritykOV{}$ with Large Vectors Directly}
%\label{sec:fromSATtoOV}
%\input{SATtoOV}

\section{Discussion and Future Work}
\label{sec:Discuss}
In Section \ref{sec:betweenProblem-bg} we discuss various reductions between factored problems. A sufficiently fast reductions and the framework in Section \ref{sec:WCACxor} would imply better lower bounds for $\oplus k$-XOR and $\oplus k$-SUM. For example, a $\oplus (OV,n,k,b,g) \rightarrow (XOR, poly(n), \Theta(k), \Theta(b), \Theta(g))$ reduction would imply an average-case lower bound for $\oplus k$-XOR of $n^{\Theta(k)}$ from SETH. 

In Theorem \ref{thm:BetterFrameworkLargeFeild} we pick each bit of $\vec{v}$ iid as $\{0,1\}$ each with probability $1/2$. However, we can instead have zero sampled with probability $\mu$ and one sampled with probability $1-\mu$.

One can go further with reductions. Notably, we don't need to start and end with factored problems. For example, if one can reduce from worst-case OV (not the factored version) to a small number of instances of the \ckxor[]~problem each with $g$ sets of size $b$ and $gb = O(k\lg(n))$ then this reduction would imply an average-case lower bound for $\oplus k$-XOR of $n^{\Theta(k)}$ from SETH. 

It would be interesting to show `Average-case Counting rETH' is implied by rETH. That is, find a distribution $D$ that can be efficiently sampled where counting $3$-SAT requires $2^{\Theta(n)}$. Note our current results have clause sizes of $\omega(1)$. This result should also imply a hardness of $n^{\Theta(k)}$ for counting $k$-SUM via Patrascu and Williams \cite{PatrascuW10}.

Finally, we would ideally like to give worst-case to average-case reduction from $\oplus k$-SUM (and OV and XOR) back to itself that are tight. Our current lower bounds are of the from $n^{\Omega(\sqrt{k})}$. However, we can hope for lower bounds of $n^{\Omega(k)}$ or $n^{k-o(1)}$ on some efficient to sample distribution. Notably, for both $k$-OV and $k$-SUM their uniform distributions are hypothesized to be hard on average. Could we perhaps find low degree polynomials for these problems? This approach can not work for $k$-OV if SETH is true \cite{factoredProblems}. So we can concentrate to low degree polynomials that rely on some structure of $k$-SUM and $k$-XOR that does not exist for $k$-OV.

% Decrease the space between bibliography items.
%\let\realbibitem=\bibitem
%\def\bibitem{\par \vspace{-0.5ex}\realbibitem}

\bibliographystyle{alpha}
\bibliography{my}

\end{document}